\documentclass[11pt]{article}

\usepackage{amsthm}
\usepackage{graphicx} 
\usepackage{array} 

\usepackage{amsmath, amssymb, amsfonts, verbatim}
\usepackage{hyphenat,epsfig,multirow}

\usepackage[usenames,dvipsnames]{xcolor}
\usepackage[ruled]{algorithm2e}

\usepackage{tcolorbox}
\tcbuselibrary{skins,breakable}
\tcbset{enhanced jigsaw}

\usepackage[compact]{titlesec}

\definecolor{DarkRed}{rgb}{0.5,0.1,0.1}
\definecolor{DarkBlue}{rgb}{0.1,0.1,0.5}

\colorlet{YellowOrange}{RawSienna}

\usepackage[linktocpage=true,
	pagebackref=true,colorlinks,
	linkcolor=DarkRed,citecolor=ForestGreen,
	bookmarks,bookmarksopen,bookmarksnumbered]
	{hyperref}

\usepackage{bm}
\usepackage{url}
\usepackage{xspace}
\usepackage[mathscr]{euscript}

\usepackage{mdframed}

\usepackage[noend]{algpseudocode}
\makeatletter
\def\BState{\State\hskip-\ALG@thistlm}
\makeatother

\usepackage{cite}
\usepackage{enumitem}

\usepackage[margin=1in]{geometry}

\newtheorem{theorem}{Theorem}
\newtheorem{lemma}{Lemma}[section]
\newtheorem{proposition}[lemma]{Proposition}
\newtheorem{corollary}[theorem]{Corollary}
\newtheorem{claim}[lemma]{Claim}
\newtheorem{fact}[lemma]{Fact}

\newtheorem{observation}[lemma]{Observation}
\newtheorem{definition}{Definition}

\newtheorem{problem}{Problem}

\newtheorem*{claim*}{Claim}
\newtheorem*{proposition*}{Proposition}
\newtheorem*{lemma*}{Lemma}
\newtheorem*{problem*}{Problem}

\newtheorem{mdresult}{Result}
\newtheorem{result}[mdresult]{Result}
\newenvironment{Result}{\begin{mdframed}[backgroundcolor=lightgray!40,topline=false,rightline=false,leftline=false,bottomline=false,innertopmargin=2pt,innerbottommargin=10pt]\begin{mdresult}}{\end{mdresult}\end{mdframed}}

\newtheorem{mdinvariant}{Invariant}

\theoremstyle{definition}
\newtheorem{remark}[lemma]{Remark}
\newtheorem{example}[lemma]{Example}

\allowdisplaybreaks

\newcommand{\Leq}[1]{\ensuremath{\underset{\textnormal{#1}}\leq}}
\newcommand{\Geq}[1]{\ensuremath{\underset{\textnormal{#1}}\geq}}
\newcommand{\Eq}[1]{\ensuremath{\underset{\textnormal{#1}}=}}

\usepackage{tikz}
\usetikzlibrary{arrows}
\usetikzlibrary{arrows.meta}
\usetikzlibrary{shapes}
\usetikzlibrary{backgrounds}
\usetikzlibrary{positioning}
\usetikzlibrary{decorations.markings}
\usetikzlibrary{patterns}
\usetikzlibrary{calc}
\usetikzlibrary{fit}
\usetikzlibrary{snakes}
\usetikzlibrary{shadows.blur}
\usetikzlibrary{petri,decorations.markings}

\usepackage{caption}
\usepackage{subcaption}

\renewcommand{\qed}{\nobreak \ifvmode \relax \else
      \ifdim\lastskip<1.5em \hskip-\lastskip
      \hskip1.5em plus0em minus0.5em \fi \nobreak
      \vrule height0.75em width0.5em depth0.25em\fi}

\newcommand{\Qed}[1]{\ensuremath{\qed_{\textnormal{~#1}}}}

\newcommand*\samethanks[1][\value{footnote}]{\footnotemark[#1]}

\newcommand{\toShrink}{-.20cm}
\newcommand{\toShrinkEnu}{-.2cm}

\newcommand{\algline}{
  \rule{0.5\linewidth}{.1pt}\hspace{\fill}%
  \par\nointerlineskip \vspace{.1pt}
}

\newcommand{\tvd}[2]{\ensuremath{\norm{#1 - #2}_{tvd}}}
\renewcommand{\tvd}[2]{\ensuremath{\Delta_{\textnormal{\texttt{TV}}}(#1,#2)}}
\newcommand{\Ot}{\ensuremath{\widetilde{O}}}
\newcommand{\eps}{\ensuremath{\varepsilon}}
\newcommand{\Paren}[1]{\Big(#1\Big)}
\newcommand{\Bracket}[1]{\Big[#1\Big]}
\newcommand{\bracket}[1]{\left[#1\right]}
\newcommand{\paren}[1]{\ensuremath{\left(#1\right)}\xspace}
\newcommand{\card}[1]{\left\vert{#1}\right\vert}
\newcommand{\Omgt}{\ensuremath{\widetilde{\Omega}}}

\newcommand{\set}[1]{\ensuremath{\left\{ #1 \right\}}}
\newcommand{\poly}{\mbox{\rm poly}}
\newcommand{\polylog}{\mbox{\rm  polylog}}

\newcommand{\alg}{\ensuremath{\mathcal{A}}\xspace}

\DeclareMathOperator*{\Exp}{\ensuremath{{\mathbb{E}}}}
\DeclareMathOperator*{\Prob}{\ensuremath{\textnormal{Pr}}}
\renewcommand{\Pr}{\Prob}

\newcommand{\Ex}{\Exp}

\newcommand{\event}[1]{\ensuremath{{\sf E}_{#1}}}

\newenvironment{tbox}{\begin{tcolorbox}[
		enlarge top by=5pt,
		enlarge bottom by=5pt,
		 breakable,
		 boxsep=0pt,
                  left=4pt,
                  right=4pt,
                  top=10pt,
                  boxrule=1pt,toprule=1pt,
                  colback=white,
                  arc=-1pt,
                  ]
	}
{\end{tcolorbox}}

\newcommand{\dist}{\ensuremath{\mathcal{D}}}

\newcommand{\Prot}{\ensuremath{\Pi}}
\newcommand{\prot}{\ensuremath{\pi}}

\newcommand{\bA}{\bm{A}}

\newcommand{\bR}{\bm{R}}

\newcommand{\bB}{\ensuremath{\bm{B}}}
\newcommand{\bC}{\ensuremath{\bm{C}}}
\newcommand{\bD}{\ensuremath{\bm{D}}}

\newcommand{\supp}[1]{\ensuremath{\textsc{supp}(#1)}}

\renewcommand{\event}{\mathcal{E}}

\title{Polynomial Pass Lower Bounds for Graph Streaming Algorithms}
\author{Sepehr Assadi\thanks{Department of Computer Science, Princeton University. Work done while the author was a graduate student at University of Pennsylvania and was supported in part
by the National Science Foundation grant CCF-1617851. Email: {{\small {\tt sassadi@princeton.edu.}}}}
\and Yu Chen\thanks{Department of Computer and Information Science, University of Pennsylvania. Supported in part by the National Science Foundation grants CCF-1617851 and CCF-1763514. Email: {{\small {\tt \{chenyu2,sanjeev\}@cis.upenn.edu.}}}}
\and Sanjeev Khanna\samethanks
}

\date{}

\begin{document}
\maketitle

\pagenumbering{roman}

\begin{abstract}
We present new lower bounds that show that a polynomial number of passes are necessary for solving some fundamental graph problems in the streaming model of computation.
For instance, we show that any streaming algorithm that finds a weighted minimum $s$-$t$ cut in an $n$-vertex undirected graph requires $n^{2-o(1)}$ space unless it makes $n^{\Omega(1)}$ passes over the stream. 

\smallskip

To prove our lower bounds, we introduce and analyze a new four-player communication problem that we refer to as the {\em hidden-pointer chasing} problem. This is a problem in spirit of the standard pointer chasing problem
with the key difference that the pointers in this problem are hidden to players and finding each one of them requires solving another communication problem, namely the set intersection problem. 
Our lower bounds for graph problems are then obtained by reductions from the hidden-pointer chasing problem. 

\smallskip

Our hidden-pointer chasing problem appears flexible enough to find other applications and is therefore interesting in its own right. To showcase this, we further
present an interesting application of this problem beyond streaming algorithms. Using a reduction from hidden-pointer chasing, we prove that any algorithm for submodular function minimization needs to make $n^{2-o(1)}$ value queries to the function unless it has a polynomial degree of adaptivity.

\end{abstract}

\clearpage

\setcounter{tocdepth}{3}
\tableofcontents

\clearpage

\pagenumbering{arabic}
\setcounter{page}{1}

\section{Introduction}\label{sec:intro}

Graph streaming algorithms are algorithms that solve computational problems on graphs, say, finding a maximum matching, when the input is
 presented as a sequence of edges, under the usual constraints of the streaming model, namely sequential access to the stream and limited memory. 
Formally, in the graph streaming model, the edges of a graph $G(V,E)$ are presented one by one in an arbitrary order. The algorithm can make one or a limited number of sequential passes over this stream, while using a limited memory to process the graph, 
preferably $O(n \cdot \polylog{(n)})$ memory, referred to as \emph{semi-streaming} restriction~\cite{FeigenbaumKMSZ05}; here $n$ is the number of vertices in $G$.

It turns out allowing for multiple passes over the stream greatly enhances the capability of graph streaming algorithms. A striking example is the (global) minimum cut problem: While $\Omega(n^2)$ space is needed for computing an exact minimum cut in a single pass~\cite{Zelke11}, 
a recent result of~\cite{RubinsteinSW18} implies that a minimum cut of an undirected unweighted graph can be computed in $\Ot(n)$ space in only
 two passes over the stream\footnote{\label{foot:imply} The result of~\cite{RubinsteinSW18} is not stated as a streaming algorithm. However, the algorithm in~\cite{RubinsteinSW18} combined with the known 
graph streaming algorithms for cut sparsifiers (see, e.g.~\cite{McGregor14}) immediately imply the claimed result.}. Table~\ref{tab:single-multi} presents 
several other examples of this phenomenon.
 
 \def\arraystretch{1.5}

\newsavebox{\tabone}
\sbox{\tabone}{
             {\small
        \centering
        \begin{tabular}{|l|c|c|c|c|c|c|c|}
            \cline{1-8}
        
          \multicolumn{1}{|c|}{\multirow{2}{*}{\textbf{Problem}}} &\multicolumn{4}{c|}{\textbf{Multi-Pass}} & \multicolumn{3}{c|}{\textbf{Single-Pass}}\\
            \cline{2-8}
             & \textbf{Space} & \textbf{Apx} & \textbf{Passes} & \textbf{Ref} & \textbf{Space} & \textbf{Apx} & \textbf{Ref} \\
            \cline{1-8}
     	     {Unweighted Min-Cut}  & $\Ot(n)$ & 1 & $2$ & \cite{RubinsteinSW18}  & $\Omega(n^2)$ & 1 & \cite{Zelke11}\\
            \cline{1-8}
     	     {Unweighted $s$-$t$ Min-Cut}  & $\Ot(n^{5/3})$ & 1 & $2$ & \cite{RubinsteinSW18}  & $\Omega(n^2)$ & 1 & \cite{Zelke11}\\
            \cline{1-8}
	     {Triangle Counting}  & $\Ot(\frac{m^{3/2}}{T})$ & $1+\eps$ & $4$ & \cite{BeraC17}  & $\Omega(\frac{m^3}{T^2})$ & $\Theta(1)$ & \cite{KutzkovP14a}\\
\hline
            {Maximum Matching}  & $\Ot(n)$ & $1+\eps$ & $O(1)$ & \cite{McGregor05}  & $n^{1+\Omega(\frac{1}{\log\log{n}})}$ & $\frac{e}{e-1}$ & \cite{Kapralov13}\\
	\hline
             {Single Source Shortest Path}  & $\Ot(n)$ & $1+\eps$ & $O(1)$ & \cite{BeckerKKL17}  & $\Omega(n^2)$ & $\frac{5}{3}$ & \cite{FeigenbaumKMSZ08}\\
	\hline
	     {Maximal Independent Set}  & $\Ot(n)$ & $-$ & $O(\log\log{n})$ & \cite{GhaffariGKMR18}  & $\Omega(n^2)$ & $-$ & \cite{AssadiCK19}\\
	\hline
	     {Minimum Dominating Set}  & $\Ot(n)$ & $O(\log{n})$ & $O(\log{n})$ & \cite{Har-PeledIMV16}  & $n^{2-o(1)}$ & $n^{o(1)}$ & \cite{AssadiKL16}\\
       \hline
        \end{tabular}
      }
  }
  
 \begin{table}[h!]
\begin{tikzpicture}
   \node[fill=white](boz){};
  \node[drop shadow={black, shadow xshift=5pt,shadow yshift=-5pt, opacity=0.5}, fill=white, inner xsep=-7pt, inner ysep=0pt](table)[right=0pt of boz]{\usebox{\tabone}};
\end{tikzpicture}
          \caption{ A sample of multi-pass graph streaming algorithms and corresponding single-pass lower bounds. 
          All results are for graphs $G(V,E)$ with $n$ vertices and $m$ edges (and $T$ triangles).
        \label{tab:single-multi}}

    \end{table}

Multi-pass graph streaming algorithms have been gaining increasing attention in recent years and for many well-studied graph problems, space efficient algorithms have been designed that use at most a logarithmic number of passes (see, 
e.g.~\cite{FeigenbaumKMSZ05,McGregor05,EggertKS09,KonradMM12,Kapralov13,KaleT17, McGregorVV16, HenzingerKN16, BeckerKKL17, AhnGM12a,AhnGM12b,SarmaGP11,GuruswamiO13,KapralovW14,ChakrabartiW16,Har-PeledIMV16,BeraC17}). 
But for many other problems, such results have proved elusive. Examples 
include shortest path and diameter computation~\cite{DistanceOP}, random walks~\cite{RandomWalkOP}, and directed reachability and maximum flow~\cite{McGregor14} (see also~\cite{StrongOP}). At the same time, known techniques for proving streaming lower bounds are unable to prove essentially any lower bounds beyond logarithmic number of passes (but see Section~\ref{sec:landscape} for an exception to this rule and the inherent limitation behind it). 
For example, the best known lower bounds for several key problems such as shortest path, directed reachability, and perfect matchings, only imply $\Omega(\frac{\log{n}}{\log\log{n}})$ passes for semi-streaming algorithms~\cite{FeigenbaumKMSZ08,GuruswamiO13}, while none of these problems currently admit an algorithm with 
$n^{2-\Omega(1)}$ space and $n^{o(1)}$ passes. 

Our goal in this paper is to remedy this situation by \textbf{{presenting new tools for proving stronger multi-pass graph streaming lower bounds}}. 
To better understand the challenges along the way, we first briefly revisit the current state-of-affairs.   

\subsection{Landscape of Graph Streaming Lower Bounds} \label{sec:landscape}

A vast body of work in graph streaming lower bounds concerns algorithms that make only one or a few passes over the stream. Examples
of single-pass lower bounds include the ones for diameter~\cite{FeigenbaumKMSZ08}, approximate matchings~\cite{GoelKK12,Kapralov13,AssadiKLY16,AssadiKL17}, exact minimum/maximum
cuts~\cite{Zelke11}, and maximal independent sets~\cite{AssadiCK19,CormodeDK18}. Examples of multi-pass lower bounds
include the ones for BFS trees~\cite{FeigenbaumKMSZ08}, perfect matchings~\cite{GuruswamiO13}, shortest path~\cite{GuruswamiO13}, 
and minimum vertex cover and dominating set~\cite{Har-PeledIMV16}.
These lower bounds are almost always obtained by considering {communication complexity} of the problem with \emph{limited number of rounds} of communication which gives
a lower bound on the space complexity of streaming algorithms with proportional number of passes to the limits on rounds of communication (see e.g.~\cite{AlonMS96,GuhaM08}). 
The communication lower bounds are then typically proved via reductions from (variants of) the \emph{pointer chasing} problem~\cite{PapadimitriouS84,NisanW91,ChakrabartiCM08} for multi-pass lower bounds and
the \emph{indexing} problem~\cite{Ablayev93,KremerNR95} and \emph{boolean hidden (hyper-)matching} problem~\cite{GavinskyKKRW07,VerbinY11} for single-pass lower bounds. 

In the pointer chasing problem, Alice and Bob are given functions $f,g :[n] \rightarrow [n]$ and the goal is to compute $f(g(\cdots f(g(0))))$ for $k$ iterations.
Computing this function in less than $k$ rounds requires $\Omgt(n/k)$ communication~\cite{Yehudayoff16} (see also~\cite{DurisGS84,PapadimitriouS84,NisanW91,PonzioRV99}).
The reductions from pointer chasing to graph streaming lower bounds are based on using vertices of the graph to encode $[n]$ and each edge to encode a pointer~\cite{FeigenbaumKMSZ08,GuruswamiO13}. Directly using pointer chasing
does not imply lower bounds stronger than $\Omega(n)$ and hence variants of pointer chasing with multiple pointers such as multi-valued pointer chasing~\cite{FeigenbaumKMSZ08,JainRS03} and set pointer chasing~\cite{GuruswamiO13}, were considered.  Using multiple pointers
however has the undesired side effect that the lower bound deteriorates exponentially with number of rounds. 
As such, these lower bounds do not go beyond $O(\log{n})$ passes even for algorithms with $O(n)$ space.

There are however a number of results that prove 
lower bounds for a very large number of passes (even close to $n$). Examples include lower bounds for approximating clique and independent set~\cite{HalldorssonSSW12}, 
approximating dominating set~\cite{Assadi17sc}, computing girth~\cite{FeigenbaumKMSZ08}, estimating the number of triangles~\cite{Bar-YossefKS02,JowhariG05,CormodeJ17,BeraC17}, and finding minimum vertex cover or coloring~\cite{AbboudCKP19}.  These results are all proven by considering the communication complexity of the 
problem with \emph{no limits on rounds} of communication. Such bounds
then imply lower bounds on the {product} of space and number of passes of streaming algorithms (see, e.g.~\cite{AlonMS96}). The communication lower bounds themselves
are proven by reductions from a handful of communication problems, mainly the \emph{set disjointness} problem~\cite{BabaiFS86,KalyanasundaramS92,Razborov92,Bar-YossefJKS02}.

This approach suffers from two main drawbacks. Firstly, these lower bounds only exhibit space bounds that scale with the reciprocal of the number of passes and 
are hence unable to capture more nuanced space/pass trade-offs. More importantly, there is an inherit limitation to this approach since the computational model considered here is much stronger than the streaming model.
This means that many problems of interest admit efficient communication protocols in this model and hence one simply cannot prove interesting lower bounds for them. 
An illustrating example is the directed $s$-$t$ reachability problem which admits an $O(n)$ communication protocol, 
ruling out the possibility of essentially any non-trivial lower bound using this approach (even ``harder'' problems such as maximum matching admit non-trivial protocols with $\Ot(n^{3/2})$ communication~\cite{IvanyosKLSW12,DobzinskiNO14}).

\subsection{Our Contributions}\label{sec:result}

We introduce and analyze a new communication problem similar in spirit to standard pointer chasing, which we refer to as the \emph{hidden-pointer chasing} (HPC) problem. 
What differentiate HPC from previous variants 
of pointer chasing is that the pointers are ``hidden'' from players and finding each one of them requires solving another communication problem, namely the \emph{set intersection} problem, in which
the goal is to \emph{find} the \emph{unique} element in the intersection of players input. 
We limit ourselves to the following informal definition of HPC here and postpone the formal definition to Section~\ref{sec:tech-hpc}. 
There are four players in HPC paired into groups of size two each. Each pair of players inside a group shares $n$ instances of the {set intersection} problem on $n$ elements. 
The intersecting element in each instance of each group ``points'' to an instance in the other group. 
The goal is to start from a fixed instance and follow these pointers for a fixed number of steps. 
We prove the following communication complexity lower bound for HPC.

\begin{Result}\label{res:hpc}
	Any $r$-round protocol that with constant probability finds the $(r+1)$-th pointer in the hidden-pointer chasing problem requires $\Omega({n^2}/{r^2})$ communication. 
\end{Result}

Result~\ref{res:hpc} implies a new approach towards proving graph streaming lower bounds that sits squarely in the middle of previous methods: 
HPC is a problem that admits an ``efficient'' protocol when there is no limit on rounds of communication and yet is ``hard'' with even a polynomial limitation on number of rounds. 
We use this result to prove strong pass lower bounds for some fundamental problems in graph streams via  reductions from HPC. 

\paragraph{Cut and Flow Problems.} One of the main applications of Result~\ref{res:hpc} is the following result. 

\begin{result}\label{res:cut}
	Any $p$-pass streaming algorithm that with a constant probability outputs the minimum $s$-$t$ cut value in a weighted graph (undirected or directed) requires 
	$\Omega({n^2}/{p^5})$ space. 
\end{result}

Prior to our work, the best lower bound known for this problem was an $n^{1+\Omega(1/p)}$ space lower bound for $p$-pass algorithms~\cite{GuruswamiO13} (for weighted undirected graphs and unweighted directed graphs).
Result~\ref{res:cut} significantly improves upon this. In particular, it implies that $\Omgt(n^{1/5})$ passes are necessary for semi-streaming algorithms, exponentially improving 
upon the $\Omega(\frac{\log{n}}{\log\log{n}})$ lower bound of~\cite{GuruswamiO13}. At the same time, Result~\ref{res:cut} also shows that any streaming algorithm for this problem with a small number of passes, namely $\polylog{(n)}$ passes, 
requires $\Omgt(n^2)$ space, almost the same space as the trivial single-pass algorithm that stores the input graph entirely. 

Our Result~\ref{res:cut} should be contrasted with the results of~\cite{RubinsteinSW18} that imply an $\Ot(n^{5/3})$ space
algorithm for unweighted minimum $s$-$t$ cut on undirected graphs in only \emph{two} passes (see Footnote~\ref{foot:imply}). 

By max-flow min-cut theorem, Result~\ref{res:cut} also implies identical bounds for computing the value of maximum $s$-$t$ flow
in capacitated graphs, making progress on a question raised in~\cite{McGregor14} regarding the streaming complexity of maximum flow in directed graphs.

\paragraph{Lexicographically-First Maximal Independent Set.} A maximal independent set (MIS) returned by the sequential greedy algorithm that visits the vertices of the graph in their lexicographical order 
is called the lexicographically-first MIS. We prove the following result for this problem. 

\begin{result}\label{res:mis}
	Any $p$-pass streaming algorithm that with constant probability finds a lexicographically first maximal independent set of in a graph requires 
	$\Omega({n^2}/{p^5})$ space. 
\end{result}

The lexicographically-first MIS has a rich history in computer science and in particular parallel algorithms~\cite{Cook85,AlonBI86,Luby86,BlellochFS12}. However, even though multiple variants of the independent set problem have been studied in the
streaming model~\cite{HalldorssonHLS10,HalldorssonSSW12,HalldorssonHLS16,CormodeDK18a,CormodeDK18,AssadiCK19,GhaffariGKMR18}, we are not aware of any work on this particular problem (we remark that standard MIS problem
admits an $\Ot(n)$ space $O(\log\log{n})$ pass algorithm~\cite{GhaffariGKMR18}).
 Besides being a fundamental problem in its own right, what makes this problem appealing for us is that it nicely illustrates the power of our techniques compared to previous approaches. 
The lexicographically-first MIS can be computed with $O(n)$ communication in the two-player 
communication model (or for any constant number of players) with no restriction on number of rounds by a direct simulation of the sequential algorithm. Hence, this problem perfectly fits the class of problems for which previous techniques 
cannot prove lower bounds beyond logarithmic passes. To our knowledge, this is the first super-logarithmic pass lower bound for any graph problem that admits an efficient protocol with no restriction on number of
rounds.

\subsubsection*{Beyond Graph Streams: An Application to Submodular Minimization}  

We also use Result~\ref{res:hpc} to prove query/adaptivity tradeoffs for the submodular function minimization (SFM) problem.
 In SFM, we have a submodular function $f: 2^{[n]} \rightarrow [M]$ and our goal is to find 
a set $S^* \subseteq [n]$ that minimizes $f(S^*)$ by making value queries to $f$. SFM has been studied extensively over the 
years~\cite{GrotschelLS81,Cunningham85,IwataFF00,IwataO09,Schrijver00,LeeSW15,ChakrabartyLSW17}, culminating in the currently best algorithms of~\cite{LeeSW15} and~\cite{ChakrabartyLSW17} with $\Ot(n^2)$ 
and $\Ot(n \cdot M^3)$ queries, respectively. The best 
lower bound for SFM is $\Omega(n)$ queries~\cite{Harvey08,Harvey08t} and determining the query complexity of this problem remains a fascinating open question~\cite{Harvey08,RubinsteinSW18}.

Another question in this area that has received a significant attention in recent years is to understand the query/adaptivity tradeoffs in submodular 
optimization~\cite{BalkanskiRS16,BalkanskiRS17,BalkanskiS17,BalkanskiS18,BalkansiRS18,BalkanskiBS18,FahrbachMZ18,FahrbachMZ18a,EneN18,EneNV18}. 
An algorithm for SFM is called \emph{$k$-adaptive} iff it makes at most $k$ rounds of adaptive queries, where the queries in each round are performed in parallel.
We prove that any $k$-round adaptive algorithm for SFM requires $\Omgt({n^{2}}/{k^5})$ queries (see Theorem~\ref{thm:sfm}). This in particular implies that if there is an algorithm with truly sub-quadratic query complexity, then it must 
have a polynomial degree of adaptivity. The only other adaptivity lower bound for SFM that we are aware of is an exponential lower bound on query complexity of 
\emph{non-adaptive} algorithms (even for approximation)~\cite{BalkanskiS17}. However, once we allow even two rounds of adaptivity, no lower bounds better than $\Omega(n)$ queries were known.

\subsection{Our Techniques}\label{sec:techniques}

Our reductions in this paper take a different path than previous pointer chasing based reductions that used edges of the graph to directly encode pointers. 
In particular, our hidden-pointer chasing problem allows us encode a single pointer among $\Theta(n)$ edges and thus work with graphs with density $\Omega(n^2)$ and still keep a polynomial dependence on number of rounds in the communication lower bound. This results
in space lower bounds of the form $n^2/p^{O(1)}$ for $p$-pass streaming algorithms. 

The main technical contribution of our paper is the communication complexity lower bound for HPC in Result~\ref{res:hpc}. This result is proved by combining inductive arguments 
for round/communication tradeoffs (see, e.g.~\cite{NisanW91,Yehudayoff16}) with direct-sum arguments for information complexity (see, e.g.~\cite{Bar-YossefJKS02,BarakBCR10,BravermanR11,BravermanEOPV13}) to account for the 
role of set intersection inside HPC. To make this argument work, we also need to prove a stronger 
lower bound for set intersection than currently known results (see, e.g.~\cite{BrodyCKWY14}). In particular, 
we prove that any protocol that can even slightly reduce the ``uncertainty'' about the intersecting element must have a ``large'' communication and information complexity. 

Our new lower bound for set intersection is also proved using tools from information complexity to reduce this problem to a primitive problem, namely set intersection itself on a universe of size two. 
This requires a novel argument to handle the protocols for set intersection that reduce the uncertainty about the intersecting element without necessarily making much ``progress'' on finding this element.
Another challenge is that unlike typical direct-sum results in this context, say reducing disjointness to the AND problem; see, e.g.~\cite{Bar-YossefJKS02,BravermanM13,BravermanGPW13,WeinsteinW15},
set intersection cannot be decomposed into \emph{independent} instances of the primitive problem (this is similar-in-spirit to challenges in analyzing information complexity of set disjointness on \emph{intersecting} distributions~\cite{JayramKS03,ChattopadhyayM15} 
as opposed to (more standard) non-intersecting ones). Finally, we prove a lower bound for the primitive problem using the product structure of Hellinger distance for communication protocols (see, e.g.~\cite{Bar-YossefJKS02,WeinsteinW15}).

\subsubsection*{Organization}

The rest of the paper is organized as follows. We set up our notation in Section~\ref{sec:prelim}. Section~\ref{sec:technical} contains a detailed technical overview of our approach. 
We present the proof of our new communication lower bound
for set intersection that is needed for establishing Result~\ref{res:hpc} in Section~\ref{sec:SI}. 
Section~\ref{sec:hpc} then uses this to finalize the proof of Result~\ref{res:hpc}. We present our lower bounds for graph streaming algorithms and for submodular minimization in Sections~\ref{sec:lower} and~\ref{sec:SFM}, respectively. 
 Appendix~\ref{app:related} presents further discussion on related work and Appendix~\ref{app:prelim} contains the backgrounds and preliminaries.

\newcommand{\kl}[2]{\ensuremath{\DD(#1~||~#2)}}
\newcommand{\rA}{\rv{A}}
\newcommand{\rB}{\rv{B}}
\newcommand{\rC}{\rv{C}}
\newcommand{\rD}{\rv{D}}

\newcommand{\rX}{\rv{X}}
\newcommand{\rY}{\rv{Y}}
\newcommand{\rI}{\rv{I}}

\newcommand{\rv}[1]{\ensuremath{\mathsf{#1}}}

\newcommand{\II}{\ensuremath{\mathbb{I}}}
\newcommand{\HH}{\ensuremath{\mathbb{H}}}
\newcommand{\mi}[2]{\ensuremath{\II(#1 \,; #2)}}
\newcommand{\en}[1]{\ensuremath{\HH(#1)}}

\newcommand{\itfacts}[1]{Fact~\ref{fact:it-facts}-(\ref{part:#1})\xspace}

\newcommand{\DD}{\ensuremath{\mathbb{D}}}

\newcommand{\distribution}[1]{\ensuremath{\textnormal{dist}(#1)}\xspace}

\newcommand{\XX}{\ensuremath{\mathcal{X}}}
\newcommand{\YY}{\ensuremath{\mathcal{Y}}}

\newcommand{\HPC}{\ensuremath{\textnormal{\textsf{HPC}}}\xspace}
\newcommand{\HPCk}{\ensuremath{\HPC_k}\xspace}

\newcommand{\protHPC}{\ensuremath{\prot_{\HPC}}}

\newcommand{\distHPC}{\ensuremath{\dist_{\HPC}}}

\newcommand{\unif}{\ensuremath{\mathcal{U}}}

\newcommand{\rZ}{\rv{Z}}
\newcommand{\rt}{\rv{t}}
\newcommand{\rT}{\rv{T}}
\newcommand{\rE}{\rv{E}}
\newcommand{\rP}{\rv{P}}
\newcommand{\rProt}{\rv{\Prot}}

\newcommand{\distSI}{\ensuremath{\dist}_{\ensuremath{\textnormal{\textsf{SI}}}}\xspace}

\newcommand{\protSI}{\ensuremath{\prot}_{\ensuremath{\textnormal{\textsf{SI}}}}\xspace}
\newcommand{\ProtSI}{\ensuremath{\Prot}_{\ensuremath{\textnormal{\textsf{SI}}}}\xspace}
\newcommand{\rProtSI}{\ensuremath{\rv{\Prot}}_{\ensuremath{\textnormal{\textsf{SI}}}}\xspace}

\newcommand{\distPI}{\ensuremath{\dist}_{\ensuremath{\textnormal{\textsf{PI}}}}\xspace}
\newcommand{\protPI}{\ensuremath{\prot}_{\ensuremath{\textnormal{\textsf{PI}}}}\xspace}
\newcommand{\ProtPI}{\ensuremath{\Prot}_{\ensuremath{\textnormal{\textsf{PI}}}}\xspace}
\newcommand{\rProtPI}{\ensuremath{\rv{\Prot}}_{\ensuremath{\textnormal{\textsf{PI}}}}\xspace}

\newcommand{\SI}{\ensuremath{\textnormal{\textsf{Set-Int}}}\xspace}
\newcommand{\PI}{\ensuremath{\textnormal{\textsf{Pair-Int}}}\xspace}

\newcommand{\errs}{\ensuremath{\textnormal{errs}}\xspace}

\newcommand{\rK}{\rv{K}}
\newcommand{\rJ}{\rv{J}}
\newcommand{\rL}{\rv{L}}
\newcommand{\rS}{\rv{S}}
\newcommand{\rM}{\rv{M}}
\newcommand{\rR}{\rv{R}}
\newcommand{\rbarS}{\rv{\barS}}

\newcommand{\distDisjN}{\ensuremath{\distDisj^{\textnormal{\textsf{N}}}}\xspace}
\newcommand{\distDisjY}{\ensuremath{\distDisj^{\textnormal{\textsf{Y}}}}\xspace}

\newcommand{\mii}[3]{\ensuremath{\II_{#3}(#1 \,; #2)}}

\newcommand{\disjin}[4]{ \ensuremath{\begin{bmatrix} #1&\hspace{-8pt}#2 \\ #3&\hspace{-10pt}#4 \end{bmatrix}}}
\renewcommand{\disjin}[4]{\textnormal{[\ensuremath{#1#2,#3#4}]}}

\newcommand{\barS}{\overline{S}}

\newcommand{\hd}[2]{\ensuremath{\textnormal{h}(#1,#2)}\xspace}
\newcommand{\hdt}[2]{\ensuremath{\textnormal{h}^2(#1,#2)}\xspace}

\newcommand{\CC}[2]{\ensuremath{\textnormal{\textsf{CC}}_{#2}(#1)}\xspace}
\newcommand{\ICost}[2]{\ensuremath{\textnormal{\textsf{IC}}_{#2}(#1)}\xspace}
\newcommand{\IC}[2]{\ICost{#1}{#2}}

\newcommand{\rTheta}{\rv{\Theta}}

\newcommand{\comp}{\succ_{\ProtSI}}
\newcommand{\compr}{\prec_{\ProtSI}}

\section{Preliminaries}\label{sec:prelim}

\paragraph{Notation.} For any integer $a $, we define $[a] := \set{1,\ldots,a}$. For a tuple $(X_1,\ldots,X_n)$ and integer $i \in [n]$, $X^{<i} := (X_1,\ldots,X_{i-1})$
and $X_{-i} := (X_1,\ldots,X_{i-1},X_{i+1},\ldots,X_n)$. We use capital `san-serif' font to denote the random variables, e.g. $\rX$. $\unif_S$ denotes the uniform distribution over $S$. 

For random variables $\rX,\rY$, $\en{\rX}$ denotes the Shannon entropy of $\rX$ and $\mi{\rX}{\rY}$ denotes the mutual information. 
For distributions $\mu,\nu$, $\kl{\mu}{\nu}$ denotes the KL-divergence, $\tvd{\mu}{\nu}$ denotes the total variation distance, and $\hd{\mu}{\nu}$ denotes the Hellinger distance. 
Necessary background on information theory, including the definitions and basic tools, is provided in Appendix~\ref{sec:info}. 

\paragraph{Communication Complexity and Information Complexity.} We consider the standard communication model of Yao~\cite{Yao79}. 
We use $\prot$ to denote the protocol used by players and use $\CC{\prot}{}$ to denote the \emph{communication cost} of $\prot$ defined as the worst-case bit-length of the messages communicated between the players. 
We further use \emph{internal information cost}~\cite{BarakBCR10} for protocols that measures the average amount of information each player learns about the input of the other in the protocol, defined formally as follows. 
Consider an input distribution $\dist$ and a protocol $\prot$. Let $(\rX,\rY) \sim \dist$ and $\rProt$ denote the random variables for the inputs and 
the transcript of the protocol (including the public randomness). 
The \emph{information cost} of $\prot$ with respect to
$\dist$ is $\ICost{\prot}{\dist}:=\mii{\rProt}{\rX \mid \rY}{\dist} + \mii{\rProt}{\rY \mid \rX}{\dist}$. As one bit of communication can only reveal one bit of information, 
information cost of a protocol lower bounds its communication cost (see Proposition~\ref{prop:cc-ic}).

Appendix~\ref{sec:cc-ic} contains the relevant background and definitions on communication complexity and information complexity that we use in this paper. 

\paragraph{Set Intersection Problem.} We use the set intersection problem in construction of our HPC problem. 
Set intersection (\SI) is a two-player communication problem in which Alice and Bob are given sets $A$ and $B$ from $[n]$, respectively, with the promise that there exists a unique
element $t$ such that $\set{t} = A \cap B$. The goal is for players to find the \emph{target element} $t$. 
An $\Omega(n)$ communication lower bound for $\SI$ follows directly from lower bounds for set disjointness~\cite{KalyanasundaramS92,Razborov92,Bar-YossefJKS02,BravermanM13,BravermanGPW13}; see, e.g.~\cite{BrodyCKWY14} 
(this lower bound by itself is however not useful for our application).

\section{Technical Overview}\label{sec:technical}

We start with defining the hidden-pointer chasing (HPC) problem and briefly discuss a reduction from HPC
that establishes the lower bound for minimum cut problem in Result~\ref{res:cut}. We then sketch the proof of the communication lower bound for HPC in Result~\ref{res:hpc}. 
Along the way, we also present a new lower bound for set intersection that is needed for establishing Result~\ref{res:hpc}. 
We emphasize that this section oversimplifies many details and the discussions will be informal for the sake of intuition. 

\subsection{The Hidden-Pointer Chasing Problem}\label{sec:tech-hpc}

The {hidden-pointer chasing} (\HPC) problem is a {four-party} communication problem with players $P_A,P_B,P_C$, and $P_D$. Let $\XX := \set{x_1,\ldots,x_n}$ and $\YY := \set{y_1,\ldots,y_n}$ be two disjoint universes. 

\begin{enumerate}[leftmargin=25pt]
	\item For any $x \in \XX$, $P_A$ and $P_B$ are given an instance $(A_x,B_x)$ of $\SI$ over the universe $\YY$ where $A_x \cap B_x = \set{t_x}$ for $t_x \in \YY$. 
	 
	 \item Similarly, for any $y \in \YY$, $P_C$ and $P_D$ are given an instance $(C_y,D_y)$ of $\SI$ over the universe $\XX$ where $C_y \cap D_y = \set{t_y}$ for $t_y \in \XX$.  
	
	\item We define two mappings $f_{AB} : \XX \rightarrow \YY$ and $f_{CD}: \YY \rightarrow \XX$ such that:
	\begin{enumerate}
		\item for any $x \in \XX$, $f_{AB}(x) = t_x \in \YY$ in the instance $(A_x,B_x)$ of $\SI$. 
		\item for any $y \in \YY$, $f_{CD}({y}) = t_y \in \XX$ in the instance $(C_y,D_y)$ of $\SI$.
	\end{enumerate}
	
	\item Let $x_1 \in \XX$ be an arbitrary fixed element of $\XX$ known to all players. The pointers $z_0,z_1,z_2,z_3,\ldots$ are defined inductively as follows: 
	$
		z_0 := x_1,  z_1 := f_{AB}(z_0),  z_2:= f_{CD}(z_1),  z_3:= f_{AB}(z_2),  \cdots .
	$
\end{enumerate}

The \emph{$k$-step hidden-pointer chasing} problem ($\HPCk$) is defined as the communication problem of finding the pointer $z_k$. See Figure~\ref{fig:hpc} for an illustration. 

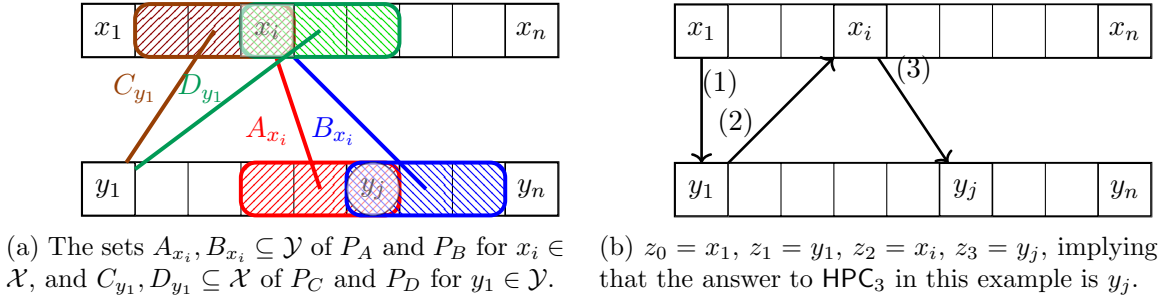
\begin{figure}[h!]
    \centering
    \subcaptionbox{The sets $A_{x_i},B_{x_i} \subseteq \YY$ of $P_A$ and $P_B$ for $x_i \in \XX$, and $C_{y_1},D_{y_1} \subseteq \XX$ of $P_C$ and $P_D$ for $y_1 \in \YY$.}[0.45\textwidth]{

\begin{tikzpicture}[ auto ,node distance =1cm and 2cm , on grid , semithick , state/.style ={ circle ,top color =white , bottom color = white , draw, black , text=black}, every node/.style={inner sep=0,outer sep=0}]

\node[state,rectangle, top color=white, black, bottom color=white, text=black, opacity=0, text opacity=1,  align=center]  (XX){$\XX:$};
\node[state,rectangle, top color=white, black, bottom color=white, text=black, minimum height=20pt, minimum width = 180pt,  align=center]  (X)[right=105pt of XX]{\tikz{\draw[step=20pt] (0,0)  grid (180pt,20pt);}};
\node[state,rectangle, top color=white, black, bottom color=white, text=black, opacity=1, text opacity=1, minimum height=20pt, minimum width = 20pt]  (X1) [left=80pt of X]{$x_1$};
\node[state,rectangle, top color=white, black, bottom color=white, text=black, opacity=1, text opacity=1, minimum height=20pt, minimum width = 20pt]  (Xi) [right=60pt of X1]{$x_i$};
\node[state,rectangle, top color=white, black, bottom color=white, text=black, opacity=1, text opacity=1, minimum height=20pt, minimum width = 20pt]  (Xn) [right=100pt of Xi]{$x_n$};

\node[state,rectangle, top color=white, black, bottom color=white, text=black, opacity=0, text opacity=1,  align=center]  (YY)[below=60pt of XX]{$\YY:$};
\node[state,rectangle, top color=white, black, bottom color=white, text=black, minimum height=20pt, minimum width = 180pt,  align=center]  (Y)[below=60pt of X] {\tikz{\draw[step=20pt] (0,0)  grid (180pt,20pt);}};
\node[state,rectangle, top color=white, black, bottom color=white, text=black, opacity=1, text opacity=1, minimum height=20pt, minimum width = 20pt]  (Y1) [left=80pt of Y]{$y_1$};
\node[state,rectangle, top color=white, black, bottom color=white, text=black, opacity=1, text opacity=1, minimum height=20pt, minimum width = 20pt]  (Yj) [right=100pt of Y1]{$y_j$};
\node[state,rectangle, top color=white, black, bottom color=white, text=black, opacity=1, text opacity=1, minimum height=20pt, minimum width = 20pt]  (Yn) [right=60pt of Yj]{$y_n$};

\node[rectangle, rounded corners=2mm, red, draw, minimum height=20pt, minimum width = 60pt, line width=1.5pt, pattern=north east lines, pattern color=red]  (Axi)[left=0pt of Y] {};
\node[rectangle, rounded corners=2mm, blue, draw, minimum height=20pt, minimum width = 60pt, line width=1.5pt, pattern=north west lines, pattern color=blue]  (Bxi)[right=40pt of Y] {};
\node[state,circle, top color=white, black, bottom color=white, text=black, opacity=0.4, text opacity=1, minimum height=20pt, minimum width = 20pt]  (Yj) [right=100pt of Y1]{$y_j$};


\draw[-,line width=1.5pt, red, text=red] (Xi) -- (Axi.center) node[below left=38pt and 0pt of Xi]{$A_{x_i}$}; 

\draw[-,line width=1.5pt, blue, text=blue] (Xi) -- (Bxi.center) node[below right=38pt and 25pt of Xi]{$B_{x_i}$};

\node[rectangle, rounded corners=2mm, YellowOrange, draw, minimum height=20pt, minimum width = 60pt, line width=1.5pt, pattern=north east lines, pattern color=YellowOrange]  (Cy1)[left=40pt of X] {};
\node[rectangle, rounded corners=2mm, ForestGreen, draw, minimum height=20pt, minimum width = 60pt, line width=1.5pt, pattern=north west lines, pattern color=ForestGreen]  (Dy1)[right=0pt of X] {};
\node[state,rectangle, top color=white, black, bottom color=white, text=black, opacity=0.4, text opacity=1, minimum height=20pt, minimum width = 20pt]  (Xi) [right=60pt of X1]{$x_i$};


\draw[-,line width=1.5pt, YellowOrange, text=YellowOrange] (Y1) -- (Cy1.center) node[above right=38pt and 10pt of Y1]{$C_{y_1}$}; 

\draw[-,line width=1.5pt, ForestGreen, text=ForestGreen] (Y1) -- (Dy1.center) node[above right=38pt and 35pt of Y1]{$D_{y_1}$};

\end{tikzpicture}
}   \hspace{2mm}
    \subcaptionbox{$z_0 = x_1$, $z_1 = y_1$, $z_2 = x_i$, $z_3 = y_j$, implying that the answer to $\HPC_3$ in this example is $y_j$.}[0.45\textwidth]{

    \begin{tikzpicture}[ auto ,node distance =1cm and 2cm , on grid , semithick , state/.style ={ circle ,top color =white , bottom color = white , draw, black , text=black}, every node/.style={inner sep=0,outer sep=0}]

\node[state,rectangle, top color=white, black, bottom color=white, text=black, opacity=0, text opacity=1,  align=center]  (XX){$\XX:$};
\node[state,rectangle, top color=white, black, bottom color=white, text=black, minimum height=20pt, minimum width = 180pt,  align=center]  (X)[right=105pt of XX]{\tikz{\draw[step=20pt] (0,0)  grid (180pt,20pt);}};
\node[state,rectangle, top color=white, black, bottom color=white, text=black, opacity=1, text opacity=1, minimum height=20pt, minimum width = 20pt]  (X1) [left=80pt of X]{$x_1$};
\node[state,rectangle, top color=white, black, bottom color=white, text=black, opacity=1, text opacity=1, minimum height=20pt, minimum width = 20pt]  (Xi) [right=60pt of X1]{$x_i$};
\node[state,rectangle, top color=white, black, bottom color=white, text=black, opacity=1, text opacity=1, minimum height=20pt, minimum width = 20pt]  (Xn) [right=100pt of Xi]{$x_n$};

\node[state,rectangle, top color=white, black, bottom color=white, text=black, opacity=0, text opacity=1,  align=center]  (YY)[below=60pt of XX]{$\YY:$};
\node[state,rectangle, top color=white, black, bottom color=white, text=black, minimum height=20pt, minimum width = 180pt,  align=center]  (Y)[below=60pt of X] {\tikz{\draw[step=20pt] (0,0)  grid (180pt,20pt);}};
\node[state,rectangle, top color=white, black, bottom color=white, text=black, opacity=1, text opacity=1, minimum height=20pt, minimum width = 20pt]  (Y1) [left=80pt of Y]{$y_1$};
\node[state,rectangle, top color=white, black, bottom color=white, text=black, opacity=1, text opacity=1, minimum height=20pt, minimum width = 20pt]  (Yj) [right=100pt of Y1]{$y_j$};
\node[state,rectangle, top color=white, black, bottom color=white, text=black, opacity=1, text opacity=1, minimum height=20pt, minimum width = 20pt]  (Yn) [right=60pt of Yj]{$y_n$};

\draw[->,line width=1pt] (X1) -- (Y1) node[near start]{(1)}; 
\draw[->,line width=1pt] (Y1) -- (Xi) node[near start]{(2)}; ; 
\draw[->,line width=1pt] (Xi) -- (Yj) node[near start]{(3)}; ; 

\end{tikzpicture}
}
    \caption{Illustration of the \HPC problem.}
    \label{fig:hpc}
\end{figure}

We define a \emph{phase} (similar to a round) for protocols that solve $\HPC$. In an odd (resp. even) phase, only $P_C$ and $P_D$ (resp. $P_A$ and $P_B$) are allowed to communicate with each other, 
and the phase ends once a message is sent to $P_A$ or $P_B$ (resp. $P_C$ or $P_D$). A protocol is called a \emph{k-phase} protocol iff it uses at most $k$ phases. See Appendix~\ref{app:phase} for more details.

It is easy to see that in $k+1$ phases, we can compute \HPCk with $O(k \cdot n)$ total communication by solving the $\SI$ instances 
corresponding to $z_0,z_1,\ldots,z_k$ one at a time in each phase. We prove that if we only have $k$ phases however, solving $\HPCk$ requires a large communication. 

\begin{theorem}[Informal]\label{thm:tech-hpck}
	Any $k$-phase protocol that outputs the correct solution to $\HPCk$ with constant probability requires $\Omega(n^2/k^2 + n)$ bits of communication. 
\end{theorem}

We give a proof sketch of the $\Omega(n^2/k^2)$ term in Theorem~\ref{thm:tech-hpck} in Section~\ref{sec:tech-cc-hpc} (the $\Omega(n)$ term follows immediately from set intersection lower bound). 
Before that, we show an application of this result in proving graph streaming lower bounds to illustrate our general approach.
 
\subsection{A Streaming Lower Bound for Minimum Weighted $s$-$t$ Cut Problem}\label{sec:tech-cut}

We sketch the proof of Result~\ref{res:cut} for directed graphs in this section. The proof is by a reduction from $\HPC$. We show how to turn any instance of $\HPCk$ for $k \geq 1$ into a 
weighted directed graph $G$ such that the minimum $s$-$t$ cut weight in $G$ determines the pointer $z_k$ in $\HPCk$. The rest of the proof then follows by standard arguments that relate communication complexity to 
space complexity of streaming algorithms. For the purpose of this proof, it would be more convenient to consider the maximum $s$-$t$ flow problem instead and then use min-cut max-flow duality. 

The high level construction of $G$ is as follows. The vertices in graph $G$ consists of $k+1$ layers each of size $n$ plus source and sink vertices $s$ and $t$. 
The even layers of this graph correspond to elements in $\XX$ while the odd layers correspond to $\YY$. The edges between the layers are then created by using the sets in the instances of $\SI$ inside the 
$\HPCk$ problem. The idea is to place the edges such that each vertex corresponding to $x_i$ (resp. $y_i$) in an even layer (resp. odd layer) can send a ``larger'' flow to the vertex corresponding to the target element of the instance $(A_{x_i},B_{x_i})$
(resp. target element of $(C_{y_i},D_{y_i})$) than any other vertex in the next layer. By choosing the weight of edges carefully and adding some extra gadgets, 
we ensure that the maximum $s$-$t$ flow should route the flow from $s$ along the path that corresponds to pointers $z_0,z_1,\ldots,z_k$. The vertices in the last layer have capacities that encode their identity and hence
the maximum $s$-$t$ flow value in this graph reveals the identity of $z_k$, thus solving $\HPCk$. See Figure~\ref{fig:cut} for an illustration. 

\begin{figure}[h!]
    \centering
\begin{tikzpicture}[ auto ,node distance =1cm and 2cm , on grid , semithick , state/.style ={ circle ,top color =white , bottom color = white , draw, black , text=black}, every node/.style={inner sep=0,outer sep=0}]

\node[circle, black, line width=0.15mm, minimum height=7pt, minimum width=7pt, fill=Gray, draw] (a1){};
\node[state, circle, black, line width=0.15mm, minimum height=7pt, minimum width=7pt] (a2) [below= 0.75cm of a1]{};
\node[state, circle, black, line width=0.15mm, minimum height=7pt, minimum width=7pt] (a3) [below=0.75cm of a2]{};
\node[state, circle, black, line width=0.15mm, minimum height=7pt, minimum width=7pt] (a4) [below=0.75cm of a3]{};
\node[state, circle, black, line width=0.15mm, minimum height=7pt, minimum width=7pt] (a5) [below=0.75cm of a4]{};
\node (v0) [below=0.75cm of a5]{$V_0$};


\node[rectangle, dashed, rounded corners = 3mm, inner sep=5pt, draw,  black, fit=(a1) (a5), line width=0.15mm] {};

\node[state, circle, black, draw, line width=0.15mm, minimum height=7pt, minimum width=7pt] (b1) [right=3cm of a1]{};
\node[ circle, black, line width=0.15mm, minimum height=7pt, minimum width=7pt, fill=Gray, draw] (b2) [below= 0.75cm of b1]{};
\node[state, circle, black, line width=0.15mm, minimum height=7pt, minimum width=7pt] (b3) [below=0.75cm of b2]{};
\node[state, circle, black, line width=0.15mm, minimum height=7pt, minimum width=7pt] (b4) [below=0.75cm of b3]{};
\node[state, circle, black, line width=0.15mm, minimum height=7pt, minimum width=7pt] (b5) [below=0.75cm of b4]{};
\node (v1) [below=0.75cm of b5]{$V_1$};

\node[rectangle, dashed,rounded corners = 3mm, inner sep=5pt, draw,  black, fit=(b1) (b5), line width=0.15mm] {};

\node[state, circle, black, draw, line width=0.15mm, minimum height=7pt, minimum width=7pt] (c1) [right=3cm of b1]{};
\node[state, circle, black, line width=0.15mm, minimum height=7pt, minimum width=7pt] (c2) [below= 0.75cm of c1]{};
\node[state, circle, black, line width=0.15mm, minimum height=7pt, minimum width=7pt] (c3) [below=0.75cm of c2]{};
\node[ circle, black, line width=0.15mm, minimum height=7pt, minimum width=7pt, fill=Gray, draw] (c4) [below=0.75cm of c3]{};
\node[state, circle, black, line width=0.15mm, minimum height=7pt, minimum width=7pt] (c5) [below=0.75cm of c4]{};
\node (v2) [below=0.75cm of c5]{$V_2$};

\node[rectangle, dashed,rounded corners = 3mm, inner sep=5pt, draw,  black, fit=(c1) (c5), line width=0.15mm] {};

\node[state, circle, black, draw, line width=0.15mm, minimum height=7pt, minimum width=7pt] (d1) [right=3cm of c1]{};
\node[state, circle, black, line width=0.15mm, minimum height=7pt, minimum width=7pt] (d2) [below= 0.75cm of d1]{};
\node[ circle, black, line width=0.15mm, minimum height=7pt, minimum width=7pt, fill=Gray, draw] (d3) [below=0.75cm of d2]{};
\node[state, circle, black, line width=0.15mm, minimum height=7pt, minimum width=7pt] (d4) [below=0.75cm of d3]{};
\node[state, circle, black, line width=0.15mm, minimum height=7pt, minimum width=7pt] (d5) [below=0.75cm of d4]{};
\node (v3) [below=0.75cm of d5]{$V_3$};

\node[rectangle, dashed, rounded corners = 3mm, inner sep=5pt, draw,  black, fit=(d1) (d5), line width=0.15mm] {};

\node[state, circle, black, draw, line width=0.15mm, minimum height=7pt, minimum width=7pt] (s)[left=3cm of a3]{};
\node (ss) [below=0.75cm of s]{$s$};
\node[state, circle, black, draw, line width=0.15mm, minimum height=7pt, minimum width=7pt] (t)[right=3cm of d3]{};
\node (tt) [below=0.75cm of t]{$t$};

\node (l1)[above right=1cm and 2cm of b1]{};

\draw[in=180, black, line width=0.5pt] (b1) to (l1.center);
\draw[in=180, black, line width=0.5pt] (b2) to (l1.center);
\draw[in=180, black, line width=0.5pt] (b3) to (l1.center);
\draw[in=180, black, line width=0.5pt] (b4) to (l1.center);
\draw[in=180, black, line width=0.5pt](b5) to (l1.center);
\draw[in=180, black, line width=0.5pt,->,>=latex, rounded corners = 20pt] (l1.center) -| (t);

\node (l2)[above right=1cm and 2cm of c1]{};

\draw[in=180, black, line width=0.5pt](c1) to (l2.center);
\draw[in=180, black, line width=0.5pt] (c2) to (l2.center);
\draw[in=180, black, line width=0.5pt](c3) to (l2.center);
\draw[in=180, black, line width=0.5pt](c4) to (l2.center);
\draw[in=180, black, line width=0.5pt](c5) to (l2.center);

\node (l3)[above right=1cm and 2cm of d1]{};

\draw[in=180, black, line width=0.5pt](d1) to (l3.center);
\draw[in=180, black, line width=0.5pt] (d2) to (l3.center);
\draw[in=180, black, line width=0.5pt](d3) to (l3.center);
\draw[in=180, black, line width=0.5pt](d4) to (l3.center);
\draw[in=180, black, line width=0.5pt](d5) to (l3.center);

\draw[->,>=latex,bend left, black, line width=0.5pt] (d1) to (t);
\draw[->,>=latex,bend left, black, line width=0.5pt] (d2) to (t);
\draw[->,>=latex,black, line width=0.5pt] (d3) to (t);
\draw[->,>=latex,bend right, black, line width=0.5pt](d4) to (t);
\draw[->,>=latex,bend right, black, line width=0.5pt] (d5) to (t);

\draw[->,>=latex, black, line width=0.5pt, bend left] (s) to (a1) node[near start]{}; 

\draw[->,>=latex, blue, bend right, line width=1.5pt] (a1) to (b2) node[near start]{}; 
\draw[->,>=latex, blue, bend right, line width=1.5pt] (a1) to (b4) node[near start]{}; 

\draw[->,>=latex, red, bend left, line width=1.5pt] (a1) to (b3) node[near start]{}; 
\draw[->,>=latex, red, bend left, line width=1.5pt] (a1) to (b2) node[near start]{}; 

\draw[->,>=latex, ForestGreen, bend right, line width=1.5pt] (b2) to (c1) node[near start]{}; 
\draw[->,>=latex, ForestGreen, bend right, line width=1.5pt] (b2) to (c4) node[near start]{}; 

\draw[->,>=latex, YellowOrange, bend left, line width=1.5pt] (b2) to (c4) node[near start]{}; 
\draw[->,>=latex, YellowOrange, bend left, line width=1.5pt] (b2) to (c3) node[near start]{}; 

\draw[->,>=latex, blue, bend right, line width=1.5pt] (c4) to (d3) node[near start]{}; 
\draw[->,>=latex, blue, bend right, line width=1.5pt] (c4) to (d2) node[near start]{}; 

\draw[->,>=latex, red, bend left, line width=1.5pt] (c4) to (d3) node[near start]{}; 
\draw[->,>=latex, red, bend left, line width=1.5pt] (c4) to (d4) node[near start]{};

\end{tikzpicture}
    \caption{Illustration of the graph in the reduction for minimum $s$-$t$ cut from $\HPC_3$ with $n=5$. The black (thin) edges form input-independent gadgets
    while blue, red , brown, and green (thick) edges depend on the inputs of $P_A$, $P_B$, $P_C$, and $P_D$, respectively. Marked nodes denote the vertices corresponding to pointers $z_0,\ldots,z_3$.
    The input-dependent edges incident on ``non-pointer'' vertices are omitted. This construction has parallel edges but Remark~\ref{rem:simple-graph} shows how to remove them.
    }
    \label{fig:cut}
\end{figure}
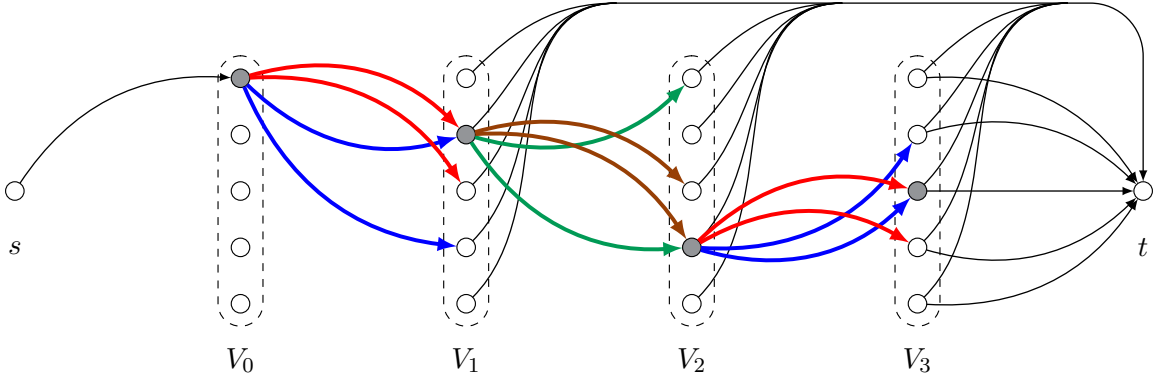

It is now easy to show that any $(k/3)$-pass streaming algorithm for minimum weighted $s$-$t$ cut with space $S$ can be turned into a $k$-phase protocol for $\HPCk$ with communication cost $O(k \cdot S)$ using this reduction. 
As the graph $G$ constructed above has $O(k \cdot n)$ vertices, we obtain the desired lower bound in Result~\ref{res:cut} by the communication complexity lower bound for $\HPC$ in Theorem~\ref{thm:tech-hpck}. 

\subsection{Communication Complexity of Hidden-Pointer Chasing}\label{sec:tech-cc-hpc}

We now sketch the proof of Theorem~\ref{thm:tech-hpck} which is the main technical contribution of this paper. 
Let $\distSI$ be a hard distribution on instances $(A,B)$ for $\SI$. In this distribution $A$ and $B$ are each sets of size almost $n/3$ such that they intersect in 
a unique element in the universe chosen uniformly at random.  We define the distribution $\distHPC$ over inputs of $\HPC$ as
the distribution in which all instances $(A_x,B_x)$ and $(C_y,D_y)$ for $x \in \XX$ and $y \in \YY$ are sampled independently from $\distSI$ (note that $\distHPC$ is not a product distribution as $\distSI$ is not a product distribution). 

Fix any $k$-phase {deterministic} protocol $\protHPC$ for $\HPCk$ throughout this section and suppose towards a contradiction that $\CC{\protHPC}{} = o(n^2/k^2)$ (the lower bound extends to randomized
protocols by Yao's minimax principle~\cite{Yao83}). 
For any $j \in [k]$, we define $\Prot_j$ as the set of all messages communicated by $\protHPC$ in phase $j$ and 
$\Prot:= (\Prot_1,\ldots,\Prot_k)$ as the transcript of the protocol $\protHPC$. 
We further define $Z = (z_1,\ldots,z_k)$, $E_j := (\Prot^{<j},Z^{<j})$ for any $j > 1$, and $E_1 = z_0$. We think of $E_j$ as the information ``easily known'' to players at the beginning of phase $j$. 
The main step of the proof of Theorem~\ref{thm:tech-hpck} is the following key lemma which we prove inductively. 

\begin{lemma}[Informal]\label{lem:tech-hpc-induction}
	For all $j \in [k]$: 
	$
		\Ex_{(E_j,\Prot_j)} \Bracket{\tvd{\distribution{\rZ_j \mid E_j,\Prot_j}}{\distribution{\rZ_j}}} = o(1). 
	$
\end{lemma}

Lemma~\ref{lem:tech-hpc-induction} states that if the communication cost of a protocol is ``small'', i.e., is $o(n^2/k^2)$, then 
even after communicating the messages in the first $j$ phases of the protocol, distribution of $z_j$ is still ``close'' to being uniform. This in particular implies that at the end of the protocol, i.e., at the end of phase $k$, the target pointer $z_k$ is essentially 
distributed as in its original distribution (which is uniform over $\YY$ or $\XX$ depending on whether $k$ is odd or even). 
Hence $\protHPC$ should not be able to find $z_k$ at the end of phase $k$. The proof of Theorem~\ref{thm:tech-hpck} follows easily from this intuition. 

\paragraph{Proof Sketch of Lemma~\ref{lem:tech-hpc-induction}.} The first step of proof is to show that finding the target element of a \emph{uniformly at random} chosen instance of $\SI$ (as opposed to an instance corresponding to any particular pointer) 
in $\HPC$ is not possible with low communication. For any $x \in \XX$ and any $y \in \YY$, define 
the random variables $\rT_{x} \in \YY$ and $\rT_y \in \XX$, which correspond to the target elements of \SI on $(A_x,B_x)$ and $(C_y,D_y)$, respectively. 
The following lemma formalizes the above statement. For simplicity, we only state it for $\rT_x$ for $x \sim \unif_{\XX}$; an identical bound also hold for $\rT_y$ for $y \sim \unif_{\YY}$. 

\begin{lemma}[Informal]\label{lem:tech-multi-SI}
	For $j \in [k]$:  	$\Ex_{(E_j,\Prot_j)}\Ex_{x \sim \unif_{\XX}}\bracket{\tvd{\distribution{\rT_x \mid E_j,\Prot_j}}{\distribution{\rT_x}}} = o(1)$.	
\end{lemma}

Let us first see why Lemma~\ref{lem:tech-multi-SI} implies Lemma~\ref{lem:tech-hpc-induction}. The proof is by induction. Consider some phase $j \in [k]$ and suppose
$j$ is odd by symmetry. The goal is to prove that distribution of $\rZ_j$ conditioned on $(E_j,\Prot_j) = (z_1,\ldots,z_{j-1},\Prot_1,\ldots,\Prot_{j-1},\Prot_j)$ is close to original distribution of $\rZ_j$ (on average over choices of $(E_j,\Prot_j)$).  
Notice that since we assumed $j$ is odd, $\rZ_j$ is a function of the inputs to $P_A$ and $P_B$. On the other hand, in an odd phase, only the players $P_C$ and $P_D$ communicate and hence
$\Prot_j$ is a function of the inputs to these players. Conditioning on $E_j$ and using the rectangle property of deterministic protocols (see Fact~\ref{fact:rectangle}), together with the fact
that inputs to $P_A,P_B$ are independent of inputs to $P_C,P_D$, implies that $\rZ_j \perp \Prot_j \mid E_j$. We now have:  
\begin{enumerate}[label=(\roman*)]
	\item Conditioned on $z_{j-1}$, $\rZ_{j}$ is the target element of the instance $(A_{z_{j-1}},B_{z_{j-1}})$, i.e., $\rZ_{j} = \rT_{z_{j-1}}$.
	\item $z_{j-1}$ itself is distributed according to $\distribution{\rZ_{j-1} \mid E_{j-1},\Prot_{j-1}}$ (because we removed the conditioning on $\Prot_j$ by the above argument). 
	\item $\distribution{\rZ_{j-1} \mid E_{j-1},\Prot_{j-1}}$ is close to the uniform distribution by induction. 
\end{enumerate}

As such we can now simply apply Lemma~\ref{lem:tech-multi-SI} (by replacing $x$ with $z_{j-1}$ since they essentially have the same distribution) and obtain that
distribution of $\rZ_j = \rT_{z_{j-1}}$ with and without conditioning on $(E_j,\Prot_j)$ is almost the same (averaged over choices of $(E_j,\Prot_j)$), proving the lemma. 

\paragraph{Proof Sketch of Lemma~\ref{lem:tech-multi-SI}}
The proof of this lemma is based on a direct-sum style argument combined with a new result that we prove for \SI. 
The direct-sum argument implies that since $x$ is chosen uniformly at random from $n$ elements in $\XX$, and protocol $\protHPC$ is communicating $o(n^2)$ bits in total, 
then it can only reveal $o(n)$ bits of information about the instance $(A_x,B_x)$. This part follows the standard direct-sum arguments for information complexity (see, e.g.~\cite{BarakBCR10,BravermanR11}) but
we also need to take into account that if $x$ is one of the pointers we conditioned on in $E_j$, then $\protHPC$ may reveal more information about $(A_x,B_x)$; fortunately, this event happens with negligible probability for $k \ll n$ and so the 
argument continues to hold.  

By above argument, proving Lemma~\ref{lem:tech-multi-SI} reduces to showing that if a protocol reveals $o(n)$ bits of information about an instance of $\SI$, 
then the distribution of the target element varies from the uniform distribution in total variation distance by only $o(1)$. This is the main part of the proof of Lemma~\ref{lem:tech-multi-SI} and is 
precisely the content of our next  technical result in the following section.

\subsection{A New Communication Lower Bound for Set Intersection}\label{sec:tech-si}

We say that a protocol $\protSI$ \emph{$\eps$-solves} $\SI$ on the distribution $\distSI$ iff it can alter the distribution of the target element from its original distribution by at least $\eps$ in total variation distance, i.e., 
$\Ex_{\ProtSI \sim \rProtSI}\Bracket{\tvd{\distribution{\rT \mid \ProtSI}}{\distribution{\rT}}} \geq \eps$; here $\rProtSI$ and $\rT$ are the random variables for the transcript of the protocol (including public randomness) and the target element, respectively. 

To finish the proof of Lemma~\ref{lem:tech-multi-SI}, we need to prove that a protocol that $\Omega(1)$-solves $\SI$ has $\Omega(n)$ communication cost (even information cost). 
Note that $\eps$-solving is an algorithmically simpler task than finding the target element. For example, a protocol may change the distribution of $\rT$ to having $(1+\eps)/n$ probability on $n/2$ elements and $(1-\eps)/n$ probability on the remaining $n/2$. 
This $\eps$-solves $\SI$ yet the target element can only be found with probability $(1+\eps)/n$ in this distribution. On the other hand, any protocol that finds the target element with probability $p \in (0,1)$
also $p$-solves $\SI$. Because of this, the lower bounds mentioned in Section~\ref{sec:prelim} for set intersection do not suffice for our purpose.  Instead, we prove the following theorem in this paper. 

\begin{theorem}[Informal]\label{thm:tech-SI-tvd} 
    Any protocol $\protSI$ that $\eps$-solves $\SI$ on distribution $\distSI$ has internal information cost $\ICost{\protSI}{\distSI} = \Omega(\eps^2 \cdot n)$.
\end{theorem}

As information cost lower bounds communication cost (see Proposition~\ref{prop:cc-ic}), Theorem~\ref{thm:tech-SI-tvd} also proves a communication lower bound for $\SI$ (although we need the stronger result
for information cost in our proofs). 
By our discussion earlier, Theorem~\ref{thm:tech-SI-tvd} can be used to finalize the proof of Lemma~\ref{lem:tech-multi-SI} (and hence Theorem~\ref{thm:tech-hpck}). 
We now give an overview of the proof of Theorem~\ref{thm:tech-SI-tvd}.

For an instance $(A,B)$ of $\SI$, with a slight abuse of notation, we write $A := (a_1,\ldots,a_n)$ and $B := (b_1,\ldots,b_n)$ for $a_i,b_i \in \set{0,1}$ as characteristic vector of the sets given to Alice and Bob. Under this notation, 
the target element corresponds to the unique index $t \in [n]$ such that $(a_t,b_t) = (1,1)$. The proof of Theorem~\ref{thm:tech-SI-tvd} is based on reducing $\SI$ to a special case of this problem on only $2$ coordinates, which we define
as the $\PI$ problem. In $\PI$, Alice and Bob are given $(x_1,x_2)$ and $(y_1,y_2)$ in $\set{0,1}^{2}$ and their goal is to find the unique index $k \in \set{1,2}$ such that $(x_k,y_k) = (1,1)$. We use $\distPI$ to denote the hard distribution for this problem
which is equivalent to $\distSI$ for $n=2$. 

Given a protocol $\protSI$ for $\eps$-solving $\SI$ on $\distSI$, we design a protocol $\protPI$ for finding the index $k$ in instances of $\PI$ sampled from $\distPI$ with probability $1/2 + \Omega(\eps)$. 
The reduction is as follows. 

\paragraph{Reduction:} Alice and Bob publicly sample $i,j \in [n]$ uniformly at random without replacement. Then, Alice sets $a_i = x_1$ and $a_j = x_2$ and Bob sets $b_i = y_1$ and $b_j = y_2$, using their given inputs in $\PI$. 
The players sample the remaining coordinates of $(A,B)$ in $[n] \setminus \set{i,j}$ using a combination of public and private randomness that we explain later in the proof sketch of Lemma~\ref{lem:tech-p2}. This sampling ensures that the resulting instance $(A,B)$ of $\SI$ is sampled from $\distSI$ such that its target element is $i$ when $k=1$ and is $j$ when $k=2$. 
After this, the players run the protocol $\protSI$ on $(A,B)$
and let $\ProtSI$ be the transcript of this protocol. Using this, Bob computes the distribution $\distribution{\rT \mid \ProtSI} = (p_1,\ldots,p_n)$ which assigns probabilities to elements in $[n]$ as being the target element. 
Finally, Bob checks the value of $p_i$ and $p_j$ and return $k=1$ if $p_i > p_j$ and $k=2$ otherwise (breaking the ties consistently when $p_i = p_j$). 
The remainder of the proof consists of three main steps:
\begin{enumerate}[label=($\roman*$)]
	\item Proving the correctness of protocol $\protPI$: 
	
	\begin{lemma}[Informal]\label{lem:tech-p1}
	    Protocol $\protPI$ outputs the correct answer with probability  $\frac{1}{2} + \Omega(\eps)$.
	\end{lemma}

	\item Proving an upper bound on ``information cost'' of $\protPI$ (the reason for quotations is that strictly speaking this quantity is not the information cost of $\protPI$ but rather 
	a lower bound for it). 
	
	\begin{lemma}[Informal]\label{lem:tech-p2}
    Let $\rProtPI$ denote the random variable for the transcript of the protocol $\protPI$ and $\rK$ be the random variable for the index $k$ in distribution $\distPI$. We have, 
    \begin{align*}
        \mii{\rX_1,\rX_2}{\rProtPI \mid \rY_1,\rY_2, \rK}{\distPI} + \mii{\rY_1,\rY_2}{\rProtPI \mid \rX_1,\rX_2, \rK}{\distPI} &\leq \frac{1}{n-1} \cdot \ICost{\protSI}{\distSI}.
    \end{align*}
   \end{lemma}
   	\item Proving a lower bound on ``information cost'' (as used in Part $(ii)$) of protocols for $\PI$: 
   	
	\begin{lemma}\label{lem:tech-p3}
    If $\protPI$ outputs the correct answer on $\distPI$ with probability at least $\frac{1}{2} + \Omega(\eps)$, then, 
    \begin{align*}
        \mii{\rX_1,\rX_2}{\rProtPI \mid \rY_1,\rY_2, \rK}{\distPI} + \mii{\rY_1,\rY_2}{\rProtPI \mid \rX_1,\rX_2, \rK}{\distPI} = \Omega(\eps^2).
    \end{align*}
	\end{lemma}

\end{enumerate} 
\noindent
By Lemma~\ref{lem:tech-p2}, $\IC{\protSI}{\distSI}$ is $\Omega(n)$ times larger than LHS of Lemma~\ref{lem:tech-p3}, and
this, combined with Lemma~\ref{lem:tech-p1}, implies that information cost of $\protSI$ needs to be $\Omega(\eps^2) \cdot \Omega(n)$, proving Theorem~\ref{thm:tech-SI-tvd}.  

\paragraph{Proof Sketch of Lemma~\ref{lem:tech-p1}.} Let us again consider a protocol $\protSI$ such that $\distribution{\rT \mid \ProtSI}$ is putting $(1+\eps)/n$ mass over $n/2$ elements and $(1-\eps)/n$ mass on the remaining ones.
 Suppose that the correct answer to the instance of $\PI$ is index $1$. We know that in this case, the index $i$ chosen by $\protPI$ will be 
the target index $t$ in the instance $(A,B)$. A key observation here is that the index $j$ however can be any of the coordinates in instance $(A,B)$ other than the target element with the same probability. 
As such, parameters $p_i$ and $p_j$ used to decide the answer in $\protPI$ are distributed as follows: $p_i$ is sampled from $\distribution{\rT \mid \ProtSI}$ and hence has value $(1+\eps)/n$ with probability $(1+\eps)/2$ and $(1-\eps)/n$ with probability
$(1-\eps)/2$. On the other hand, $p_j$ is chosen uniformly at random from $(p_1,\ldots,p_n)$ and hence is $(1+\eps)/n$ or $(1-\eps)/n$ with the same probability of half. Thus $p_i > p_j$
with probability $1/2+\Omega(\eps)$ and hence $\protPI$ has $\Omega(\eps)$ advantage over random guessing. 

The proof of Lemma~\ref{lem:tech-p1} then formalizes the observations above and extend this argument to any protocol $\protSI$ that $\eps$-solves $\SI$ no matter how
it alters the distribution of the target element. 

\paragraph{Proof Sketch of Lemma~\ref{lem:tech-p2}.} We first note that the LHS in Lemma~\ref{lem:tech-p2} is \emph{not} the internal information cost of $\protPI$ 
due to further conditioning on $\rK$ (this term can only be smaller than $\IC{\protPI}{\distPI}$). Hence, Lemma~\ref{lem:tech-p2} 
is proving a ``weaker'' statement than a direct-sum result for information cost of $\protPI$ based on $\protSI$. The reason for settling for this weaker statement has to do with the fact that the coordinates in distribution $\distSI$ are \emph{not} chosen 
independently (see Section~\ref{sec:cc-hpc} for more detail). 

The intuition behind the proof is as follows. The LHS in Lemma~\ref{lem:tech-p3} is the information revealed about the input of players (in $\PI$) averaged over choices of $k=1$ and $k=2$.
Let us assume $k=1$ by symmetry. In this case, this quantity is simply the information revealed about $(x_2,y_2)$ by the protocol as $(x_1,y_1) = (1,1)$ and hence has no entropy. 
However, when $k=1$, $(x_2,y_2)$ is embedded in index $j$, i.e., $(x_2,y_2) = (a_j,b_j)$ and has the same distribution as all other coordinates in $A_{-i},B_{-i}$. 
As such, since the protocol $\protSI$ called inside $\protPI$ is oblivious to the choice of $j$, the information revealed about $(a_j,b_j)$ in average is smaller than the information revealed by $\protSI$ about $A_{-i},B_{-i}$ 
(which itself is at most the information cost of $\protSI$) by a factor of $n-1$.

This outline oversimplifies many details. One such detail is the way of ensuring a ``symmetric treatment'' of both indices $i$ and $j$. This is crucial for the above argument to work for both $k=1$ and $k=2$ cases simultaneously, without the players knowing which index the ``averaging'' of information is being done for (index $j$ in the context of the discussion above). 
The key step in making this information-theoretic argument work is the following public-private sampling: Alice and Bob use public randomness to pick an integer $\ell \in [n-2]$ uniformly at random and then pick a set $S$ of size $\ell$ uniformly at random from $[n] \setminus \set{i,j}$. 
Next, the players sample $a_{i'}$ and $b_{j'}$ for $i' \in S$ and $j' \in ([n] \setminus \set{i,j}) \setminus S$ from $\distSI$ again using public randomness. Finally, each player samples the remaining coordinates in the input 
using private randomness from $\distSI$. Figure~\ref{fig:si-sample} gives an example. 

\input{si-sampling}

\paragraph{Proof Sketch of Lemma~\ref{lem:tech-p3}.}
Let $\Prot_{\disjin{x_1}{x_2}{y_1}{y_2}}$ denote the transcript of the protocol condition on the inputs $(x_1,x_2)$ and $(y_1,y_2)$ to Alice and Bob. 
Suppose towards a contradiction that the LHS of Lemma~\ref{lem:tech-p3} is $o(\eps^2)$. By focusing on the conditional terms when $k=1$, we can
show that distribution of $\Prot_{\disjin{1}{x'_2}{1}{y'_2}}$ and $\Prot_{\disjin{1}{x''_2}{1}{y''_2}}$ for all choices of $(x'_2,y'_2)$ and $(x''_2,y''_2)$ in the support of $\distPI$ are quite close. This is intuitively because the information revealed about $(x_2,y_2)$ by $\protPI$ conditioned on $k=1$ is small (the same result holds for 
$\Prot_{\disjin{x'_2}{1}{y'_2}{1}}$ and $\Prot_{\disjin{x''_2}{1}{y''_2}{1}}$ by $k=2$ terms). 

Up until this point, there is no contradiction as the answer to inputs $(1,*)$,$(1,*)$ to Alice and Bob is always $1$ and hence there is no problem with the corresponding transcripts 
in $\Prot_{\disjin{1}{*}{1}{*}}$ to be similar (similarly for $\Prot_{\disjin{*}{1}{*}{1}}$ separately). However, we combine this  with the cut-and-paste property of randomized protocols  based on Hellinger distance (see Fact~\ref{fact:r-rectangle})
to argue that in fact the distribution of $\Prot_{\disjin{1}{0}{1}{0}}$ and $\Prot_{\disjin{0}{1}{0}{1}}$ are also similar. This then implies that $\Prot_{\disjin{1}{*}{1}{*}}$ essentially has the same distribution as 
$\Prot_{\disjin{*}{1}{*}{1}}$; but then this is a contradiction as the answer to the protocol (which is only a function of the transcript) needs to be different between these two types of inputs.

\section{The Set Intersection Problem}\label{sec:SI}

Starting from this section, we delve into the formal proofs of our results.  
 This section contains our new lower bound for the set intersection problem (stated informally in Theorem~\ref{thm:tech-SI-tvd}). 
Recall that $\SI$ is a two-player communication problem in which Alice and Bob are given sets $A$ and $B$ from $[n]$, respectively, with the promise that there exists a unique
element $t$ such that $\set{t} = A \cap B$. The goal is for Alice and Bob to find $t$, referred to as the \emph{target element}. It is sometimes more convenient to consider the characteristic vector
of sets $A$ and $B$ rather than the sets directly. Hence, with a slight abuse of notation, we write $A := (a_1,\ldots,a_n) \in \set{0,1}^{n}$ and $B:=(b_1,\ldots,b_n) \in \set{0,1}^{n}$ where $a_i = 1$ (resp. $b_i = 1$) 
iff the element $i$ belongs to the set $A$ (resp. to $B$). In this notation, the target element $t$ corresponds to the \emph{unique} index where $(a_t,b_t) = (1,1)$. 

The \SI problem is closely related to the well-known \emph{set disjointness} problem. It is in fact straightforward to prove an $\Omega(n)$ lower bound on the communication complexity of $\SI$ using a simple reduction from the 
set disjointness problem. However, in this paper, we are interested in an algorithmically simpler variant of this problem which we define below. 

\subsection{Problem Statement}\label{sec:SI-def}

Consider the following distribution $\distSI$ for $\SI$. 

\begin{tbox}
    \textbf{Distribution $\distSI$} on sets $(A,B)$ from the universe $[n]$: 

    \begin{enumerate}
        \item Define $\mu$ as the uniform distribution over the set $\set{(0,0),(0,1),(1,0)}$. 
        \item For $i \in [n]$, choose $(a_i,b_i)$ independently from distribution $\mu$. 
        \item Sample an element $t \in [n]$ uniformly at random and change $(a_t,b_t) = (1,1)$. 

    \end{enumerate}
\end{tbox}

Rather than finding the target element $t$, we are only interested in slightly reducing the ``uncertainty'' about its identity as formalized below. 
\begin{definition}\label{def:eps-solve}
    We say that a protocol $\protSI$ \textbf{$\bm{\eps}$-solves} the $\SI$ problem on the distribution $\distSI$ iff 
    \begin{align}
        \Ex_{\ProtSI \sim \rProtSI}\Bracket{\tvd{\distribution{\rT \mid \ProtSI}}{\unif_{[n]}}} \geq \eps, \label{eq:si}
    \end{align} 
    where $\rT$ is the random variable for the target element and $\unif_{[n]}$ is the uniform distribution on $[n]$. 
\end{definition}

Let us first consider two ``extreme examples'' of a protocol that $\eps$-solves $\SI$ and see how much communication is needed to realize each one. 

\begin{example}\label{ex:si-1} One way of ensuring Eq~(\ref{eq:si}) is to have protocols that after communication can rule out  $\Theta(\eps \cdot n)$ elements as candidates for $t$ and leave the 
target element to be uniformly distributed on the remaining $n-\Theta(\eps \cdot n)$ elements. 

Intuitively, such a protocol should require a large communication as it is making a significant ``progress'' towards finding the target element. Indeed, if the communication cost of this protocol is small, we can run this protocol again on the 
remaining candidates and shrink their number further, and continue doing this until we find the target element $t$, without making a large communication. This contradicts the  $\Omega(n)$ communication lower bound for finding the element $t$ exactly. 
\end{example}

\begin{example}\label{ex:si-2} Another way of satisfying Eq~(\ref{eq:si}) is to have protocols that simply change the probability mass of the target element $t$ on half of the elements
from $1/n$ to $(1+\eps)/n$, and on the remaining half from $1/n$ to $(1-\eps)/n$. 

Analyzing the communication cost of such protocols is distinctly more delicate. On the surface, it does not seem that the protocol has made much ``progress'' towards finding the target element $t$ as nearly all elements are still quite likely candidates
for being the target. Hence, to show such protocols require large communication, we now need to go beyond reducing this problem to finding the target element $t$ exactly. Roughly speaking, we show that to be able to make such a change
in distribution of $t$, the protocol needs to communicate non-trivial information for every potential element, hence requiring a large communication again. 
\end{example}

\medskip

In the following, we show that no matter
how a protocol decides to change the variation distance of $t$ from its original distribution, it needs a large communication. However, we also encourage the reader to consider our arguments in the context of the above two examples for 
concreteness. 

\subsection{Communication Complexity of $\eps$-solving \SI}

We prove the following lower bound on the information cost of protocols for $\eps$-solving $\SI$. 

\begin{theorem}\label{thm:SI-tvd} 
    Suppose $\protSI$ is a protocol for $\SI$ on instances $(A,B)$ sampled from $\distSI$.  Let $\ProtSI$ denote the transcript of the protocol $\protSI$. If 
    $\Ex_{\ProtSI \sim \rProtSI}\Bracket{\tvd{\distribution{\rT \mid \ProtSI}}{\unif_{[n]}}} \geq \eps$, i.e., $\protSI$ $\eps$-solves $\SI$,
    then the internal information cost of $\protSI$ on $\distSI$ is $\ICost{\protSI}{\distSI} = \Omega(\eps^2 \cdot n)$.
\end{theorem}

We shall remark that for our purpose, we crucially use the fact that the lower bound in Theorem~\ref{thm:SI-tvd} is for the internal information cost and for the distribution $\distSI$. However, 
as information cost lower bounds communication cost by Proposition~\ref{prop:cc-ic}, this immediately implies that communication complexity of $\SI$ is also large, which is of independent interest. 

\begin{corollary}\label{cor:SI-tvd}
    Any protocol $\protSI$ for $\eps$-solving $\SI$ on distribution $\distSI$ needs to communicate $\Omega(\eps^2 \cdot n)$ bits of communication, i.e., $\CC{\prot}{\dist} = \Omega(\eps^2 \cdot n)$. 
\end{corollary}

One standard approach to proving the lower bound in Theorem~\ref{thm:SI-tvd} is to reduce the \SI problem---via a direct-sum type argument---to \emph{many} instances of a \emph{simpler} problem, and then prove the lower bound for the
simpler problem directly. To do so, we reduce $\SI$ to the same problem on only two coordinates, which we refer to as the \emph{pair intersection} problem, denoted by $\PI$. 
In $\PI$, Alice and Bob are given tuples $(x_1,x_2) \in \set{0,1}^{2}$ and $(y_1,y_2) \in \set{0,1}^{2}$, respectively (we 
also use the concise notation $\disjin{x_1}{x_2}{y_1}{y_2}$ to denote the joint inputs to the players), with the promise that there exists a unique index $k \in \set{1,2}$ such that $(x_k,y_k) = (1,1)$. The goal is to output 
the index $k$. Note that this problem is equivalent to 
$\SI$ when $n=2$ modulo the fact that here we actually care about finding $k$ as opposed to $\eps$-solving (to avoid ambiguity, we use $k$ to denote the target element for $\PI$ and $t$ for $\SI$). 
Consider the following distribution which is equivalent to $\distSI$ for $n=2$.

\begin{tbox}
    \textbf{Distribution $\distPI$} on tuples $(x_1,x_2)$ and $(y_1,y_2)$ from $\set{0,1}^2$. 
    \begin{enumerate}
        \item For $i \in \set{1,2}$, choose $(x_i,y_i)$ uniformly at random from distribution $\mu$ (defined in $\distSI$). 
        \item Pick $k \in \set{1,2}$ uniformly at random and change $(x_k,y_k)$ to $(1,1)$.
    \end{enumerate}
\end{tbox}

We prove that any protocol that $\eps$-solves $\SI$ on $\distSI$ with internal information cost $o(\eps^2 \cdot n)$ bits can be used to obtain a protocol for 
$\PI$ that only reveals $o(\eps^2)$ bits of information about the input (with respect to distribution $\distPI$) but is able to solve this problem with probability at least $1/2+\eps$ on distribution $\distPI$. We then prove that such
a protocol cannot exist for $\PI$. We should note that the notion of information revealed for $\PI$ that we use is rather non-standard (it neither corresponds to internal information cost nor to external information cost that are typically studied). We elaborate more on this later in Lemma~\ref{lem:protDisj-info-cost}.

\subsection*{Proof of Theorem~\ref{thm:SI-tvd}}
In the following, let $\protSI$ be any protocol for $\SI$ that satisfies Eq~(\ref{eq:si}), i.e., $\eps$-solves $\SI$ on $\distSI$. We use this protocol to obtain
a protocol $\protPI$ for $\PI$.

\begin{tbox}
    \textbf{Protocol $\protPI$}: The protocol for $\PI$ using a protocol $\protSI$ for $\SI$. 

    \smallskip

    \textbf{Input:} An instance $\disjin{x_1}{x_2}{y_1}{y_2} \sim \distPI$. \\
    \textbf{Output:} $k \in \set{1,2}$ as the answer to $\PI$.

    \algline

    \begin{enumerate}
        \item \textbf{Sampling the instance.} The players create an instance $(A,B)$ of $\SI$ as follows (see Figure~\ref{fig:si-sample} on page~\pageref{fig:si-sample} for an illustration):
            \begin{enumerate}
                \item Using \underline{public coins}, Alice and Bob sample $i,j \in [n]$ uniformly without replacement. 
                \item Alice sets $a_i = x_1$ and $a_j = x_2$ and Bob sets $b_i = y_1$ and $b_j = y_2$, using their given inputs in $\PI$. 
                \item Using \underline{public coins}, Alice and Bob sample $\ell \in \set{0,1,\ldots,n-2}$ uniformly at random and then pick an $\ell$-subset $S$ of $[n] \setminus \set{i,j}$ uniformly at random. Let $\barS := ([n] \setminus \set{i,j}) \setminus S$.
                \item Using \underline{public coins}, Alice and Bob sample $A_{S},B_{\barS}$ independently from distribution $\mu$ (defined in $\distSI$). 
                \item Using \underline{private coins}, Alice samples the remaining coordinates in $A_{\barS}$ so that joint distribution of each coordinate is $\mu$. Similarly, Bob samples the coordinates in $B_S$. 
            \end{enumerate}
        \item \textbf{Computing the answer.} Alice and Bob run the protocol $\protSI$ on $(A,B)$ and let $\ProtSI$ be the transcript of the protocol. They compute the answer to $\PI$ as follows: 
            \begin{enumerate}
                \item The players compute the distribution $\distribution{\rT \mid \ProtSI} = (p_1,\ldots,p_n)$ where $\rT$ denotes the random variable for the target element of $\SI$.  

                \item Fix a total ordering $\comp$ on $[n]$ such that for $x \neq y \in [n]$, $x \comp y$ iff $p_x > p_y$ or $p_x = p_y$ and $x > y$. We use $x \compr y$ to mean $y \comp x$.

                \item Return $1$ if $i \comp j$ and $2$ otherwise.  
            \end{enumerate}
    \end{enumerate}
\end{tbox}

The following  observations are in order. Firstly, we note that the rather peculiar way of sampling the instances $(A,B)$ in $\protPI$ via public and private randomness is only for the purpose of making the information-theoretic arguments needed to reduce $\SI$ to $\PI$ work; for the purpose of correctness of the reduction, we only need the fact that these instances are sampled from $\distSI$ as captured by the following observation. 

\begin{observation}\label{obs:dist-PI-SI}
    For an input $\disjin{x_1}{x_2}{y_1}{y_2} \sim \distPI$, the distribution of the instances $(A,B)$ constructed in $\protPI$ is $\distSI$, where target $t=i$ when $x_1 \wedge y_1 = 1$ and target $t=j$ when $x_2 \wedge y_2 = 1$. 
\end{observation}

The following observation states a key property of the ``non-target'' index in $\distPI$. 

\begin{observation}\label{obs:dist-PI-SI-hidden}
    Conditioned on $x_1 \wedge y_1 = 0$ and any fixed choice of $(A,B)$, the index $i$ in $\protPI$ is {uniformly distributed} on $[n] \setminus \set{j}$ (similarly for index $j$ if $x_2 \wedge y_2 = 0$). 
\end{observation}
\begin{proof}
    Conditioned on $x_1 \wedge y_1 = 0$, the distribution of $(a_i,b_i)$ in $(A,B)$ is $\mu$, the same as all other indices except for $j$. 
\end{proof}

The proof of Theorem~\ref{thm:SI-tvd} consists of three main steps: bounding the error probability of protocol $\protPI$, analyzing the information cost of $\protPI$ in terms of information cost of $\protSI$, 
and proving a lower bound on the information cost of $\protPI$ based on its error probability. 
Formally, in the first step we prove that: 

\begin{lemma}[Correctness of $\protPI$]\label{lem:protDisj-correctness}
    For instances sampled from $\distPI$, $\protPI$ outputs the correct answer with probability at least $\frac{1}{2} + \Omega(\eps)$ (over the randomness of the distribution and the protocol). 
\end{lemma}

In the second step, we show that: 
\begin{lemma}[Information cost of $\protPI$]\label{lem:protDisj-info-cost}
    Let $\rProtPI$ denote the random variable for the transcript of the protocol $\protPI$ and $\rK$ be the random variable for the index $k$ in distribution $\distPI$. We have, 
    \begin{align*}
        \mii{\rX_1,\rX_2}{\rProtPI \mid \rY_1,\rY_2, \rK}{\distPI} + \mii{\rY_1,\rY_2}{\rProtPI \mid \rX_1,\rX_2, \rK}{\distPI} &\leq \frac{1}{n-1} \cdot \ICost{\protSI}{\distSI}.
    \end{align*}
\end{lemma}
\noindent
The LHS in Lemma~\ref{lem:protDisj-info-cost} is \emph{not} the internal information cost of $\protPI$ 
due to further conditioning on $\rK$. In fact, it is not hard to show that this quantity can only be smaller than the internal information cost of $\protPI$. Hence, Lemma~\ref{lem:protDisj-info-cost} 
is proving a ``weaker'' statement than a direct-sum result for internal information cost of $\protPI$ based on $\protSI$. The reason for settling for this weaker statement has to do with the fact that the coordinates in distribution $\distSI$ are \emph{not} chosen independently and so the stronger bound does not seem to be true for our reduction\footnote{Similar issues arise when analyzing information complexity
of set disjointness on \emph{intersecting} distributions~\cite{JayramKS03} as opposed to the more standard case of non-intersecting distributions (e.g.~\cite{Bar-YossefJKS02,BravermanM13,BravermanGPW13,WeinsteinW15}).}. Nevertheless, we show in the third part of the argument that this weaker statement suffices for our purpose. 

In the final step of the proof, we prove that any protocol for $\PI$ that has a small error probability should have a large information cost with respect to the measure in Lemma~\ref{lem:protDisj-info-cost}.

\begin{lemma}[Information complexity of $\PI$]\label{lem:PI-ic}
    Suppose $\protPI$ outputs the correct answer on $\distPI$ with probability at least $\frac{1}{2} + \Omega(\eps)$. Then, 
    \begin{align*}
        \mii{\rX_1,\rX_2}{\rProtPI \mid \rY_1,\rY_2, \rK}{\distPI} + \mii{\rY_1,\rY_2}{\rProtPI \mid \rX_1,\rX_2, \rK}{\distPI} = \Omega(\eps^2).
    \end{align*}
\end{lemma}

We prove each of these three lemmas in the following sections. Before that, we show Theorem~\ref{thm:SI-tvd} follows easily from these lemmas. 
\begin{proof}[Proof of Theorem~\ref{thm:SI-tvd} (assuming Lemmas~\ref{lem:protDisj-correctness},~\ref{lem:protDisj-info-cost}, and~\ref{lem:PI-ic})]
    Suppose towards a contradiction that $\protSI$ is a protocol that $\eps$-solves $\SI$ on $\distSI$ and has information cost $\ICost{\protSI}{\distSI} = o(\eps^2 \cdot n)$. 
    Create the protocol $\protPI$ using $\protSI$ as described in the reduction above. We have, 
    \begin{itemize}
        \item By Lemma~\ref{lem:protDisj-correctness}, 
            $\protPI$ outputs the correct answer on $\distPI$ w.p. at least $\frac{1}{2} + \Omega(\eps)$. 

        \item By Lemma~\ref{lem:protDisj-info-cost}, 
            $
            \mii{\rX_1,\rX_2}{\rProtPI\mid \rY_1,\rY_2, \rK}{\distPI} + \mii{\rY_1,\rY_2}{\rProtPI \mid \rX_1,\rX_2, \rK}{\distPI}  = o(\eps^2).
            $
    \end{itemize}
    However, these two properties contradict Lemma~\ref{lem:PI-ic}. As such, the internal information cost of $\protSI$ on $\distSI$ should be
    $\Omega(\eps^2 \cdot n)$, finalizing the proof. \Qed{Theorem~\ref{thm:SI-tvd}}

\end{proof}

\subsection*{Proof of Lemma~\ref{lem:protDisj-correctness}: Correctness of Protocol $\protPI$}

The following is a re-statement of Lemma~\ref{lem:protDisj-correctness} that we prove in this section. 

\begin{lemma*}[Restatement of Lemma~\ref{lem:protDisj-correctness}]
    For an instance $\disjin{x_1}{x_2}{y_1}{y_2} \sim \distPI$, $\protPI$ outputs the correct answer with probability at least $\frac{1}{2} + \Omega(\eps)$ (over the randomness of the distribution and the protocol). 
\end{lemma*}

To give some intuition about this lemma, let us consider the Examples~\ref{ex:si-1} and~\ref{ex:si-2}. Suppose the correct answer to the instance of $\PI$ is index $1$ and  
protocol $\protSI$ that we use in reduction is of the type described in Example~\ref{ex:si-1}. We know that the set of $n-\Theta(\eps \cdot n)$ elements computed by $\distSI$ definitely contains element $i$. 
What can be said about element $j$ here? By Observation~\ref{obs:dist-PI-SI-hidden}, the element $j$ is chosen uniformly at random from all elements $[n] \setminus \set{i}$, \emph{even} conditioned on a choice
of $A$ and $B$. As such, with probability $\Theta(\eps)$, element $j$ does not belong to the set of candidates for the target element computed by $\protSI$. In this case, protocol $\protPI$ outputs the correct answer. This allows us to infer that
$\protPI$ is able to get $\Theta(\eps)$ advantage over random guessing, exactly what is asserted by Lemma~\ref{lem:protDisj-correctness}. A similar argument also works if protocol $\protSI$ is of the type in Example~\ref{ex:si-2}. We now prove this lemma
for general protocols. 

\begin{proof}[Proof of Lemma~\ref{lem:protDisj-correctness}]
    Assume $x_1 \wedge y_1 = 1$, i.e., index $1$ is the correct answer to $\PI$ (the other case is symmetric). Let $(A,B)$ be the instance of $\SI$ constructed by $\protPI$ and 
    let $\ProtSI$ be the transcript of the protocol $\protSI$ on $(A,B)$ which is communicated inside $\protPI$. Recall that $\distribution{\rT \mid \ProtSI} = (p_1,\ldots,p_n)$ is defined in $\protPI$. Also, define $\rI$ and $\rJ$ as the random variables for indices
    $i$ and $j$ in $\protPI$. We claim,
    \begin{align}
        \Pr\paren{\protPI~\errs \mid x_1 \wedge y_1 = 1} = \Ex_{\ProtSI \sim \rProtSI\mid \rT=\rI} \bracket{\Pr\paren{\rI \compr \rJ \mid \rProtSI = \ProtSI,\rT=\rI}} \label{eq:disj-claim-correct}. 
    \end{align}
    This is by construction of the protocol as $x_1 \wedge y_1 = 1$ and $\rT=\rI$ are equivalent, and conditioned on $x_1 \wedge y_1 = 1$, the correct answer is the index $1$ which would be output by the protocol iff $i \comp j$. 

    For any fixed transcript $\ProtSI$, the bound in RHS of Eq~(\ref{eq:disj-claim-correct}) is only a function of the distribution of $(\rI,\rJ)$. Hence, let us examine 
    $\distribution{\rI,\rJ \mid \ProtSI,\rT=\rI} = \distribution{\rI \mid \ProtSI,\rT=\rI} \cdot \distribution{\rJ \mid \ProtSI,\rT=\rI=i}$. For any $\ell \in [n]$, we have,
    \begin{align}
        \Pr_{\distPI}\paren{\rI = \ell \mid \ProtSI,\rT=\rI} = \Pr_{\distSI}\paren{\textnormal{target element is $\ell$} \mid \ProtSI} = p_\ell. \label{eq:dist-I}
    \end{align}
    This is simply by Observation~\ref{obs:dist-PI-SI} that implies instances created in $\protPI$ are sampled from $\distSI$ and because we conditioned on $\rT = \rI$. 
    On the other hand, conditioned on $\rT=\rI=i$, for any $\ell \in [n] \setminus \set{i}$, 
    \begin{align}
        \Pr_{\distPI}\paren{\rJ = \ell \mid \ProtSI, \rT=\rI=i} = \Pr_{\distPI}\paren{\rJ = \ell \mid \rT=\rI=i} = \frac{1}{n-1}. \label{eq:dist-J}
    \end{align}
    This is by Observation~\ref{obs:dist-PI-SI-hidden} as $\ProtSI$ is only a function of $(A,B)$, while $\rJ$ is independent of $(\rA,\rB)$ (conditioned on $\rJ \neq \rT$) 
    and is uniform on any index which is not the target element.

    Now that we have determined the distribution of $(\rI,\rJ)$ (conditioned on $\ProtSI$ and $\rT=\rI$), our goal is to simply bound the RHS of Eq~(\ref{eq:disj-claim-correct}) (for any fixed choice of $\ProtSI$). 
    Intuitively, we should expect this quantity to be small as we are picking $\rI$ by gravitating towards higher rank numbers according to $\comp$, while $\rP_\rJ$ is chosen independent of $\comp$. We formalize this intuition in the following. 

    \begin{claim}\label{clm:tvd-calculation}
        Let $\delta := \tvd{\distribution{\rI \mid \ProtSI,\rT=\rI}}{\unif_{[n]}}$; then $\Pr\paren{\rI \compr \rJ \mid \ProtSI,\rT=\rI} \leq \frac{1}{2} - \Omega(\delta)$.
    \end{claim}
    \begin{proof}[Proof of Claim~\ref{clm:tvd-calculation}]
        In the following, all random variables are conditioned on $(\ProtSI,\rT=\rI)$ and hence with a slight abuse of notation we drop this conditioning throughout the proof. 
        Recall that $\distribution{\rI} = (p_1,\ldots,p_n)$ (by Eq~(\ref{eq:dist-I})) and without loss of generality assume $p_1 \leq p_2 \leq \ldots \leq p_n$ as we can always rename the indices to obtain this property (and breaking the ties as in the protocol $\protPI$ by the original index). As for the distribution of $\rJ$, note that for any $\ell \in [n]$, $\Pr\paren{\rJ \in [\ell+1,n] \mid \rI = \ell} = \frac{n-\ell}{n-1}$ by Eq~(\ref{eq:dist-J}). Note that after this renaming, $\rI \compr \rJ$ iff $\rI < \rJ$. 
        Hence, we have, 
        \begin{align*}
            \Pr\paren{\rI \compr \rJ} = \Pr\paren{\rI < \rJ} = \sum_{\ell=1}^n  \Pr\paren{\rI=\ell}\Pr(\rJ\in [\ell+1,n] \mid \rI=\ell) = \sum_{\ell=1}^n p_\ell \cdot \frac{n-\ell}{n-1}.
        \end{align*}

        Let $k \in [n]$ be the largest index such that $p_k < 1/n$. 
        Define $q := \sum_{\ell=1}^{k} p_\ell$ as the total probability mass of indices with probability less than $1/n$. We have,
        \begin{align} \label{equ:del_pL}
            \delta = \tvd{\rI}{\unif_{[n]}} = \frac{1}{2} \cdot \sum_{\ell=1}^{n} \card{p_\ell - \frac{1}{n}} = \frac{1}{2} \cdot \Paren{(\frac{k}{n}-q)+((1-q)-\frac{n-k}{n})} = \frac{k}{n}-q
        \end{align}
        which implies that $q = \frac{k}{n} - \delta$. By the equation above for $\Pr\paren{\rI < \rJ}$, we have, 
        \[
        \Pr\paren{\rI < \rJ}  = \sum_{\ell=1}^k p_\ell\cdot \frac{n-\ell}{n-1} + \sum_{\ell=k+1}^n p_\ell\cdot \frac{n-\ell}{n-1}.
        \]
        Now, using the assumption that $p_1\le p_2 \le \dots \le p_n$ and by the inequality of Proposition~\ref{prop:rearr},
        \begin{align*}
            \Pr\paren{\rI < \rJ} & \le \frac{1}{k}\sum_{\ell=1}^k p_\ell \sum_{\ell=1}^k \frac{n-\ell}{n-1} + \frac{1}{n-k}\sum_{\ell=k+1}^n p_l \sum_{\ell=k+1}^n \frac{n-\ell}{n-1} \\
                                 & = \frac{q}{k} \cdot \frac{k\cdot(2n-k-1)}{2n-2} + \frac{1-q}{n-k} \cdot \frac{(n-k-1)(n-k)}{2n-2} \\
                                 & = q \cdot \frac{2n-k-1}{2n-2} + (1-q) \cdot \frac{n-k-1}{2n-2} = \frac{n-k-1}{2n-2} + q \cdot \frac{n}{2n-2} \\
                                 & = \frac{1}{2} - \frac{k-n \cdot q}{2n-2}  \Eq{Eq~(\ref{equ:del_pL})} \frac{1}{2} - \frac{n\delta}{2n-2} < 1/2 - \delta/2,
        \end{align*}
        completing the proof. \Qed{Claim~\ref{clm:tvd-calculation}} 
        
    \end{proof}

    We are now ready to finalize the proof of Lemma~\ref{lem:protDisj-correctness}. 
    \begin{align*}
        \Pr\paren{\protPI~\errs \mid x_1 \wedge y_1 = 1} &\Eq{Eq~(\ref{eq:disj-claim-correct})} \Ex_{\ProtSI \sim \rProtSI\mid \rT=\rI} \bracket{\Pr\paren{\rI \compr \rJ \mid \rProtSI = \ProtSI,\rT=\rI}} \\
                                                          &\hspace{-0.2cm} \Leq{Claim~\ref{clm:tvd-calculation}}\Ex_{\ProtSI \sim \rProtSI\mid \rT=\rI} \bracket{\frac{1}{2} - \Omega\paren{\tvd{\distribution{\rI \mid \rProtSI = \ProtSI,\rT=\rI}}{\unif_{[n]}}}} \\
                                                             &\hspace{0.2cm} = \Ex_{\ProtSI \sim \rProtSI} \bracket{\frac{1}{2} - \Omega\paren{\tvd{\distribution{\rT \mid \rProtSI=\ProtSI}}{\unif_{[n]}}}} \tag{distribution of $\rI = \rT$ and $\rProtSI \perp \rT=\rI$} \\
                                                                &\hspace{0.2cm} \leq \frac{1}{2} - \Omega(\eps), 
    \end{align*}
    where the last inequality is because $\protSI$ $\eps$-solves $\SI$. We can also do the same exact analysis for the case when $x_2 \wedge y_2 = 1$, hence obtaining that
    $\Pr\paren{\protPI~\errs} = \frac{1}{2} - \Omega(\eps)$. \Qed{Lemma~\ref{lem:protDisj-correctness}}

\end{proof}

\subsection*{Proof of Lemma~\ref{lem:protDisj-info-cost}: Information Cost of Protocol $\protPI$}

We prove this lemma by a direct-sum type argument that shows if the (internal) information cost of $\protSI$ is small, then protocol $\protPI$ is revealing a small information about its input \emph{assuming conditioning on the target element}. 
We emphasize that this information revealed 
is \emph{not} equivalent with the internal information cost as we are conditioning on some information not known to neither Alice nor Bob. 
The following is a restatement of Lemma~\ref{lem:protDisj-info-cost} that we prove in this section. 
\begin{lemma*}[Restatement of Lemma~\ref{lem:protDisj-info-cost}]
    Let $\rProtPI$ denote the random variable for the transcript of the protocol $\protPI$ and $\rK$ be the random variable for index $k$ in distribution $\distPI$. We have, 
    \begin{align*}
        \mii{\rX_1,\rX_2}{\rProtPI \mid \rY_1,\rY_2, \rK}{\distPI} + \mii{\rY_1,\rY_2}{\rProtPI \mid \rX_1,\rX_2, \rK}{\distPI} &\leq \frac{1}{n-1} \cdot \ICost{\protSI}{\distSI}.
    \end{align*}
\end{lemma*}

The intuition behind the proof is as follows. The LHS in Lemma~\ref{lem:protDisj-info-cost} is the information revealed about the input of players (in $\PI$) averaged over choices of $k=1$ and $k=2$.
Let us assume $k=1$, as the other case is symmetric. In this case, this quantity is simply the information revealed about $(x_2,y_2)$ by the protocol as $(x_1,y_1) = (1,1)$ and hence has $0$ information (once we have conditioned on the event $k=1$). 
However, when $k=1$, $(x_2,y_2)$ is embedded in index $j$, i.e., $(x_2,y_2) = (a_j,b_j)$ and have the same distribution as all other coordinates in $A_{-i},B_{-i}$. 
As such, since the protocol $\protSI$ called inside $\protPI$ is oblivious to the choice of $j$, the information revealed about $(a_j,b_j)$ in average is smaller than the information revealed by $\protSI$ about $A_{-i},B_{-i}$ 
(which itself is at most the internal information cost of $\protSI$), by a factor of $n-1$ (i.e., the number of coordinates in $[n] \setminus \set{i}$ we are averaging over).

The outline above oversimplifies many details. One such detail is the way of ensuring a ``symmetric treatment'' of both indices $i$ and $j$
through the rather peculiar choice of public-private sampling in $\protPI$ (via the choices of $\ell$ and $S$). This is crucial 
for the above argument to work for both $k=1$ and $k=2$ cases simultaneously, without the players knowing which index the ``averaging'' of information is being done for (index $j$ in the context of the discussion above). 

\begin{proof}[Proof of Lemma~\ref{lem:protDisj-info-cost}]
    For simplicity of exposition, we drop the subscript $\distPI$ from all mutual information terms with the understanding that all random variables
    are distributed according to $\distPI$ (and the randomness of protocol $\protPI$ on $\distPI$) unless explicitly stated otherwise. 

    We bound the first term in LHS above (the second term can be bounded the same way). By expanding the conditional mutual information term we have, 
    \begin{align}
        \mi{\rX_1,\rX_2}{\rProtPI \mid \rY_1,\rY_2, \rK}{} &= \frac{1}{2} \cdot \mi{\rX_1,\rX_2}{\rProtPI \mid \rY_1,\rY_2, \rK=1}{} \notag \\
                                                                 &\hspace{3cm}+ \frac{1}{2} \cdot \mi{\rX_1,\rX_2}{\rProtPI \mid \rY_1,\rY_2, \rK=2}{} \label{eq:disj-info-K}.
    \end{align}
    We now focus on the first term in the LHS of Eq~(\ref{eq:disj-info-K}). We have, 
    \begin{align*}
        \mi{\rX_1,\rX_2}{\rProtPI \mid \rY_1,\rY_2, \rK=1} &= \mi{\rX_2}{\rProtPI \mid \rY_2, \rK=1} \tag{$(\rX_1,\rY_1)$ is always equal to $(1,1)$ in $\distPI$ conditioned on $\rK=1$} \\
                                                                   &= \mi{\rX_2}{\rProtSI \mid \rY_2, \rI,\rJ,\rS,\rL,\rA_{\rS},\rB_{\rbarS}, \rK=1} \tag{$\protPI$ runs $\protSI$ with public randomness $\rI,\rJ,\rS,\rL,\rA_{\rS},\rB_{\rbarS}$ ($\rL$ is for $\ell$) and by Proposition~\ref{prop:public-random}} \\
                                                                         &= \sum_{i \neq j} \frac{1}{n (n-1)} \cdot \mi{\rA_j}{\rProtSI \mid \rB_j, \rL,\rS,\rA_{\rS},\rB_{\rbarS},\rI=i,\rJ=j, \rK=1}  \tag{$(\rX_2,\rY_2)$ is embedded in $(\rA_j,\rB_j)$ conditioned on $\rJ=j$}.
    \end{align*}

    Recall that $\rT$ denotes the unique index in $[n]$ in instances $(A,B) \sim \distSI$ which is equal to $(1,1)$. Note that $\rT = i$ conditioned on $\rI = i$ and $\rK=1$, and that conditioning
    on the event $\rT = i$ has the same effect on all random variables above as conditioning on the joint event $\rI = i,\rK=1$. Hence, we can write the RHS above as, 
    \begin{align*}
        \mi{\rX_1,\rX_2}{\rProtPI \mid \rY_1,\rY_2, \rK=1}{} &= \frac{1}{n (n-1)} \sum_{i \neq j}  \mi{\rA_j}{\rProtSI \mid \rB_j, \rL,\rS,\rA_{\rS},\rB_{\rbarS},\rT=i,\rJ=j}{} \\
                                                                     &\leq \frac{1}{n (n-1)} \sum_{i \neq j}  \mi{\rA_j}{\rProtSI \mid \rL,\rS,\rA_{\rS},\rB_{-i},\rT=i,\rJ=j}{}. \tag{as $\rA_j \perp \rB_{-i} \mid \rB_j$ (and other variables above) and hence we can apply Proposition~\ref{prop:info-increase}} 
    \end{align*}
    By further expanding the conditional mutual information term in RHS over $\rL$ and $\rS$, 
    \begin{align}
        &\mi{\rX_1,\rX_2}{\rProtPI \mid \rY_1,\rY_2, \rK=1}{} \notag \\ 
              &\hspace{0.5cm}\leq 
        \frac{1}{n(n-1)} \sum_{i=1}^{n} \sum_{\substack{j=1 \\ j\neq i}}^{n} \sum_{\ell = 0}^{n-2} \sum_{\substack{S \subseteq [n] \setminus \set{i,j} \\ \card{S} = \ell}} \frac{1}{n-1}  {{n-2}\choose{\ell}}^{-1} \cdot
        \mi{\rA_j}{\rProtSI \mid \rA_{S},\rB_{-i},\rL = \ell, \rS=S,\rT=i,\rJ=j}{} \notag \\
        &\hspace{0.5cm}= 
        \frac{1}{n(n-1) \cdot (n-1)!} \sum_{i=1}^{n} \sum_{\substack{j=1 \\ j\neq i}}^{n} \sum_{\ell = 0}^{n-2} \sum_{\substack{S \subseteq [n] \setminus \set{i,j} \\ \card{S} = \ell}} \paren{(n-2-\ell)!\ell!} \cdot
        \mi{\rA_j}{\rProtSI \mid \rA_{S},\rB_{-i},\rT=i}{}, \label{eq:apply-clm-RHS}
    \end{align}
    by reorganization of the terms and dropping the conditioning on events $\rL = \ell,\rS=S,\rJ=j$ as the distribution of remaining random variables are independent of these events. We now have the following
    auxiliary claim. 
    \begin{claim}\label{clm:aux-ds-sum}
        For any choice of  $i \in [n]$, 
        \begin{align*}
            &\sum_{\substack{j=1 \\ j\neq i}}^{n} \sum_{\ell = 0}^{n-2} \sum_{\substack{S \subseteq [n] \setminus \set{i,j} \\ \card{S} = \ell}} \paren{(n-2-\ell)!\ell!} \cdot
            \mi{\rA_j}{\rProtSI \mid \rA_{S},\rB_{-i},\rT=i}{} \\
            &\hspace{1cm} = \sum_{\sigma \in \mathcal{S}_{-i}} \sum_{\ell=0}^{n-2}\mi{\rA_{\sigma(\ell+1)}}{\rProtSI \mid \rA_{\sigma(<\ell+1)},\rB_{-i},\rT=i}{},
        \end{align*}
        where $\mathcal{S}_{-i}$ is the set of all permutations of $[n] \setminus \set{i}$.
    \end{claim}
    \begin{proof}
        Fix any $(j,S)$ in the LHS. For integer $\ell = \card{S}$, there are exactly $\paren{(n-2-\ell)!\ell!}$ permutations $\sigma \in \mathcal{S}_{-i}$ such that $(i)$ $\sigma(\ell+1) = j$ and $(ii)$ $\set{\sigma(1),\ldots,\sigma(\ell)} = S$. 
        Hence, $\mi{\rA_j}{\rProtSI \mid \rA_{S},\rB_{-i},\rT=i}{}$ for $(j,S)$ appears exactly $\paren{(n-2-\ell)!\ell!}$ times in RHS as $\mi{\rA_{\sigma(\ell+1)}}{\rProtSI \mid \rA_{\sigma(<\ell+1)},\rB_{-i},\rT=i}{}$ (for appropriate choices of $\sigma$ as described above),
        proving the claim. \Qed{Claim~\ref{clm:aux-ds-sum}}

    \end{proof}

    \noindent
    By applying Claim~\ref{clm:aux-ds-sum} to the RHS of Eq~(\ref{eq:apply-clm-RHS}), we obtain that, 
    \begin{align*}
        \mi{\rX_1,\rX_2}{\rProtPI \mid \rY_1,\rY_2, \rK=1}{} &\leq
        \frac{1}{n(n-1)  (n-1)!} \sum_{i=1}^{n} \sum_{\sigma \in \mathcal{S}_{-i}} \sum_{\ell=0}^{n-2}\mi{\rA_{\sigma(\ell+1)}}{\rProtSI \mid \rA_{\sigma(<\ell+1)},\rB_{-i},\rT=i}{} \\
        &= \frac{1}{n(n-1)  (n-1)!} \sum_{i=1}^{n} \sum_{\sigma \in \mathcal{S}_{-i}} \mi{\rA_{-i}}{\rProtSI \mid \rB_{-i},\rT=i} \tag{by chain rule of mutual information in \itfacts{chain-rule}} \\ 
              &= \frac{1}{n(n-1)} \sum_{i=1}^{n} \mi{\rA_{-i}}{\rProtSI \mid \rB_{-i},\rT=i} \tag{as $\card{\mathcal{S}_{-i}} = (n-1)!$} \\
                    &= \frac{1}{n-1} \cdot \mi{\rA}{\rProtSI \mid \rB,\rT} \tag{as $(\rA_\rT,\rB_\rT) = (1,1)$ in $\distPI$ and hence we can add them to the information term} \\
                          &\leq \frac{1}{n-1} \cdot \mi{\rA}{\rProtSI \mid \rB} = \frac{1}{n-1} \cdot \mii{\rA}{\rProtSI \mid \rB}{\distSI},
    \end{align*}
    where the last inequality is because $\rProtSI \perp \rT \mid \rA,\rB$ (as the transcript is only a function of the inputs) and hence we can apply Proposition~\ref{prop:info-decrease}, and the last equality
    is because by Observation~\ref{obs:dist-PI-SI}, joint distribution of $\distPI$ and randomness of the protocol $\protPI$ is the same as distribution $\distSI$. Using the same exact analysis (by switching the role of indices $i$ and $j$ and noting that
    the rest is all symmetric), we also obtain the following bound for the second term of Eq~(\ref{eq:disj-info-K}),
    \begin{align*}
        \mi{\rX_1,\rX_2}{\rProtPI \mid \rY_1,\rY_2, \rK=2}{} \leq \frac{1}{n-1} \cdot \mii{\rA}{\rProtSI \mid \rB}{\distSI}.
    \end{align*}
    Plugging in these bounds in Eq~(\ref{eq:disj-info-K}), we obtain that,
    \begin{align}
        \mii{\rX_1,\rX_2}{\rProtPI \mid \rY_1,\rY_2, \rK}{\distPI}  \leq \frac{1}{n-1} \cdot \mii{\rA}{\rProtSI \mid \rB}{\distSI}. \label{eq:sum-1}
    \end{align}

    Similarly, the second term in the LHS of Lemma~\ref{lem:protDisj-info-cost} can be upper bounded using a similar analysis (by switching the role of $\rA$ and $\rB$, and $\rS$ and $\rbarS$ and noting that the rest is all symmetric), implying
    the following bound:
    \begin{align}
        \mii{\rY_1,\rY_2}{\rProtPI \mid \rX_1,\rX_2, \rK}{\distPI}  \leq \frac{1}{n-1} \cdot \mii{\rB}{\rProtSI \mid \rA}{\distSI}. \label{eq:sum-2}
    \end{align}	
    Summing up the LHS and RHS in Eq~(\ref{eq:sum-1}) and Eq~(\ref{eq:sum-2}), finalizes the proof. \Qed{Lemma~\ref{lem:protDisj-info-cost}}

\end{proof}

\subsection*{Proof of Lemma~\ref{lem:PI-ic}: Information Complexity of $\PI$}

We now prove the final step of the proof of Theorem~\ref{thm:SI-tvd}. The following is a restatement of Lemma~\ref{lem:PI-ic}. 

\begin{lemma*}[Restatement of Lemma~\ref{lem:PI-ic}]
    Suppose $\protPI$ outputs the correct answer on $\distPI$ with probability at least $\frac{1}{2} + \Omega(\eps)$. Then, 
    \begin{align*}
        \mii{\rX_1,\rX_2}{\rProtPI \mid \rY_1,\rY_2, \rK}{\distPI} + \mii{\rY_1,\rY_2}{\rProtPI \mid \rX_1,\rX_2, \rK}{\distPI} = \Omega(\eps^2).
    \end{align*}
\end{lemma*}

The idea behind the proof of Lemma~\ref{lem:PI-ic} is as follows. Recall that $\Prot_{\disjin{x_1}{x_2}{y_1}{y_2}}$ denotes the transcript of the protocol condition on the input being $\disjin{x_1}{x_2}{y_1}{y_2}$. 
Suppose towards the contradiction that the LHS of Lemma~\ref{lem:PI-ic} is $o(\eps^2)$ instead. By focusing on the conditional terms when $k=1$, we can
show that distribution of $\Prot_{\disjin{1}{x'_2}{1}{y'_2}}$ and $\Prot_{\disjin{1}{x''_2}{1}{y''_2}}$ for all choices of $(x'_2,y'_2)$ and $(x''_2,y''_2)$ in the support of $\distPI$ (basically everything except for $(1,1)$) 
are quite close. This is intuitively because the information revealed about $(x_2,y_2)$ by $\protPI$ conditioned on $k=1$ is small. Similarly, by focusing on the $k=2$ terms, we obtain the same result for 
$\Prot_{\disjin{x'_2}{1}{y'_2}{1}}$ and $\Prot_{\disjin{x''_2}{1}{y''_2}{1}}$. 

Up until this point, there is no contradiction as the answer to $\disjin{1}{*}{1}{*}$ is always $1$ and hence there is no problem with the corresponding transcripts 
in $\Prot_{\disjin{1}{*}{1}{*}}$ to be similar (similarly for $\Prot_{\disjin{*}{1}{*}{1}}$ separately). However, we combine the previous part with the cut-and-paste property of randomized protocols (Fact~\ref{fact:r-rectangle})
to argue that in fact the distribution of $\Prot_{\disjin{1}{0}{1}{0}}$ and $\Prot_{\disjin{0}{1}{0}{1}}$ are also similar. This then basically implies that $\Prot_{\disjin{1}{*}{1}{*}}$ essentially has the same distribution as 
$\Prot_{\disjin{*}{1}{*}{1}}$; but then this is a contradiction as the answer to the protocol (which is only a function of the transcript) needs to be different between these two types of inputs. We now formalize the proof (a schematic organization of the proof is provided in Appendix~\ref{app:schematic}).

\begin{proof}[Proof of Lemma~\ref{lem:PI-ic}]
    The distribution of random variables below is always $\distPI$ (and the randomness of the protocol $\protPI$ on $\distPI$) and 
    hence  we drop the subscript $\distPI$ from all mutual information terms. Suppose towards a contradiction that the LHS in the lemma statement is $o(\eps^2)$. 
    As we showed in Eq~(\ref{eq:disj-info-K}) and the subsequent equation in the proof of Lemma~\ref{lem:protDisj-info-cost}, the LHS can be written as
    \begin{align}
        &\frac{1}{2} \cdot \paren{\mi{\rX_2}{\rProtPI \mid \rY_2, \rK=1}{} + \mii{\rY_2}{\rProtPI \mid \rX_2, \rK=1}{}} \notag \\
              &\hspace{1.5cm} + \frac{1}{2} \cdot \paren{\mi{\rX_2}{\rProtPI \mid \rY_2, \rK=1}{} + \mii{\rY_2}{\rProtPI  \mid \rX_2, \rK=1}{}} = o(\eps^2) \label{eq:each-term-cont}.
    \end{align}
    By bounding each of the above term above separately by $o(\eps^2)$ and expanding the mutual information terms, we prove the following claim.
    \begin{claim}\label{clm:disj2-info-terms}
        Assuming Eq~(\ref{eq:each-term-cont}),
        \begin{align*}
            (1) ~ \mi{\rX_2}{\rProtPI \mid \rY_2=0, \rK=1} = o(\eps^2), \qquad (2) ~ \mi{\rY_2}{\rProtPI \mid \rX_2=0, \rK=1} = o(\eps^2), \\
            (3) ~ \mi{\rX_1}{\rProtPI \mid \rY_1=0, \rK=2} = o(\eps^2), \qquad (4) ~ \mi{\rY_1}{\rProtPI \mid \rX_1=0, \rK=2} = o(\eps^2).
        \end{align*}
    \end{claim} 
    \begin{proof}
        To prove the first equation, we write the first term in Eq~(\ref{eq:each-term-cont}) as follows: 
        \begin{align*}
            \mi{\rX_2}{\rProtPI \mid \rY_2, \rK=1} &= \frac{2}{3} \cdot \mi{\rX_2}{\rProtPI \mid \rY_2 = 0, \rK=1} + \frac{1}{3} \cdot \mi{\rX_2}{\rProtPI \mid \rY_2 = 1, \rK=1} \\
                                                      &=  \frac{2}{3} \cdot \mi{\rX_2}{\rProtPI \mid \rY_2 = 0, \rK=1},
        \end{align*}
        since for $(\rX_2,\rY_2) \sim \distPI \mid \rK=1$, if $\rY_2 = 1$, then $\rX_2$ is always equal to $0$ and hence the second term above is zero. As the LHS of above equation is $o(\eps^2)$ by
        Eq~(\ref{eq:each-term-cont}) (and non-negativity of mutual information in~\itfacts{info-zero}), we obtain the first equation in the statement of the claim. The remaining equations can be proven exactly the same. 
        \Qed{Claim~\ref{clm:disj2-info-terms}}

    \end{proof}

    We now use Claim~\ref{clm:disj2-info-terms}, to bound the distance between different transcripts of the protocol. Recall that $\rProt_{\disjin{x_1}{x_2}{y_1}{y_2}}$ denotes the transcript of the protocol 
    conditioned on the input $(x_1,x_2)$ to Alice, and $(y_1,y_2)$ to Bob. 
    \begin{claim}\label{clm:disj2-h-terms-1}
        Assuming Eq~(\ref{eq:each-term-cont}),
        \begin{align*}
            (1) ~ \hdt{\rProt_{\disjin{1}{1}{1}{0}}}{\rProt_{\disjin{1}{0}{1}{0}}} = o(\eps^2), \qquad (2) ~ \hdt{\rProt_{\disjin{1}{0}{1}{1}}}{\rProt_{\disjin{1}{0}{1}{0}}} = o(\eps^2), \\
            (3) ~ \hdt{\rProt_{\disjin{1}{1}{0}{1}}}{\rProt_{\disjin{0}{1}{0}{1}}} = o(\eps^2), \qquad (4) ~ \hdt{\rProt_{\disjin{0}{1}{1}{1}}}{\rProt_{\disjin{0}{1}{0}{1}}} = o(\eps^2).  
        \end{align*}
    \end{claim} 
    \begin{proof}
        We write the LHS of the first equation in Claim~\ref{clm:disj2-info-terms} in terms of the KL-divergence using Fact~\ref{fact:kl-info}. Define $\rProt_{\disjin{1}{*}{1}{0}}$ as the distribution of $\rProt$ conditioned on the given value for $x_1,y_1,y_2$ (leaving out 
        the assignment for $x_2$). We have,
        \begin{align*}
            \mi{\rX_2}{\rProtPI \mid \rY_2=0, \rK=1} &\Eq{Fact~\ref{fact:kl-info}} \Ex_{x_2 \sim \rX_2 \mid \rY_2=0,\rK=1}[{\kl{\rProt_{\disjin{1}{x_2}{1}{0}}}{\rProt_{\disjin{1}{*}{1}{0}}}}] \\
                                                           &\hspace{10pt}= \frac{1}{2} \cdot \kl{\rProt_{\disjin{1}{0}{1}{0}}}{\rProt_{\disjin{1}{*}{1}{0}}} + \frac{1}{2} \cdot \kl{\rProt_{\disjin{1}{1}{1}{0}}}{\rProt_{\disjin{1}{*}{1}{0}}} \\
                                                                 &\Geq{Fact~\ref{fact:hellinger-kl}} \hdt{\rProt_{\disjin{1}{0}{1}{0}}}{\rProt_{\disjin{1}{1}{1}{0}}}.
        \end{align*}
        The distribution of $\rX_2$ conditioned on $\rY_2 = 0,\rK=1$ in $\distPI$ 
        is uniform over $\set{0,1}$ (hence the second equality). As such, $\rProt_{\disjin{1}{*}{1}{0}} = \frac{1}{2} \cdot \paren{\rProt_{\disjin{1}{0}{1}{0}}+\rProt_{\disjin{1}{1}{1}{0}}}$ and so we can apply Fact~\ref{fact:hellinger-kl} to obtain the last inequality.  
        As $\mi{\rX_2}{\rProtPI \mid \rY_2=0, \rK=1} = o(\eps^2)$ by Claim~\ref{clm:disj2-info-terms}, we obtain the first equation (note that h is symmetric). The remaining equations can be proven similarly. 
        \Qed{Claim~\ref{clm:disj2-h-terms-1}}

    \end{proof}

    The next step is to use the cut-and-paste property (Fact~\ref{fact:r-rectangle}) of randomized protocols to prove the following claim. 
    \begin{claim}\label{clm:disj2-h-terms-2}
        Assuming Eq~(\ref{eq:each-term-cont}), $\hdt{\rProt_{\disjin{1}{0}{1}{0}}}{\rProt_{\disjin{0}{1}{0}{1}}} = o(\eps^2)$. 
    \end{claim}
    \begin{proof}
        We start with proving the following two equations first:
        \begin{align*}
            (1)~ \hdt{\rProt_{\disjin{1}{1}{1}{1}}}{\rProt_{\disjin{1}{0}{1}{0}}} = o(\eps^2), \qquad (2)~ \hdt{\rProt_{\disjin{1}{1}{1}{1}}}{\rProt_{\disjin{0}{1}{0}{1}}} = o(\eps^2). 
        \end{align*}
        For the first equation, 
        \begin{align*}
            \hdt{\rProt_{\disjin{1}{1}{1}{1}}}{\rProt_{\disjin{1}{0}{1}{0}}} &=  \hdt{\rProt_{\disjin{1}{1}{1}{0}}}{\rProt_{\disjin{1}{0}{1}{1}}} \tag{by the cut-and-paste property in Fact~\ref{fact:r-rectangle}} \\
                                                                                   &\leq \paren{\hd{\rProt_{\disjin{1}{1}{1}{0}}}{\rProt_{\disjin{1}{0}{1}{0}}} + \hd{\rProt_{\disjin{1}{0}{1}{0}}}{\rProt_{\disjin{1}{0}{1}{1}}}}^2 \tag{by triangle inequality} \\
                                                                                         &\leq 2\cdot \paren{\hdt{\rProt_{\disjin{1}{1}{1}{0}}}{\rProt_{\disjin{1}{0}{1}{0}}} + \hdt{\rProt_{\disjin{1}{0}{1}{0}}}{\rProt_{\disjin{1}{0}{1}{1}}}} \tag{by Cauchy-Schwartz} \\
                                                                                               &= o(\eps^2) \tag{by parts (1) and (2) of Claim~\ref{clm:disj2-h-terms-1}}.
        \end{align*}
        The second equation can be proven similarly using parts (3) and (4) of Claim~\ref{clm:disj2-h-terms-1}. We can now prove the claim as follows:
        \begin{align*}
            \hdt{\rProt_{\disjin{1}{0}{1}{0}}}{\rProt_{\disjin{0}{1}{0}{1}}} &\leq \paren{\hd{\rProt_{\disjin{1}{0}{1}{0}}}{\rProt_{\disjin{1}{1}{1}{1}}} + \hd{\rProt_{\disjin{1}{1}{1}{1}}}{\rProt_{\disjin{0}{1}{0}{1}}}}^2 \tag{by triangle inequality} \\
                                                                                   &\leq 2\cdot \paren{\hdt{\rProt_{\disjin{1}{0}{1}{0}}}{\rProt_{\disjin{1}{1}{1}{1}}} + \hdt{\rProt_{\disjin{1}{1}{1}{1}}}{\rProt_{\disjin{0}{1}{0}{1}}}} \tag{by Cauchy-Schwartz} \\
                                                                                         &= o(\eps^2). \tag{by part (1) and (2) of the equation above}
        \end{align*}
        This concludes the proof. \Qed{Claim~\ref{clm:disj2-h-terms-2}}

    \end{proof}
    Define $I_1 := \set{\disjin{1}{0}{1}{0}, \disjin{1}{1}{1}{0}, \disjin{1}{0}{1}{1}}$ and $I_2 := \set{\disjin{0}{1}{0}{1}, \disjin{1}{1}{0}{1}, \disjin{0}{1}{1}{1}}.$
    The tuples in $I_1 \cup I_2$ partition all the input tuples in the support of $\distPI$ and moreover, for every tuple in $I_1$, the correct answer to $\PI$ is the first index, 
    while for every tuple in $I_2$, the correct answer is the second index. We now bound the total variation distance between every pair of tuples in $I_1$ and $I_2$. 

    \begin{claim}\label{clm:disj2-sep-terms}
        Assuming Eq~(\ref{eq:each-term-cont}), for every $(T_1,T_2) \in I_1 \times I_2$, $\tvd{\rProt_{T_1}}{\rProt_{T_2}} = o(\eps)$. 
    \end{claim}
    \begin{proof}
        Proving the claim amounts to proving the following nine equations: 
        \begin{align*}
            (1)~ \tvd{\rProt_{\disjin{1}{0}{1}{0}}}{\rProt_{\disjin{0}{1}{0}{1}}} = o(\eps), ~ (2)~ \tvd{\rProt_{\disjin{1}{0}{1}{0}}}{\rProt_{\disjin{1}{1}{0}{1}}} = o(\eps), ~ (3)~\tvd{\rProt_{\disjin{1}{0}{1}{0}}}{\rProt_{\disjin{0}{1}{1}{1}}} = o(\eps), \\
            (4)~ \tvd{\rProt_{\disjin{1}{1}{1}{0}}}{\rProt_{\disjin{0}{1}{0}{1}}} = o(\eps), ~ (5)~ \tvd{\rProt_{\disjin{1}{1}{1}{0}}}{\rProt_{\disjin{1}{1}{0}{1}}} = o(\eps), ~ (6)~\tvd{\rProt_{\disjin{1}{1}{1}{0}}}{\rProt_{\disjin{0}{1}{1}{1}}} = o(\eps), \\
            (7)~ \tvd{\rProt_{\disjin{1}{0}{1}{1}}}{\rProt_{\disjin{0}{1}{0}{1}}} = o(\eps), ~ (8)~ \tvd{\rProt_{\disjin{1}{0}{1}{1}}}{\rProt_{\disjin{1}{1}{0}{1}}} = o(\eps), ~ (9)~\tvd{\rProt_{\disjin{1}{0}{1}{1}}}{\rProt_{\disjin{0}{1}{1}{1}}} = o(\eps),
        \end{align*}
        The first equation can be proven as follows: 
        \begin{align*}
            \tvd{\rProt_{\disjin{1}{0}{1}{0}}}{\rProt_{\disjin{0}{1}{0}{1}}} \leq \sqrt{2} \cdot \hd{\rProt_{\disjin{1}{0}{1}{0}}}{\rProt_{\disjin{0}{1}{0}{1}}} = o(\eps), 
        \end{align*}
        where the  inequality is by Fact~\ref{fact:hellinger-tvd} and the equality is by Claim~\ref{clm:disj2-h-terms-2}. This proves the equation (1) above. 
        Now note that,
        \begin{align*}
            \tvd{\rProt_{\disjin{1}{0}{1}{0}}}{\rProt_{\disjin{1}{1}{0}{1}}} &\leq \tvd{\rProt_{\disjin{1}{0}{1}{0}}}{\rProt_{\disjin{0}{1}{0}{1}}} + \tvd{\rProt_{\disjin{0}{1}{0}{1}}}{\rProt_{\disjin{1}{1}{0}{1}}} \tag{by triangle inequality} \\
                                                                                   &\leq o(\eps) + \sqrt{2} \cdot \hd{\rProt_{\disjin{0}{1}{0}{1}}}{\rProt_{\disjin{1}{1}{0}{1}}} \tag{by equation (1) above for the first term and Fact~\ref{fact:hellinger-tvd} for the second} \\
                                                                                         &= o(\eps). \tag{by part (3) of Claim~\ref{clm:disj2-h-terms-1}}
        \end{align*}

        This proves the equation (2). All the remaining equations can now be proven using a similar argument as above by first relating the distance between the two variables to the distance between
        $\tvd{\rProt_{\disjin{1}{0}{1}{0}}}{\rProt_{\disjin{1}{1}{0}{1}}}$ (which we know is $o(\eps)$ by equation (1)) using triangle inequality, and then use Fact~\ref{fact:hellinger-tvd} combined with Claim~\ref{clm:disj2-h-terms-1} to bound
        each of the remaining terms with $o(\eps)$. \Qed{Claim~\ref{clm:disj2-sep-terms}}

    \end{proof}

    We are now almost done. By Claim~\ref{clm:disj2-sep-terms}, if we assume Eq~(\ref{eq:each-term-cont}), then for every $(T_1,T_2) \in I_1 \times I_2$, $\tvd{\rProt_{T_1}}{\rProt_{T_2}} = o(\eps)$. On the other hand, 
    for $\protPI$ to be able to output the correct answer with probability $1/2+\Omega(\eps)$ (over the randomness of the protocol and the distribution), for at least one pair $(T_1,T_2) \in I_1 \times I_2$, we should have
    $\tvd{\rProt_{T_1}}{\rProt_{T_2}} = \Omega(\eps)$ as the output of the protocol on $T_1$ (resp. $T_2$) is only a function of $\rProt_{T_1}$ (resp. $\rProt_{T_2}$), and hence otherwise would be the same with probability $1-o(\eps)$ 
    by Fact~\ref{fact:tvd-small}. This implies that assuming Eq~(\ref{eq:each-term-cont}), the protocol errs with probability at least $1/2-o(\eps)$, which is a contradiction. Hence Eq~(\ref{eq:each-term-cont}) cannot hold
    \Qed{Lemma~\ref{lem:PI-ic}}

\end{proof}

\section{The Hidden-Pointer Chasing Problem}\label{sec:hpc}

Recall that the {hidden-pointer chasing} (\HPC) problem is a \emph{four-party} communication problem with players $P_A,P_B,P_C$, and $P_D$ defined as follows. 
Let $\XX := \set{x_1,\ldots,x_n}$ and $\YY := \set{y_1,\ldots,y_n}$ be two disjoint universes of size $n$ each. We define \HPC as follows: 

\begin{enumerate}[leftmargin=25pt]
	\item For any $x \in \XX$, $P_A$ and $P_B$ are given an instance $(A_x,B_x)$ of $\SI$ over the universe $\YY$ where $A_x \cap B_x = \set{t_x}$ for a single target element $t_x \in \YY$. We define $\bA := \set{A_{x_1},\ldots,A_{x_n}}$
	and $\bB := \set{B_{x_1},\ldots,B_{x_n}}$ as the whole input to $P_A$ and $P_B$, respectively. 
	\item For any $y \in \YY$, $P_C$ and $P_D$ are given an instance $(C_y,D_y)$ of $\SI$ over the universe $\XX$ where $C_y \cap D_y = \set{t_y}$ for a single target element $t_y \in \XX$.  We define $\bC := \set{C_{y_1},\ldots,C_{y_n}}$
	and $\bD := \set{D_{y_1},\ldots,D_{y_n}}$ as the whole input to $P_C$ and $P_D$, respectively.

	\item We define two mappings $f_{AB} : \XX \rightarrow \YY$ and $f_{CD}: \YY \rightarrow \XX$ such that:
	\begin{enumerate}
		\item for any $x \in \XX$, $f_{AB}(x) = t_x \in \YY$ in the instance $(A_x,B_x)$ of $\SI$. 
		\item for any $y \in \YY$, $f_{CD}({y}) = t_y \in \XX$ in the instance $(C_y,D_y)$ of $\SI$.
	\end{enumerate}
	
	\item Let $x_1 \in \XX$ be an arbitrary fixed element of $\XX$ known to all players. The pointers $z_0,z_1,z_2,z_3,\ldots$ are defined inductively as follows: 
	\begin{align*}
		z_0 := x_1, \qquad z_1 := f_{AB}(z_0), \qquad z_2:= f_{CD}(z_1), \qquad z_3:= f_{AB}(z_2), \qquad \ldots.
	\end{align*}
\end{enumerate}

For any integer $k \geq 1$, the $k$-step hidden-pointer chasing problem, denoted by \HPCk is defined as the communication problem of finding the pointer $z_k$. 
See Figure~\ref{fig:hpc} on page~\pageref{fig:hpc} for an illustration.

\subsection{Communication Complexity of \HPCk}\label{sec:cc-hpc}

It is easy to see that in $k+1$ phases, we can compute \HPCk with $O(k \cdot n)$ total communication: we simply skip the first phase; in the second phase, $P_A$ and $P_B$ solve the \SI instance
$(A_{z_0},B_{z_0})$ with $O(n)$ communication to compute $z_1 = f_{AB}(z_0)$ and send this pointer to $P_C$ and $P_D$; $P_C$ and $P_D$ in the next phase compute $f_{CD}(z_1)$ and the players continue like this to find the pointer $z_k$, which takes
$k+1$ phases in total. 

In the following, we prove that if we only have $k$ phases however, solving $\HPCk$ requires $\Omega(n^2/k^2 + n)$ bits of communication. 

\begin{theorem}\label{thm:hpck}
	For any integer $k \geq 1$, any $k$-phase protocol that outputs the correct solution to $\HPCk$ with constant probability requires $\Omega(n^2/k^2 + n)$ bits of communication. 
\end{theorem}

The rest of this section is devoted to the proof of Theorem~\ref{thm:hpck}. We start with defining our hard distribution of instances for $\HPCk$ and then use this distribution to
prove the lower bound.  

\subsection*{A Hard Distribution for $\HPC$}
The hard distribution for $\HPC$ is simply the product of distribution $\distSI$ for every $x \in \XX$ and $y \in \YY$. 

\begin{tbox}
    \textbf{Distribution $\distHPC$} on tuples $(\bA,\bB,\bC,\bD)$ from the universes $\XX$ and $\YY$: 

    \begin{enumerate}
		\item For any $x \in \XX$, sample $(A_x,B_x) \sim \distSI$ from the universe $\YY$ \emph{independently}. 
		\item For any $y \in \YY$, sample $(C_y,D_y) \sim \distSI$ from the universe $\XX$ \emph{independently}. 
    \end{enumerate}
\end{tbox}

The following simple observation is in order. 

\begin{observation}\label{obs:distHPC}
	Distribution $\distHPC$ is \underline{not} a product distribution. However, in this distribution: 
	\begin{enumerate}[label=(\roman*)]
		\item The inputs to $P_A$ and $P_B$ are {independent} of the inputs to $P_C$ and $P_D$, i.e., $(\bA,\bB) \perp (\bC,\bD)$. 
		\item For any $x \in \XX$, $(A_x,B_x)$ is independent of all other $(A_{x'},B_{x'})$ for $x'\neq x \in \XX$. Similarly for all $y,y' \in \YY$ and $(C_y,D_y)$ and $(C_{y'},D_{y'})$. 
	\end{enumerate}
\end{observation}

Based on this observation, we also have the following simple property. 

\begin{proposition}\label{prop:hpc-rectangle}
	Let $\protHPC$ be any deterministic protocol for $\HPCk$ on $\distHPC$. Then, for any transcript $\Prot$ of $\protHPC$, $(\bA,\bB) \perp (\bC,\bD) \mid \rProt = \Prot$. 
\end{proposition}
\begin{proof}
	Follows from the rectangle property of the protocol $\protHPC$ (Fact~\ref{fact:rectangle}). In particular, the same exact argument as in the two-player case implies that
	if $[(\bA_1,\bB_1),(\bC_1,\bD_1)]$ and $[(\bA_2,\bB_2),(\bC_2,\bD_2)]$ are mapped to the same transcript $\Prot$, then $[(\bA_1,\bB_1),(\bC_2,\bD_2)]$ and $[(\bA_2,\bB_2),(\bC_1,\bD_1)]$ are mapped to $\Prot$ as well. 
	Hence, since $(\bA,\bB) \perp (\bC,\bD)$ by Observation~\ref{obs:distHPC}, the inputs corresponding to the same protocol would also be independent of each other, namely, $(\bA,\bB) \perp (\bC,\bD) \mid \rProt = \Prot$. 
\end{proof}

\subsection*{Proof of Theorem~\ref{thm:hpck}: A Communication Lower Bound for \HPCk}
We prove the lower bound for any arbitrary deterministic protocol $\protHPC$ and then apply Yao's minimax principle~\cite{Yao83} to extend
it to randomized protocols as well. We first setup some notation. 

\paragraph{Notation.}  Fix any $k$-phase \emph{deterministic} protocol $\protHPC$ for $\HPCk$ throughout the proof.  We use $j=1$ to $k$ to index the phases of this protocol, as well as the pointers $z_1,\ldots,z_k$. 
For any $j \in [k]$, we define $\Prot_j$ as the set of all messages communicated by $\protHPC$ in phase $j$ and 
$\Prot:= (\Prot_1,\ldots,\Prot_k)$ as the transcript of the protocol $\protHPC$. 
\smallskip

For any $x \in \XX$ and any $y \in \YY$, we define 
the random variables $\rT_{x} \in \YY$ and $\rT_y \in \XX$, which correspond to the target elements of the \SI problem on $(A_x,B_x)$ and $(C_y,D_y)$, respectively. 

\smallskip

We further define $\rE_j := (\rProt^{<j},\rZ^{<j})$ for any $j > 1$ and $\rE_1 = z_0$, i.e., the first pointer. We can think of $\rE_j$ as the information ``easily known'' to all players at the beginning of phase $j$. 

\medskip

The main step of the proof of Theorem~\ref{thm:hpck} is the following key lemma which we prove inductively.  

\begin{lemma}\label{lem:hpc-induction}
	Let $\CC{\protHPC}{} := \CC{\protHPC}{\distHPC}$. There exists an absolute constant $c > 0$ such that for all $j \in [k]$: 
	\begin{align*}
		\Ex_{(E_j,\Prot_j)} \Bracket{\tvd{\distribution{\rZ_j \mid E_j,\Prot_j}}{\distribution{\rZ_j}}} \leq j \cdot c \cdot \Paren{\frac{\sqrt{\CC{\protHPC}{} + k \cdot \log{n}} + k}{n}}. 
	\end{align*}
\end{lemma}

Recall that distribution of each pointer $z_j$ is uniform over its support, i.e., over $\XX$ if $j$ is even, and over $\YY$ if $j$ is odd. 
Intuitively speaking, Lemma~\ref{lem:hpc-induction} states that if communication cost of a protocol is ``small'', i.e., is $o(n^2/k^2)$, then 
even after communicating the messages in the first $j$ phases of the protocol, distribution of $z_j$ is still ``close'' to being uniform. In other words,
the first $j$ phases of the protocol do not reveal ``any useful information'' about $z_j$. This in particular implies that at the end of the protocol, i.e., at the end of phase $k$, the target pointer $z_k$ is still 
uniform and $\protHPC$ should not be able to find it. We first formalize this inution and use it to prove Theorem~\ref{thm:hpck} and then present a proof of Lemma~\ref{lem:hpc-induction} which is the heart of the argument. 

\begin{proof}[Proof of Theorem~\ref{thm:hpck} (assuming Lemma~\ref{lem:hpc-induction})]
	The $\Omega(n)$ term in the lower bound trivially follows from the $\Omega(n)$ lower bound for set intersection (e.g. Theorem~\ref{thm:SI-tvd} with constant $\eps$). 
	In the following we prove the first (and the main) term. Note that for this purpose, we can assume $k = o(\sqrt{n})$ as otherwise the dominant term would already be the second term. 
	
	Let $\protHPC$ be any deterministic protocol for $\HPCk$ for $k=o(\sqrt{n})$ with communication cost $\CC{\protHPC}{\distHPC} = o(n^2/k^2)$. 
	Recall that $\distribution{\rZ_k} = \unif_\XX$ if $k$ is even and $\distribution{\rZ_k} = \unif_\YY$ if $k$ is odd. Let us assume by symmetry that $k$ is even. 
	By Lemma~\ref{lem:hpc-induction}, we have, 
	\begin{align}
		\Ex_{(E_k,\Prot_k)} \Bracket{\tvd{\distribution{\rZ_k \mid E_k,\Prot_k}}{\unif_\XX}} &\leq  k \cdot c \cdot \Paren{\frac{\sqrt{\CC{\protHPC}{} + k \cdot \log{n}} + k}{n}} \notag \\
		&= k \cdot c \cdot \paren{o(\frac{1}{k}) + o(\frac{\sqrt{\log{n}}}{n^{3/4}}) + o(\frac{k}{n})}  \notag \\ 
		&= o(\frac{k}{k}) + o(\frac{k \cdot \sqrt{\log{n}}}{n^{3/4}}) + o(\frac{k^2}{n}) = o(1),\label{eq:hpc-dist-unif}
	\end{align}
	as $c$ is an absolute constant. 
	
	On the other hand, $(E_k,\Prot_k)$ contains the whole transcript $\Prot$ of the protocol and hence the output of the protocol $\protHPC$ is fixed conditioned on $(E_k,\Prot_k)$. 
	We use $O(E_k,\Prot_k)$ to denote this output. We have, 
	\begin{align*}
		\Pr_{(E_k,\Prot_k)}\paren{\textnormal{$\protHPC$ is correct}} &= \Ex_{(E_k,\Prot_k)}\Pr_{\rZ_k \mid (E_k,\Prot_k)}\paren{\rZ_k = O(E_k,\Prot_k)} \\
		&\hspace{-11pt}\Leq{Fact~\ref{fact:tvd-small}} \Ex_{(E_k,\Prot_k)}\bracket{\Pr_{\rZ_k \sim \unif_\XX}\paren{\rZ_k = O(E_k,\Prot_k)} + \tvd{\distribution{\rZ_k \mid E_k,\Prot_k}}{\unif_\XX}} \\
		&\leq \frac{1}{n} + \Ex_{(E_k,\Prot_k)} \Bracket{\tvd{\distribution{\rZ_k \mid E_k,\Prot_k}}{\unif_\XX}} \Leq{Eq~(\ref{eq:hpc-dist-unif})} \frac{1}{n} + o(1). 
	\end{align*}
	Hence, $\protHPC$ cannot output the correct solution with at least a constant probability of success, proving the lower bound for deterministic algorithms. 
	
	To finalize, we can extend this (distributional) lower bound to randomized protocols by the easy direction of Yao's minimax principle~\cite{Yao83}, namely by an averaging argument that picks the ``best'' choice for randomness of the protocol. 
	This concludes the proof. \Qed{Theorem~\ref{thm:hpck}}
	
\end{proof}

\subsection*{Proof of Lemma~\ref{lem:hpc-induction}}

The following is a restatement of Lemma~\ref{lem:hpc-induction}.
\begin{lemma*}[Restatement of Lemma~\ref{lem:hpc-induction}]
	Let $\CC{\protHPC}{} := \CC{\protHPC}{\distHPC}$. There exists an absolute constant $c > 0$ such that for all $j \in [k]$: 
	\begin{align*}
		\Ex_{(E_j,\Prot_j)} \Bracket{\tvd{\distribution{\rZ_j \mid E_j,\Prot_j}}{\distribution{\rZ_j}}} \leq j \cdot c \cdot \Paren{\frac{\sqrt{\CC{\protHPC}{} + k \cdot \log{n}} + k}{n}}. 
	\end{align*}
\end{lemma*}

The proof of Lemma~\ref{lem:hpc-induction} consists of two main steps. We first show that finding the target element of a \emph{uniformly at random} chosen instance of $\SI$ (as opposed to the instance corresponding to any particular pointer) 
in $\HPC$ is not possible unless we make a large communication. Then, we prove inductively that in each phase $j$, the distribution of the pointer $z_j$ is close to uniform and hence by the argument in the first step, we should not be able to 
find the target element $t_{z_j}$ associated with $z_j$ and use this to finalize the proof. The following lemma captures the first part. 

\begin{lemma}\label{lem:multi-SI}
	There exists an absolute constant $c > 0$ such that for any $j \in [k]$,  
	\begin{align*}
		&\Ex_{(E_j,\Prot_j)}\Ex_{x \sim \unif_{\XX}}\bracket{\tvd{\distribution{\rT_x \mid E_j,\Prot_j}}{\distribution{\rT_x}}} \leq c \cdot \Paren{\frac{\sqrt{\CC{\protHPC}{} + j \cdot \log{n}} + j}{n}} ,\\
		&\Ex_{(E_j,\Prot_j)}\Ex_{y \sim \unif_{\YY}}\bracket{\tvd{\distribution{\rT_y \mid E_j,\Prot_j}}{\distribution{\rT_y}}} \leq c \cdot \Paren{\frac{\sqrt{\CC{\protHPC}{} + j \cdot \log{n}} + j}{n}} .
	\end{align*}
\end{lemma}
\noindent
The proof of this lemma is based on a direct-sum style argument combined with Theorem~\ref{thm:SI-tvd}. For intuition, consider a protocol that uses $o(n^2)$ communication in its first $j$ phases and assume by way of contradiction that it 
can reduce the LHS of one of the equations in Lemma~\ref{lem:multi-SI} by $\Omega(1)$. Using a direct-sum style argument, we can then argue that the transcript of the first $j$ phases of this protocol only reveal $o(n)$ bits of information about 
a uniformly at random chosen instance $(A_x,B_x)$ of $\SI$ but is enough to $\Omega(1)$-solve the instance $(A_x,B_x)$ (according to Definition~\ref{def:eps-solve}), which is in contradiction with our bounds in Theorem~\ref{thm:SI-tvd}. 
Note that in this discussion, for the sake of simplicity, we neglected the role of extra conditioning on $Z^{<j}$ in $E_j$ in the LHS of equations; handling this extra conditioning results in the extra additive factor in RHS. 

\begin{proof}[Proof of Lemma~\ref{lem:multi-SI}]
	 We only prove the first equation; the second one can be proven analogously.	
	Suppose towards a contradiction that this equation does not hold. We use $\protHPC$ to design a protocol $\protSI$
	that can $\eps$-solve the \SI problem $(A_x,B_x)$ for a uniformly at random chosen $x \in \XX$ and appropriately chosen $\eps \in (0,1)$ to be determined later (see Definition~\ref{def:eps-solve} for the notion of $\eps$-solve). 
	
	\begin{tbox}
    \textbf{Protocol $\protSI$}: The protocol for $\eps$-solving $\SI$ using a protocol $\protHPC$ for $\HPC_k$. 
    
     \smallskip
    
    \textbf{Input:} An instance $(A,B) \sim \distSI$ over the universe $\YY$. 
    
    \algline
    
    \begin{enumerate}
    	\item \textbf{Sampling the instance.} Alice and Bob create an instance $(\bA,\bB,\bC,\bD)$ of $\HPC_k$ as follows (see Figure~\ref{fig:hpc-sample} below for an illustration): 
	\begin{enumerate}
	\item Using \underline{public coins}, Alice and Bob sample an index $i \in [n]$ uniformly at random, and  
	Alice sets $A_{x_i} = A$ and Bob sets $B_{x_i} = B$ using their given inputs in $\SI$. 
	\item Using \underline{public coins}, Alice and Bob sample $A_{x_j}$ and $B_{x_k}$ from $\distSI$ for all $j < i < k$. 
	\item Using \underline{private coins}, Alice samples $A_{x_k}$ for $k > i$ such that $(A_{x_k},B_{x_k}) \sim \distSI$. Similarly Bob samples $B_{x_j}$ for $j < i$. This completes construction of $(\bA,\bB)$. 
        \item Using \underline{public coins}, Alice and Bob sample $(\bC,\bD)$ completely from $\distHPC$ (this is possible by Observation~\ref{obs:distHPC} as $(\bA,\bB) \perp (\bC,\bD)$). 
        \end{enumerate}
        
           \item\label{line:AB} \textbf{Computing the answer.}  Alice and Bob first check whether $x_i$ belongs to $z_0,z_1,\ldots,z_{j-1}$ or not. To do so, they start computing these pointers using the fact
           that for any underlying instance $(A_x,B_x) \in (\bA,\bB) \setminus (A_{x_i},B_{x_i})$ either Alice or Bob knows the entire instance. They terminate the protocol if ever $x_i$ belongs to one of the pointers computed so far. 
           We use $\Prot^*$ to denote the transcript of the protocol in this step (which is either $z_1,\ldots,z_{j-1}$ or some prefix of it ending in $x_i$). 
           
           \item Next, Alice and Bob run the protocol $\protHPC$ on the instance $(\bA,\bB,\bC,\bD)$ until its $j$-th phase by Alice playing $P_A$, Bob playing $P_B$, and both Alice and Bob simulating $P_C$ and $P_D$ with no communication (this is possible as 
           both Alice and Bob know $(\bC,\bD)$ entirely).
        
	\item The players return $\ProtSI := (\Prot_1,\ldots,\Prot_j,\Prot^*)$. 
    \end{enumerate}
\end{tbox}

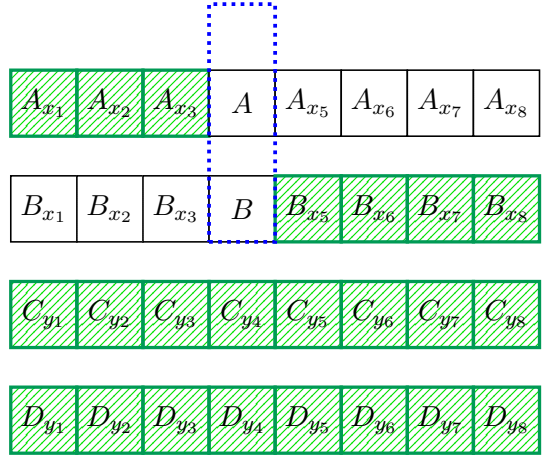
\begin{figure}[h!]
    \centering
\begin{tikzpicture}[ auto ,node distance =1cm and 2cm , on grid , semithick , state/.style ={ circle ,top color =white , bottom color = white , draw, black , text=black}, every node/.style={inner sep=0,outer sep=0}]

\node[rectangle, ForestGreen, text=black, draw, pattern=north east lines, pattern color=ForestGreen, line width=1.25pt, minimum height=25pt, minimum width = 25pt] (A1) {$A_{x_1}$};
\node[rectangle, ForestGreen, text=black, draw, pattern=north east lines, pattern color=ForestGreen, line width=1.25pt, minimum height=25pt, minimum width = 25pt] (A2) [right=25pt of A1]{$A_{x_2}$};
\node[rectangle, ForestGreen, text=black, draw, pattern=north east lines, pattern color=ForestGreen, line width=1.25pt, minimum height=25pt, minimum width = 25pt]  (A3) [right=25pt of A2]{$A_{x_3}$};
\node[state,rectangle, top color=white, black, bottom color=white, text=black, opacity=1, text opacity=1, minimum height=25pt, minimum width = 25pt]  (A4) [right=25pt of A3]{$A$};
\node[state,rectangle, top color=white, black, bottom color=white, text=black, opacity=1, text opacity=1, minimum height=25pt, minimum width = 25pt]  (A5) [right=25pt of A4]{$A_{x_5}$};
\node[state,rectangle, top color=white, black, bottom color=white, text=black, opacity=1, text opacity=1, minimum height=25pt, minimum width = 25pt]  (A6) [right=25pt of A5]{$A_{x_6}$};
\node[state,rectangle, top color=white, black, bottom color=white, text=black, opacity=1, text opacity=1, minimum height=25pt, minimum width = 25pt]  (A7) [right=25pt of A6]{$A_{x_7}$};
\node[state,rectangle, top color=white, black, bottom color=white, text=black, opacity=1, text opacity=1, minimum height=25pt, minimum width = 25pt]  (A8) [right=25pt of A7]{$A_{x_8}$};

\node[state,rectangle, top color=white, black, bottom color=white, text=black, opacity=1, text opacity=1, minimum height=25pt, minimum width = 25pt]  (B1) [below=40pt of A1]{$B_{x_1}$};
\node[state,rectangle, top color=white, black, bottom color=white, text=black, opacity=1, text opacity=1, minimum height=25pt, minimum width = 25pt]  (B2) [right=25pt of B1]{$B_{x_2}$};
\node[state,rectangle, top color=white, black, bottom color=white, text=black, opacity=1, text opacity=1, minimum height=25pt, minimum width = 25pt]  (B3) [right=25pt of B2]{$B_{x_3}$};
\node[state,rectangle, top color=white, black, bottom color=white, text=black, opacity=1, text opacity=1, minimum height=25pt, minimum width = 25pt]  (B4) [right=25pt of B3]{$B$};
\node[rectangle, ForestGreen, text=black, draw, pattern=north east lines, pattern color=ForestGreen, line width=1.25pt, minimum height=25pt, minimum width = 25pt]  (B5) [right=25pt of B4]{$B_{x_5}$};
\node[rectangle, ForestGreen, text=black, draw, pattern=north east lines, pattern color=ForestGreen, line width=1.25pt, minimum height=25pt, minimum width = 25pt]  (B6) [right=25pt of B5]{$B_{x_6}$};
\node[rectangle, ForestGreen, text=black, draw, pattern=north east lines, pattern color=ForestGreen, line width=1.25pt, minimum height=25pt, minimum width = 25pt]  (B7) [right=25pt of B6]{$B_{x_7}$};
\node[rectangle, ForestGreen, text=black, draw, pattern=north east lines, pattern color=ForestGreen, line width=1.25pt, minimum height=25pt, minimum width = 25pt] (B8) [right=25pt of B7]{$B_{x_8}$};

\node[state,rectangle, top color=white, black, bottom color=white, text=black, opacity=0, text opacity=1, minimum height=25pt, minimum width = 25pt]  (I) [above=25pt of A4]{$i$};

\node[inner sep=0pt, draw, dotted, blue, fit=(I) (B4), line width=0.5mm] {};

\node[rectangle, ForestGreen, text=black, draw, pattern=north east lines, pattern color=ForestGreen, line width=1.25pt, minimum height=25pt, minimum width = 25pt]  (C1) [below=40pt of B1]{$C_{y_1}$};
\node[rectangle, ForestGreen, text=black, draw, pattern=north east lines, pattern color=ForestGreen, line width=1.25pt, minimum height=25pt, minimum width = 25pt] (C2) [right=25pt of C1]{$C_{y_2}$};
\node[rectangle, ForestGreen, text=black, draw, pattern=north east lines, pattern color=ForestGreen, line width=1.25pt, minimum height=25pt, minimum width = 25pt]   (C3) [right=25pt of C2]{$C_{y_3}$};
\node[rectangle, ForestGreen, text=black, draw, pattern=north east lines, pattern color=ForestGreen, line width=1.25pt, minimum height=25pt, minimum width = 25pt]  (C4) [right=25pt of C3]{$C_{y_4}$};
\node[rectangle, ForestGreen, text=black, draw, pattern=north east lines, pattern color=ForestGreen, line width=1.25pt, minimum height=25pt, minimum width = 25pt]  (C5) [right=25pt of C4]{$C_{y_5}$};
\node[rectangle, ForestGreen, text=black, draw, pattern=north east lines, pattern color=ForestGreen, line width=1.25pt, minimum height=25pt, minimum width = 25pt]  (C6) [right=25pt of C5]{$C_{y_6}$};
\node[rectangle, ForestGreen, text=black, draw, pattern=north east lines, pattern color=ForestGreen, line width=1.25pt, minimum height=25pt, minimum width = 25pt]  (C7) [right=25pt of C6]{$C_{y_7}$};
\node[rectangle, ForestGreen, text=black, draw, pattern=north east lines, pattern color=ForestGreen, line width=1.25pt, minimum height=25pt, minimum width = 25pt]  (C8) [right=25pt of C7]{$C_{y_8}$};

\node[rectangle, ForestGreen, text=black, draw, pattern=north east lines, pattern color=ForestGreen, line width=1.25pt, minimum height=25pt, minimum width = 25pt]  (D1) [below=40pt of C1]{$D_{y_1}$};
\node[rectangle, ForestGreen, text=black, draw, pattern=north east lines, pattern color=ForestGreen, line width=1.25pt, minimum height=25pt, minimum width = 25pt]  (D2) [right=25pt of D1]{$D_{y_2}$};
\node[rectangle, ForestGreen, text=black, draw, pattern=north east lines, pattern color=ForestGreen, line width=1.25pt, minimum height=25pt, minimum width = 25pt]  (D3) [right=25pt of D2]{$D_{y_3}$};
\node[rectangle, ForestGreen, text=black, draw, pattern=north east lines, pattern color=ForestGreen, line width=1.25pt, minimum height=25pt, minimum width = 25pt]  (D4) [right=25pt of D3]{$D_{y_4}$};
\node[rectangle, ForestGreen, text=black, draw, pattern=north east lines, pattern color=ForestGreen, line width=1.25pt, minimum height=25pt, minimum width = 25pt]  (D5) [right=25pt of D4]{$D_{y_5}$};
\node[rectangle, ForestGreen, text=black, draw, pattern=north east lines, pattern color=ForestGreen, line width=1.25pt, minimum height=25pt, minimum width = 25pt]  (D6) [right=25pt of D5]{$D_{y_6}$};
\node[rectangle, ForestGreen, text=black, draw, pattern=north east lines, pattern color=ForestGreen, line width=1.25pt, minimum height=25pt, minimum width = 25pt]  (D7) [right=25pt of D6]{$D_{y_7}$};
\node[rectangle, ForestGreen, text=black, draw, pattern=north east lines, pattern color=ForestGreen, line width=1.25pt, minimum height=25pt, minimum width = 25pt]  (D8) [right=25pt of D7]{$D_{y_8}$};

\end{tikzpicture}
    \caption{Illustration of the process of sampling of instances of $\HPC$ in $\protSI$ for $n=8$. In this example, $i=4$ and hence $(A_{x_4},B_{x_4}) = (A,B)$ and the players sample $\set{A_{x_1},A_{x_2},A_{x_3},B_{x_5},B_{x_6},B_{x_7},B_{x_8}}$ as well
    as the entire $\bC$ and $\bD$ using public randomness. Then, Alice samples $\set{A_{x_5},A_{x_6},A_{x_7},A_{x_8}}$ and Bob samples $\set{B_{x_1},B_{x_2},B_{x_3}}$ using private randomness, respectively. }
    \label{fig:hpc-sample}
\end{figure}

Similar to the case of the sampling in protocol $\protPI$ in Section~\ref{sec:SI}, here also the public-private randomness sampling of the instance of $\HPC$ inside $\protSI$ is only for the sake of the information theoretic arguments; for the 
rest of the analysis, we only care that the distribution of the instances of $\HPC$ sampled in $\protSI$ is $\distHPC$. We first determine the parameter $\eps$ for which $\protSI$ $\eps$-solves $\SI$. 

\begin{claim}\label{clm:si-eps-solve-reduction}
	$\protSI$ $\eps$-solves $\SI$ on $\distSI$ for 
	\begin{align*}
		\eps \geq \Ex_{(E_{j},\Prot_j)} \Ex_{x \sim \unif_{\XX}}\Bracket{\tvd{\distribution{\rT_{x} \mid E_j,\Prot_j}}{\distribution{\rT_x}}} - \frac{j}{n}, 
	\end{align*}
	where $(E_j,\Prot_j,\rT_x)$ are distributed according to $\distHPC$.

\end{claim}
\begin{proof}
	By Definition~\ref{def:eps-solve}, $\protSI$ $\eps$-solves $\SI$ for $\eps := \Ex_{\ProtSI} \Bracket{\tvd{\distribution{\rT \mid \ProtSI}}{\distribution{\rT}}}$. We thus bound the RHS of this equation. We have, 
	\begin{align*}
		\Ex_{\ProtSI} \Bracket{\tvd{\distribution{\rT \mid \ProtSI}}{\distribution{\rT}}} &= \Ex_{(E_j,\Prot_j,\ProtSI,i)}\Bracket{\tvd{\distribution{\rT_{x_i} \mid \ProtSI}}{\distribution{\rT_{x_i}}}} 
		\tag{as $\rT = \rT_{x_i}$ for $\rI = i$} \\
		&=  \Ex_{(E_j,\Prot_j)} \Ex_{i}\Ex_{\ProtSI \mid (E_j,\Prot_j,i)}\Bracket{\tvd{\distribution{\rT_{x_i} \mid \ProtSI}}{\distribution{\rT_{x_i}}}} \\
		\tag{as $(\rE_j,\rProt_j) \perp \rI$} \\
		&= \Ex_{(E_j,\Prot_j)} \Bracket{\sum_{i=1}^{n} \frac{1}{n}  \Ex_{\ProtSI \mid (E_j,\Prot_j,i)}\Bracket{\tvd{\distribution{\rT_{x_i} \mid \ProtSI}}{\distribution{\rT_{x_i}}}}} 
		\tag{distribution of $i$ is uniform over $[n]$}\\
		&= \Ex_{(E_j,\Prot_j)} \Big[ \sum_{x_i \in Z^{<j}} \frac{1}{n} \cdot \tvd{\distribution{\rT_{x_i} \mid Z^{<j'}}}{\distribution{\rT_{x_i}}} \\
		&\hspace{40pt} + \sum_{x_i \notin Z^{<j}} \frac{1}{n} \cdot \tvd{\distribution{\rT_{x_i} \mid E_j,\Prot_j}}{\distribution{\rT_{x_i}}} \Big] 
		\tag{$\Prot^* := Z^{<j'}$ for some $j' < j-1$ when $x_i \in Z^{<j}$ and is otherwise equal to $E_j,\Prot_j$)} \\
		&= \Ex_{(E_j,\Prot_j)} \bracket{\sum_{x_i \notin Z^{<j}} \frac{1}{n} \cdot \tvd{\distribution{\rT_{x_i} \mid E_j,\Prot_j}}{\distribution{\rT_{x_i}}} }
		\tag{$\rT_{x_i} \perp \Prot^*$ and so $\tvd{\distribution{\rT_{x_i} \mid Z^{<j'}}}{\distribution{\rT_{x_i}}} = \tvd{\distribution{\rT_{x_i}}}{\distribution{\rT_{x_i}}} = 0$} \\
		&\geq \Ex_{(E_{j},\Prot_j)} \Ex_{i}\Bracket{\tvd{\distribution{\rT_{x_i} \mid E_j,\Prot_j}}{\distribution{\rT_{x_i}}}} - \frac{j}{n}.
		\tag{as total variation distance is bounded by one $\card{Z^{<j}} = j$}
	\end{align*}
	Replacing $x_i$ for $i$ chosen randomly from $[n]$ above by $x \sim \unif_{\XX}$ concludes the proof. \Qed{Claim~\ref{clm:si-eps-solve-reduction}}
	
\end{proof}
\noindent
The RHS in Claim~\ref{clm:si-eps-solve-reduction} is the quantity we aim to bound in this lemma (minus the extra additive $j/n$ term). 
To do so, we are going to bound the internal information cost of $\protSI$ by the communication cost of $\protHPC$ in the following claim and then use Theorem~\ref{thm:SI-tvd} to relate
this quantity to $\eps$. 

\begin{claim}\label{clm:si-info-cost}
	$\IC{\protSI}{\distSI} = O\Paren{\frac{\CC{\protHPC}{} }{n}   + \frac{j \cdot \log{n}}{n}}$. 
\end{claim}
\begin{proof}
For any $i \in [n]$, define $\bA^{<i} := \set{A_{x_1},\ldots,A_{x_{i-1}}}$, $\bB^{>i} := \set{B_{x_{i+1}},\ldots,B_{x_{n}}}$. Recall that 
the internal information cost of $\protSI$ is $\IC{\protSI}{\distSI} := \mi{\rA}{\rProtSI \mid \rB} + \mi{\rB}{\rProtSI \mid \rA}$. In the following, we focus on bounding the first term. The second term can be bounded exactly
the same by symmetry.  

As $(\rI,\bA^{<\rI},\bB^{>\rI},\bC,\bD)$ is sampled via public randomness in $\protSI$, by Proposition~\ref{prop:public-random},   
\begin{align*}
	\mi{\rA}{\rProtSI \mid \rB} &= \mi{\rA}{\rProtSI \mid \rB, \rI,\bA^{<\rI},\bB^{>\rI},\bC,\bD} \leq \mi{\rA}{\rProtSI \mid \rB, \rI,\bA^{<\rI},\bB^{>\rI}}. 
\end{align*}
The inequality is by Proposition~\ref{prop:info-decrease} as we now show $\rA \perp (\bC,\bD) \mid \rProtSI,\rB,\rI,\bA^{<\rI},\bB^{>\rI}$ (and hence conditioning on $(\bC,\bD)$ can only decrease the mutual information). 
This is because $\rA \perp (\bC,\bD) \mid \rB,\rI,\bA^{<\rI},\bB^{>\rI}$ by Observation~\ref{obs:distHPC} and $\rProtSI$ is transcript of a deterministic protocol plus $z_1,\ldots,z_j$ obtained deterministically and hence we can apply
Proposition~\ref{prop:hpc-rectangle}.  

Define a random variable $\rTheta \in \set{0,1}$ where $\rTheta = 1$ iff in Line~(\ref{line:AB}) of protocol $\protSI$, we terminate the protocol. In other words $\rTheta = 1$ iff $x_i \in Z^{<j}$. Since $\rA \perp \rTheta \mid \rB, \rI,\bA^{<\rI},\bB^{>\rI}$, 
further conditioning on $\rTheta$ can only increase the mutual information term above by Proposition~\ref{prop:info-increase}, hence, 
\begin{align}
	\mi{\rA}{\rProtSI \mid \rB} &\leq \mi{\rA}{\rProtSI \mid \rB, \rI,\bA^{<\rI},\bB^{>\rI},\rTheta} \notag \\
	&= \frac{n-j}{n} \cdot \mi{\rA}{\rProtSI \mid \rB, \rI,\bA^{<\rI},\bB^{>\rI},\rTheta=0} + \frac{j}{n} \cdot \mi{\rA}{\rProtSI \mid \rB, \rI,\bA^{<\rI},\bB^{>\rI},\rTheta=1} \notag \\
	&\leq \frac{n-j}{n} \cdot \mi{\rA}{\rProtSI \mid \rB, \rI,\bA^{<\rI},\bB^{>\rI},\rTheta=0}, \label{eq:mi-ra}
\end{align}
since conditioned on $\rTheta=1$, the protocol $\rProtSI$ is simple some prefix of $\rZ^{<j}$ and is hence independent of the input $(\rA,\rB)$ and carries no information about $\rA$ (see \itfacts{info-zero}). 
We now further bound the RHS of Eq~(\ref{eq:mi-ra}). When $\rTheta = 0$, $\rProtSI = (\rZ^{<j} ,\rProt_1,\ldots,\rProt_j) = (\rE^{<j},\rProt_j)$. Hence, we can write, 
  \begin{align*}
	 \mi{\rA}{\rProtSI \mid \rB, \rI,\bA^{<\rI},\bB^{>\rI},\rTheta=0} &\leq \mi{\rA}{\rE_j,\rProt_j \mid \rB, \rI,\bA^{<\rI},\bB^{>\rI},\rTheta=0} \\
	&= \mi{\rA}{\rZ^{<j} \mid \rB, \rI,\bA^{<\rI},\bB^{>\rI},\rTheta=0} \\
	&\hspace{25pt} + \mi{\rA}{\rProt^{<j},\rProt_j \mid \rZ^{<j},\rB, \rI,\bA^{<\rI},\bB^{>\rI},\rTheta=0} \tag{by chain rule in \itfacts{chain-rule} and since $\rE_j = (\rProt^{<j},\rZ^{<j})$} \\
	&\leq \mi{\rA}{\rProt \mid \rZ^{<j},\rB, \rI,\bA^{<\rI},\bB^{>\rI},\rTheta=0},
\end{align*}
as $\rA \perp \rZ^{<j} \mid \rTheta=0$ (and other variables) and hence the first term is zero, and in the second term $\rProt$ contains $\rProt^{<j},\rProt_j$ (plus potentially other terms) and so having $\rProt$ in instead can only increase
the information. By further expanding the conditional information term above,
\begin{align*}
	 &\mi{\rA}{\rProtSI \mid \rB, \rI,\bA^{<\rI},\bB^{>\rI},\rTheta=0} \\
	 &\hspace{50pt}\leq \Ex_{(Z^{<j},i) \mid \rTheta=0} \bracket{\mi{\rA}{\rProt \mid \rB, \bA^{<i},\bB^{>i},\rI=i,\rZ^{<j}=Z^{<j},\rTheta=0}} \\
	&\hspace{50pt}= \Ex_{Z^{<j} \mid \rTheta=0} \bracket{\sum_{\substack{i=1 \\~i \notin Z^{<j}}}^{n} \frac{1}{n-j}  \mi{\rA_{x_i}}{\rProt \mid \rB_{x_i}, \bA^{<i},\bB^{>i},\rI=i,\rZ^{<j}=Z^{<j},\rTheta=0}} \tag{conditioned on $\rTheta=0$, $i$ is chosen
	uniformly at random from $Z^{<j}$; also $(\rA,\rB) = (\rA_{x_i},\rB_{x_i})$} \\
	&\hspace{50pt}= \Ex_{Z^{<j} \mid \rTheta=0} \bracket{\sum_{i \notin Z^{<j}} \frac{1}{n-j} \cdot \mi{\rA_{x_i}}{\rProt \mid \rB_{x_i}, \bA^{<i},\bB^{>i},\rZ^{<j}=Z^{<j},\rTheta=0}}
	\tag{we dropped the conditioning on $\rI=i$ as all remaining variables are independent of this event}  \\
	&\hspace{50pt}= \Ex_{Z^{<j} \mid \rTheta=0} \bracket{\sum_{i \notin Z^{<j}} \frac{1}{n-j} \cdot \mi{\rA_{x_i}}{\rProt \mid \bA^{<i},\bB,\rZ^{<j}=Z^{<j},\rTheta=0}}
	\tag{as $\rA_{x_i} \perp \bB^{<i} \mid \rB_{x_i},\bA^{<i}$ by Observation~\ref{obs:distHPC} and hence we can apply Proposition~\ref{prop:info-increase}} \\
	&\hspace{50pt}\leq \Ex_{Z^{<j} \mid \rTheta=0} \bracket{\sum_{i=1}^{n} \frac{1}{n-j} \cdot \mi{\rA_{x_i}}{\rProt \mid \bA^{<i},\bB,\rZ^{<j}=Z^{<j},\rTheta=0}} \\ 
	\tag{mutual information is non-negative by \itfacts{info-zero} and so we can add the terms in $Z^{<j}$  as well} \\
	&\hspace{50pt}= \Ex_{Z^{<j} \mid \rTheta=0} \bracket{\sum_{i=1}^{n} \frac{1}{n-j} \cdot \mi{\rA_{x_i}}{\rProt \mid \bA^{<i},\bB,\rZ^{<j}=Z^{<j},\rTheta=0}} \\
	&\hspace{50pt}= \frac{1}{n-j} \cdot \Ex_{Z^{<j} \mid \rTheta=0} \bracket{\mi{\bA}{\rProt \mid \bB,\rZ^{<j}=Z^{<j},\rTheta=0}} \tag{by chain rule in \itfacts{chain-rule}} \\
	&\hspace{50pt}= \frac{1}{n-j} \cdot \mi{\bA}{\rProt \mid \bB,\rZ^{<j},\rTheta=0} \tag{by Proposition~\ref{prop:bar-hopping}}\\
	&\hspace{50pt}\leq \frac{1}{n-j} \cdot \paren{\mi{\bA}{\rProt \mid \bB,\rTheta=0} + \en{\rZ^{<j}}} \\
	&\hspace{50pt}= \frac{1}{n-j} \cdot \paren{\mi{\bA}{\rProt \mid \bB} + \en{\rZ^{<j}}} \tag{transcript of the protocol $\protHPC$ (namely $\rProt$) on input $(\bA,\bB)$ is independent of $\rTheta$} \\
	&\hspace{50pt}\leq \frac{1}{n-j} \cdot \paren{\en{\rProt} + \en{\rZ^{<j}}} \leq \frac{\CC{\protHPC}{} }{n-j}   + \frac{j \cdot \log{n}}{n-j}. \tag{by sub-additivity of entropy (\itfacts{sub-additivity}) and \itfacts{uniform}}
\end{align*}

By plugging in this bound in Eq~(\ref{eq:mi-ra}), we have that,
\begin{align*}
	\mi{\rA}{\rProtSI \mid \rB} &\leq \frac{n-j}{n} \cdot \paren{\frac{\CC{\protHPC}{} }{n-j}   + \frac{j \cdot \log{n}}{n-j}} = \frac{\CC{\protHPC}{} }{n}   + \frac{j \cdot \log{n}}{n}. 
\end{align*}
By symmetry, we can also prove the same bound on $\mi{\rB}{\rProtSI \mid \rA}$. As such, we have,
\begin{align*}
	\mi{\rA}{\rProtSI \mid \rB} + \mi{\rB}{\rProtSI \mid \rA} \leq 2 \cdot \paren{\frac{\CC{\protHPC}{} }{n}   + \frac{j \cdot \log{n}}{n}}.
\end{align*}

We shall note that strictly speaking the factor $2$ above is not needed (similar to the proof of Proposition~\ref{prop:cc-ic}) but as this factor is anyway suppressed through O-notation later in the proof, the above bound
suffices for our purpose. \Qed{Claim~\ref{clm:si-info-cost}}

\end{proof}

Now by Claim~\ref{clm:si-info-cost}, we have that
\[
	\IC{\protSI}{\distSI} = O\Paren{\frac{\CC{\protHPC}{}}{n}   + \frac{j \cdot \log{n}}{n}}.
\]
Combined with Theorem~\ref{thm:SI-tvd}, this implies that $\protSI$ can only $\eps$-solves $\SI$ for parameter $\eps$ such that 
\[
	\eps^2 \cdot n = O\Paren{\frac{\CC{\protHPC}{}}{n}   + \frac{j \cdot \log{n}}{n}} \implies \eps = O\Paren{\frac{\sqrt{\CC{\protHPC}{} + j \cdot \log{n}}}{n}}. 
\]
On the other hand, by Claim~\ref{clm:si-eps-solve-reduction}, we know that
\[
	\eps \geq \Ex_{(E_{j},\Prot_j)} \Ex_{x \sim \unif_{\XX}}\Bracket{\tvd{\distribution{\rT_{x} \mid E_j,\Prot_j}}{\distribution{\rT_x}}} - \frac{j}{n}.
\]
which implies 
\[
	\Ex_{x \sim \unif_{\XX}}\Bracket{\tvd{\distribution{\rT_{x} \mid E_j,\Prot_j}}{\distribution{\rT_x}}} = O\Paren{\frac{\sqrt{\CC{\protHPC}{} + j \cdot \log{n}} + j}{n}}. 
\]
This concludes the proof. \Qed{Lemma~\ref{lem:multi-SI}}

\end{proof}

Before getting to the proof of Lemma~\ref{lem:hpc-induction}, we also need the following simple claim based on the rectangle property of the protocol $\protHPC$. 

\begin{claim}\label{clm:hpc-rectangle}
	For any $j \in [k]$ and choice of $(E_j,\Prot_j)$, $\distribution{\rZ_j \mid E_j,\Prot_j} = \distribution{\rZ_j \mid E_j}$. 
\end{claim}
\begin{proof}
	This is because for any $j \in [k]$, $\rZ_j \perp \rProt_j \mid E_j$: Conditioned on $\rE_j = E_j = (Z^{<j},\Prot^{<j})$, $\rProt_j$ is only a function of $(\bA,\bB)$ if $j$ is even and a function of $(\bC,\bD)$ if $j$ is odd. 
	On the other hand, $\rZ_j$ is only a function of $(\bA,\bB)$ if $j$ is odd and a function of $(\bC,\bD)$ if $j$ is even. Finally, by Observation~\ref{obs:distHPC}, $(\bA,\bB) \perp (\bC,\bD)$ and this continues to hold even when we condition 
	on $E_j$ by the rectangle property of the protocol $\protHPC$; hence the claim follows. \Qed{Claim~\ref{clm:hpc-rectangle}}
	
\end{proof}

We are now finally ready to prove Lemma~\ref{lem:hpc-induction}. 
\begin{proof}[Proof of Lemma~\ref{lem:hpc-induction}]

	Let $c$ be the constant in Lemma~\ref{lem:multi-SI}. 
	We prove Lemma~\ref{lem:hpc-induction} by induction. We start with the proof of the base case for $j=1$ and then prove the inductive step. 
	
	\paragraph{\emph{{Base case.}}} Recall that we defined $E_1 = z_0$ which is deterministically fixed. This, together with Claim~\ref{clm:hpc-rectangle}, implies that $\distribution{\rZ_1 \mid E_1,\Prot_1} = \distribution{\rZ_1}$, which
	finalizes proof of the base case. 
	
	\paragraph{\emph{Induction step.}} Let us now prove the lemma inductively for $j > 1$. We have, 
	\begin{align*}
		\Ex_{(E_j,\Prot_j)} \Bracket{\tvd{\distribution{\rZ_j \mid E_j,\Prot_j}}{\distribution{\rZ_j}}} &\Eq{Claim~\ref{clm:hpc-rectangle}} \Ex_{E_j} \Bracket{\tvd{\distribution{\rZ_j \mid E_j}}{\distribution{\rZ_j}}} \\
		&\hspace{0.45cm}= \Ex_{(Z^{<j},\Prot^{<j})} \Bracket{\tvd{\distribution{\rZ_j \mid Z^{<j},\Prot^{<j}}}{\distribution{\rZ_j}}} \tag{by definition of $E_j := (Z^{<j},\Prot^{<j})$} \\
		&\hspace{0.45cm}=\Ex_{(Z^{<j},\Prot^{<j})} \Bracket{\tvd{\distribution{\rT_{z_{j-1}} \mid Z^{<j-1},z_{j-1},\Prot^{<j}}}{\distribution{\rZ_j}}} \tag{by definition, the pointer $\rZ_j = \rT_{z_{j-1}}$}.
	\end{align*}
	We can write the RHS above as: 
	\begin{align*}
		&\Ex_{(E_j,\Prot_j)} \Bracket{\tvd{\distribution{\rZ_j \mid E_j,\Prot_j}}{\distribution{\rZ_j}}} \\
		&\qquad \qquad = \Ex_{(Z^{<j-1},\Prot^{<j})} \Ex_{z_{j-1} \sim \rZ_{j-1} \mid \paren{Z^{<j-1},\Prot^{<j}}}\Bracket{\tvd{\distribution{\rT_{z_{j-1}} \mid Z^{<j-1},\Prot^{<j}}}{\distribution{\rZ_j}}}.
	\end{align*}
	This is because $\rT_{z_{j-1}} \perp (\rZ_{j-1}=z_{j-1}) \mid Z^{<j-1},\Prot^{<j}$: if $j-1$ is odd, $\rT_{z_{j-1}}$ is a function of  $(\bC,\bD)$ and if $j-1$ is even, $\rT_{z_{j-1}}$ is a function of $(\bA,\bB)$. On the other hand,
	if $j-1$ is odd, then $\rZ_{j-1}$ is a function of $(\bA,\bB)$ and if even, then $\rZ_{j-1}$ is a function of $(\bC,\bD)$. Finally, by Proposition~\ref{prop:hpc-rectangle}, 
	$(\bA,\bB) \perp (\bB,\bD) \mid \Prot^{<j}$, proving the conditional independence. 
	
	Now notice that distribution of $z_{j-1}$ in the expectation-term above is $\distribution{\rZ_{j-1} \mid E_{j-1},\Prot_{j-1}}$.  By symmetry, let us assume $j-1$ is odd and hence $z_{j-1} \in \YY$. Using Fact~\ref{fact:tvd-small} and since total variation distance is bounded by $1$ always, 
	we can upper bound RHS above with:
	\begin{align*}
		&\Ex_{(E_j,\Prot_j)} \Bracket{\tvd{\distribution{\rZ_j \mid E_j,M_j}}{\distribution{\rZ_j}}}  \\
		&\qquad \qquad \leq \Ex_{(Z^{<j-1},\Prot^{<j})} \bracket{\Ex_{\paren{z_{j-1} \sim \unif_{\YY}}}\Bracket{\tvd{\distribution{\rT_{z_{j-1}} \mid Z^{<j-1},\Prot^{<j}}}{\distribution{\rZ_j}}}} \\
		&\qquad \qquad  \qquad \qquad \qquad  \qquad + \Ex_{(Z^{<j-1},\Prot^{<j})} \bracket{\tvd{\distribution{\rZ_{j-1} \mid E_{j-1},\Prot_{j-1}}}{\unif_{\YY}}}  \\
		&\qquad \qquad =  \Ex_{(E_{j-1},\Prot_{j-1})} \Ex_{{y \sim \unif_{\YY}}}\Bracket{\tvd{\distribution{\rT_{y} \mid E_{j-1},\Prot_{j-1}}}{\distribution{\rZ_j}}} \\
		&\qquad \qquad  \qquad \qquad \qquad  \qquad + \Ex_{(E_{j-1},\Prot_{j-1})} \bracket{\tvd{\distribution{\rZ_{j-1} \mid E_{j-1},\Prot_{j-1}}}{\distribution{\rZ_{j-1}}}}, 	
	\end{align*}
	where in the first term above we only changed the name of variable $z_{j-1}$ to $y$ and in the second term we used $\distribution{\rZ_{j-1}} = \unif_{\YY}$. By Lemma~\ref{lem:multi-SI}, we can bound the first term and by induction, we can bound the second one. Hence, 
	\begin{align*}
		\Ex_{(E_j,\Prot_j)} \Bracket{\tvd{\distribution{\rZ_j \mid E_j,\Prot_j}}{\distribution{\rZ_j}}} &\leq c \cdot \Paren{\frac{\sqrt{\CC{\protHPC}{}+ j \cdot \log{n}} + j}{n}} \\
		&\hspace{20pt} + (j-1) \cdot c \cdot \Paren{\frac{\sqrt{\CC{\protHPC}{} + k \cdot \log{n}} + k}{n}} \\
		&\leq j \cdot c \cdot \Paren{\frac{\sqrt{\CC{\protHPC}{} + k \cdot \log{n}} + k}{n}}. \tag{where we replaced $j \leq k$ by $k$ in the first term}
	\end{align*}
	This concludes the proof. \Qed{Lemma~\ref{lem:hpc-induction}}
	 
\end{proof}

\newcommand{\MF}{\ensuremath{\textnormal{\textsf{MaxFlow}}}\xspace}
\newcommand{\LFMIS}{\ensuremath{\textnormal{\textsf{LFMIS}}}\xspace}

\renewcommand{\mod}{\ensuremath{~\textnormal{mod}~}}
\newcommand{\wstar}{w^{*}}
\newcommand{\istar}{i^{*}}

\section{Graph Streaming Lower Bounds}\label{sec:lower}

We now present our graph streaming lower bounds using reductions from the hidden-pointer chasing problem. In particular, 
we prove the following two results in this section. 

\begin{theorem}[Formalizing Result~\ref{res:cut}] \label{thm:cut}
For any integer $p \geq 1$, any $p$-pass streaming algorithm that with a constant probability outputs the minimum $s$-$t$ cut value in a weighted directed or undirected graph $G(V,E,w)$ requires 
$\Omega({n^2}/{p^5})$ bits of space. 
\end{theorem}

By max-flow min-cut theorem, Theorem~\ref{thm:cut} also holds for streaming algorithms that can compute the value of maximum $s$-$t$ flow in a capacitated graph (directed or undirected). 

\begin{theorem}[Formalizing Result~\ref{res:mis}] \label{thm:mis}
For any integer $p \geq 1$, any $p$-pass streaming algorithm that with a constant probability outputs the lexicographically-first maximal independent set of an undirected graph $G(V,E)$ requires 
$\Omega({n^2}/{p^5})$ bits of space. 
\end{theorem}

\noindent
We prove Theorems~\ref{thm:cut} and~\ref{thm:mis} in Sections~\ref{sec:cut} and~\ref{sec:mis}, respectively.  

\newcommand{\PP}{\ensuremath{\mathcal{P}}}

\subsection{Weighted Minimum $s$-$t$ Cut Problem}\label{sec:cut}

We prove Theorem~\ref{thm:cut} by a  reduction from our hidden-pointer chasing (\HPC) problem. We first give the lower bound for directed graphs and then show how to 
extend it using standard techniques to undirected graphs. 

We turn an instance $(\bA,\bB,\bC,\bD)$ of \HPCk over universes $\XX$ and $\YY$ of $n$ elements, into a weighted directed graph $G(V,E,w)$. The reduction is as follows (see Figure~\ref{fig:cut} on page~\pageref{fig:cut} for an example): 

\begin{itemize}
	\item The vertex-set $V$ of $G$ is partitioned into $k+1$ layers $V_0,\ldots,V_k$ each of size $n$ plus the source and sink vertices $s$ and $t$. We denote the $i$-th vertex in layer $V_j$ by $v^j_i$. 
	
	\item Define the following sequence of weights $w_0,w_1,\ldots,w_k$ where $w_j := (n+1)^{k+1-j}$ for all $j \in [k]$. Hence, $w_k = (n+1)$ and $w_j = (n+1) \cdot w_{j+1}$ for all $j < k$. 
	\item The edge-set $E$ of $G$ contains the following \underline{input-independent} edges. 
	\begin{itemize}
		\item source $s$ is connected to $v^0_1$ with weight $w(s,v^0_1) = w_0$. 
		\item for $0 < j \leq k$, every vertex $v^j_i$ in layer $V_j$ is connected to sink $t$ with weight $w(v^j_i,t) = w_j$. 
		\item any vertex $v^k_i$ in layer $V_k$ is connected to sink $t$ with weight $w(v^k_i,t) = i-1$ (notice that $v^k_i$ also has another edge of weight $w_k$ to $t$ by the previous part). 
	\end{itemize}
	\item The edge-set $E$ also contains the following \underline{input-dependent} edges.
	\begin{itemize}
		\item for all $i \in [n]$, if $A_{x_i}\in \bA$ (resp. $B_{x_i} \in \bB$) contains $y_{i'} \in \YY$, we connect $v^{j}_i$ in layer $V_j$ to $v^{j+1}_{i'}$ in layer $V_{j+1}$ with weight $w(v^j_i,v^{j+1}_{i'}) = w_{j+1}$ for every \emph{even} $0 \leq j < k$.%
		\footnote{\label{foot:simple-graph} Note that we will add two edges between $v^{j}_i$ and $v^{j+1}_{i'}$ iff $y_{i'} \in A_{x_i} \cap B_{x_i}$ and we will keep both copies of these edges in $G$ (see also Remark~\ref{rem:simple-graph} on how to remove the parallel edges).} 
		\item for all $i \in [n]$, if $C_{y_i}\in \bC$ (resp. $D_{y_i} \in \bD$) contains $x_{i'} \in \XX$, we connect $v^{j}_i$ in layer $V_j$ to $v^{j+1}_{i'}$ in layer $V_{j+1}$ with weight $w(v^j_i,v^{j+1}_{i'}) = w_{j+1}$ for every \emph{odd} $0 < j < k$. 
	\end{itemize}
\end{itemize}
\noindent
This concludes the description of the weighted graph $G(V,E,w)$ in the reduction. It is straightforward to verify that this graph can be constructed from an instance $(\bA,\bB,\bC,\bD)$ with no communication between the players. We now prove the following key lemma which establishes the correctness of the reduction.

\begin{lemma} \label{lem:cut}
    Let $\wstar$ be the weight of a minimum $s$-$t$ cut in graph $G$ in the reduction. Let the pointer $z_k$ be $x_{\istar}$ (resp. $y_{\istar}$) if $k$ is even (resp. odd). Then $\istar = (\wstar \mod (n+1)) +1$.
\end{lemma}

\begin{proof}
We prove this lemma by considering the maximum $s$-$t$ flow in $G$ and then use the duality of maximum flow and minimum cut to conclude the proof. For the flow problem, we assume that the capacity $c(e)$
of an edge $e=(u,v)$ in $G$ is equal to the total weight of the edges (in $w$) that connect $u$ to $v$ (recall that $G$ may have parallel edges; see Footnote~\ref{foot:simple-graph}). 

We start with some definitions. Define $u_j$ in layer $V_j$ to be the vertex corresponding to the pointer $z_j$, namely, for all even (resp. odd) values of $j$, $u_j=v^j_i$ where $x_i = z_j$ (resp. $y_i = z_j$). 
Furthermore, let $\PP := \PP_1 \cup \ldots \cup \PP_k \cup \set{P^*}$ be a collection of flow paths defined as follows: For any $j \in [k]$, the set of paths $\PP_j :=  \set{(s,u_0,u_1,\dots, u_{j-1}, v^j_i ,t) \mid (u_{j-1},v^j_i) \in E}$ and
each path in $\PP_j$ carries $w_j$ units of flow; moreover, $P^*=(s,u_0,u_1,\dots,u_k,t)$ and carries $\istar-1$ units of flow. See Figure~\ref{fig:flow} for an illustration. 

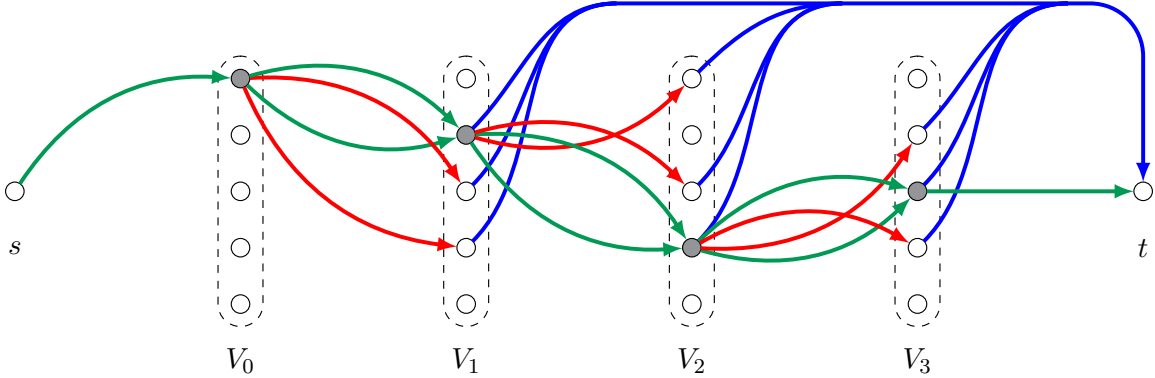
\begin{figure}[h!]
    \centering
\begin{tikzpicture}[ auto ,node distance =1cm and 2cm , on grid , semithick , state/.style ={ circle ,top color =white , bottom color = white , draw, black , text=black}, every node/.style={inner sep=0,outer sep=0}]

\node[circle, black, line width=0.15mm, minimum height=7pt, minimum width=7pt, fill=Gray, draw] (a1){};
\node[state, circle, black, line width=0.15mm, minimum height=7pt, minimum width=7pt] (a2) [below= 0.75cm of a1]{};
\node[state, circle, black, line width=0.15mm, minimum height=7pt, minimum width=7pt] (a3) [below=0.75cm of a2]{};
\node[state, circle, black, line width=0.15mm, minimum height=7pt, minimum width=7pt] (a4) [below=0.75cm of a3]{};
\node[state, circle, black, line width=0.15mm, minimum height=7pt, minimum width=7pt] (a5) [below=0.75cm of a4]{};
\node (v0) [below=0.75cm of a5]{$V_0$};


\node[rectangle, dashed, rounded corners = 3mm, inner sep=5pt, draw,  black, fit=(a1) (a5), line width=0.15mm] {};

\node[state, circle, black, draw, line width=0.15mm, minimum height=7pt, minimum width=7pt] (b1) [right=3cm of a1]{};
\node[ circle, black, line width=0.15mm, minimum height=7pt, minimum width=7pt, fill=Gray, draw] (b2) [below= 0.75cm of b1]{};
\node[state, circle, black, line width=0.15mm, minimum height=7pt, minimum width=7pt] (b3) [below=0.75cm of b2]{};
\node[state, circle, black, line width=0.15mm, minimum height=7pt, minimum width=7pt] (b4) [below=0.75cm of b3]{};
\node[state, circle, black, line width=0.15mm, minimum height=7pt, minimum width=7pt] (b5) [below=0.75cm of b4]{};
\node (v1) [below=0.75cm of b5]{$V_1$};

\node[rectangle, dashed,rounded corners = 3mm, inner sep=5pt, draw,  black, fit=(b1) (b5), line width=0.15mm] {};

\node[state, circle, black, draw, line width=0.15mm, minimum height=7pt, minimum width=7pt] (c1) [right=3cm of b1]{};
\node[state, circle, black, line width=0.15mm, minimum height=7pt, minimum width=7pt] (c2) [below= 0.75cm of c1]{};
\node[state, circle, black, line width=0.15mm, minimum height=7pt, minimum width=7pt] (c3) [below=0.75cm of c2]{};
\node[ circle, black, line width=0.15mm, minimum height=7pt, minimum width=7pt, fill=Gray, draw] (c4) [below=0.75cm of c3]{};
\node[state, circle, black, line width=0.15mm, minimum height=7pt, minimum width=7pt] (c5) [below=0.75cm of c4]{};
\node (v2) [below=0.75cm of c5]{$V_2$};

\node[rectangle, dashed,rounded corners = 3mm, inner sep=5pt, draw,  black, fit=(c1) (c5), line width=0.15mm] {};

\node[state, circle, black, draw, line width=0.15mm, minimum height=7pt, minimum width=7pt] (d1) [right=3cm of c1]{};
\node[state, circle, black, line width=0.15mm, minimum height=7pt, minimum width=7pt] (d2) [below= 0.75cm of d1]{};
\node[ circle, black, line width=0.15mm, minimum height=7pt, minimum width=7pt, fill=Gray, draw] (d3) [below=0.75cm of d2]{};
\node[state, circle, black, line width=0.15mm, minimum height=7pt, minimum width=7pt] (d4) [below=0.75cm of d3]{};
\node[state, circle, black, line width=0.15mm, minimum height=7pt, minimum width=7pt] (d5) [below=0.75cm of d4]{};
\node (v3) [below=0.75cm of d5]{$V_3$};

\node[rectangle, dashed, rounded corners = 3mm, inner sep=5pt, draw,  black, fit=(d1) (d5), line width=0.15mm] {};

\node[state, circle, black, draw, line width=0.15mm, minimum height=7pt, minimum width=7pt] (s)[left=3cm of a3]{};
\node (ss) [below=0.75cm of s]{$s$};
\node[state, circle, black, draw, line width=0.15mm, minimum height=7pt, minimum width=7pt] (t)[right=3cm of d3]{};
\node (tt) [below=0.75cm of t]{$t$};

\node (l1)[above right=1cm and 2cm of b1]{};

\draw[in=180, blue, line width=1.5pt] (b2) to (l1.center);
\draw[in=180, blue, line width=1.5pt] (b3) to (l1.center);
\draw[in=180, blue, line width=1.5pt] (b4) to (l1.center);
\draw[in=180, blue, line width=1.5pt,->,>=latex, rounded corners = 20pt] (l1.center) -| (t);

\node (l2)[above right=1cm and 2cm of c1]{};

\draw[in=180, blue, line width=1.5pt](c1) to (l2.center);
\draw[in=180, blue, line width=1.5pt](c3) to (l2.center);
\draw[in=180, blue, line width=1.5pt](c4) to (l2.center);

\node (l3)[above right=1cm and 2cm of d1]{};

\draw[in=180, blue, line width=1.5pt] (d2) to (l3.center);
\draw[in=180, blue, line width=1.5pt](d3) to (l3.center);
\draw[in=180, blue, line width=1.5pt](d4) to (l3.center);

\draw[->,>=latex,ForestGreen, line width=1.5pt] (d3) to (t);

\draw[->,>=latex, ForestGreen, line width=1.5pt, bend left] (s) to  (a1); 

\draw[->,>=latex, ForestGreen, bend right, line width=1.5pt] (a1) to  (b2) ; 
\draw[->,>=latex, red, bend right, line width=1.5pt] (a1) to (b4); 

\draw[->,>=latex, red, bend left, line width=1.5pt] (a1) to (b3); 
\draw[->,>=latex, ForestGreen, bend left, line width=1.5pt] (a1) to (b2); 

\draw[->,>=latex, red, bend right, line width=1.5pt] (b2) to (c1) node[near start]{}; 
\draw[->,>=latex, ForestGreen, bend right, line width=1.5pt] (b2) to (c4) node[near start]{}; 

\draw[->,>=latex, ForestGreen, bend left, line width=1.5pt] (b2) to (c4) node[near start]{}; 
\draw[->,>=latex, red, bend left, line width=1.5pt] (b2) to (c3) node[near start]{}; 

\draw[->,>=latex, ForestGreen, bend right, line width=1.5pt] (c4) to (d3) node[near start]{}; 
\draw[->,>=latex, red, bend right, line width=1.5pt] (c4) to (d2) node[near start]{}; 

\draw[->,>=latex, ForestGreen, bend left, line width=1.5pt] (c4) to (d3) node[near start]{}; 
\draw[->,>=latex, red, bend left, line width=1.5pt] (c4) to (d4) node[near start]{};

\node[ForestGreen,opacity=0, text opacity=1](L1)[above left=0.5cm and 1.5cm of v0]{$(5+1)^{4}$};

\node[ForestGreen,opacity=0, text opacity=1](L2)[above left=0.5cm and 1.5cm of v1]{$(5+1)^{3}$};
\node[red,opacity=0, text opacity=1](R2)[below=0.5cm of L2]{$(5+1)^{3}$};

\node[ForestGreen,opacity=0, text opacity=1](L3)[above left=0.5cm and 1.5cm of v2]{$(5+1)^{2}$};
\node[red,opacity=0, text opacity=1](R3)[below=0.5cm of L3]{$(5+1)^{2}$};
\node[blue,opacity=0, text opacity=1](B3)[above=0.5cm of L3]{$(5+1)^{3}$};
   
\node[ForestGreen,opacity=0, text opacity=1](L4)[above left=0.5cm and 1.5cm of v3]{$(5+1)^1$};
\node[red,opacity=0, text opacity=1](R4)[below=0.5cm of L4]{$(5+1)^1$};
\node[blue,opacity=0, text opacity=1](B4)[above=0.5cm of L4]{$(5+1)^{2}$};
             
\node[ForestGreen,opacity=0, text opacity=1](L5)[above right=0.5cm and 1.5cm of v3]{$(3-1)$};
\node[blue,opacity=0, text opacity=1](B5)[above=0.5cm of L5]{$(5+1)$};         
\end{tikzpicture}
    \caption{Illustration of the flow paths in $\PP$ in the proof of Lemma~\ref{lem:cut} for $n=5$ and $k=3$. The green edges belong to $P^*$ while red and blue edges are the edges that belong to a path in some $\PP_j$ but not $P^*$.
    The numbers denote the value of the flow sent over each outgoing edge in the corresponding layer with the same color. The value of this flow mod $(n+1)$ is $(\istar-1)$ where $\istar=3$. 
    }
    \label{fig:flow}
\end{figure}

We have the following auxiliary claim. 

\begin{claim}\label{clm:flow-target}
	For any $j \in [k]$, capacity of the edge $e=(u_{j-1},u_j)$ is $c(e) = 2 w_j$. 
\end{claim}
\begin{proof}
	Suppose $u_{j-1}=v^{j-1}_i$ and $u_j=v^j_{i'}$ and assume that $j$ is odd; the even $j$ case is symmetric. Since $j$ is odd, $y_{i'}$ is contained in both $A_{x_i}$ and $B_{x_i}$. Hence, there are two parallel edges from $u_{j-1}$ to $u_j$ each of weight $w_j$. 
	So the capacity of $(u_{j-1},u_j)$ is $2w_j$. \Qed{Claim~\ref{clm:flow-target}}
	
\end{proof}

We claim that $\PP$ gives a maximum flow in graph $G$. This proves the lemma as for all $j \in [k]$, the contribution of each path in $\PP_j$ to the flow$\mod (n+1)$ is $0$. Hence $P^*$ determines the value of the flow mod $(n+1)$ which is $(\istar-1)$ and
$\istar$ encodes the pointer $z_k$. The proof consists of the following two claims that ensure feasibility and optimality of $\PP$, respectively.  

\begin{claim}\label{clm:flow-feasible}
	$\PP$ induces a feasible flow in $G(V,E,w)$ with capacity $w_e$ on every edge $e \in E$. 
\end{claim}
\begin{proof}
Since all the paths in $\PP$ are $s$-$t$ paths, for any vertex in $V \setminus \set{s,t}$, the amount of flow going in that vertex is equal to the amount of flow going out of it. Hence, the flow is preserved on all vertices in $V \setminus \set{s,t}$. It thus remains to prove
that no edge is assigned a flow more than its capacity. 

Any edge $e$ \emph{not} in $P^*$ is contained in at most one path in $\PP$. For paths in $\PP_j$, these are edges $(u_{j-1},v^j_i)$ and $(v^{j}_i,t)$ for some $j \in [k]$ and $i \in [n]$. The amount of flow on these paths
is then equal to $w_j = w(v^{j}_i,t)$ by construction and hence the flow on these edges does not exceed their capacity.

We now prove the result for edges in $P^*$. First consider the edge $(u_k,t)$. There are two paths in $\PP$ that contain $(u_k,t)$: the path $P^*$ that carries $\istar-1$ units of flow and the path in $\PP_k$ that carries $w_k$ units of flow. As $u_k = v^{k}_{\istar}$, the capacity of the edge $(u_k,t)$ is also $w_k + (\istar-1)$ (as there are two edges connecting $v^k_{\istar}$ to $t$ with weights $w_k$ and $(\istar-1)$). Hence the flow on these edges also does not exceed their capacity. 

We next prove that for every $j \in [k]$, there are at most $2w_j$ units of flow passing through $(u_{j-1},u_j)$. By Claim~\ref{clm:flow-target}, this implies that the flow on these edges does not exceed capacity. 
The proof is by induction for $j=k$ down to $j=0$ in this order, where 
the base case is $(u_{k-1},u_k)$. All the paths that contain this edge also contain $(u_k,t)$, so there are $w_k+\istar-1<2w_k$ units of flow passing through this edge by the previous part of the argument.

For the induction step, 
consider the flow paths that contain $(u_{j-1},u_j)$. There is exactly one path in $\PP_j$ that contains this edge and that path carries $w_j$ units of flow by definition. There are 
also at most $n-1$ paths in $\PP_{j+1}$ that contain $(u_{j-1},u_j)$ but do not contain $(u_j,u_{j+1})$. The total flow these paths are carrying is at most $(n-1) \cdot w_{j+1}$. 
 All other paths in $\PP$ that contain $(u_{j-1},u_j)$ also contain $(u_j,u_{j+1})$ and hence by the induction hypothesis, these paths carry at most $2w_{j+1}$ units of flow. 
 So the total flow going through $(u_{j-1},u_j)$ is at most $w_j+(n-1)w_{j+1}+2w_{j+1} \le 2w_j$, proving the induction hypothesis. 

 Finally, consider the edge $(s,u_0)$. There are at most $n-1$ paths in $\PP_1$ that contain $(s,u_0)$ but not $(u_0,u_1)$. The total flow passing through these paths  is 
 at most $(n-1) \cdot w_1$. All other paths in $\PP$ contain $(u_0,u_1)$; these paths carry at most $2w_1$ units of flow as we proved above by induction. So the total flow passing through $(s,u_0)$
 is at most $(n-1)\cdot w_1 + 2w_1 = w_0$ which is equal to the capacity of $(s,u_0)$.
 \Qed{Claim~\ref{clm:flow-feasible}}

\end{proof}

\begin{claim}\label{clm:flow-optimal}
	There is no $s$-$t$ path in the residual graph of $G$ with respect to the flow paths in $\PP$. 
\end{claim}
\begin{proof}
We prove by induction that in the residual graph, $s$ can only reach $u_j$ in layer $V_j$ (strictly speaking, we will prove that if some other vertex in $V_j$ is reachable from $s$, then the path can only go through $t$, but in the end we will prove that $t$ is not 
reachable from $s$). 

The base case trivially holds as $s$ only has an outgoing edge to a single vertex in $V_0$, namely, the vertex $v^0_1 = u_0$. Furthermore, the outgoing edges of vertices in $V_0$ do not belong to any flow path in $\PP$. 
For the induction step, consider the layer $V_{j+1}$. By the induction hypothesis, $s$ can only reach $u_j$ in $V_j$. For any vertex $v^{j+1}_i$ which is not $u_{j+1}$, if the edge $(u_j,v^{j+1}_i)$ exists in $G$, then it is contained in a path in $\PP_{j+1}$ which carries 
$w_{j+1}$ units of flow. As the capacity of this edge is also $w_{j+1}$, the direction of this edge in the residual graph is from $v^{j+1}_i$ to $u_j$. Moreover, no outgoing edge of $v^{j+1}_i$ (except for the one going to $t$) is contained in any path in $\PP$. This means that in the 
residual graph, $v^{j+1}_i$ is not reachable from $s$, proving the induction hypothesis. 

By the above argument, the only vertex reachable from $s$ in $V_k$ is $u_k$. Now consider the sink $t$. For any $j \in [k]$, $(u_j,t)$ is contained in a path in $\PP_j$ and thus its flow matches its capacity. For edge $(u_k,t)$, there are two paths in $\PP$ that contain this 
edge, the first one is in $\PP_k$ which carries $w_k$ units of flow and the other is $P^*$ which carries $\istar-1$ units of flow. So $(u_k,t)=(v^k_{\istar},t)$ is also full. Thus $t$ is not reachable from $s$.
 \Qed{Claim~\ref{clm:flow-optimal}}

\end{proof}

Claims~\ref{clm:flow-feasible} and~\ref{clm:flow-optimal} prove that $\PP$ induces a maximum $s$-$t$ flow in $G$. We are now done as the amount of flow carried by all flow paths in $\PP$ is divisible by $n+1$ except for $P^*$. This is because
the flow carried by each path in $\PP_j$ for $j \in [k]$ is of weight $w_j$ and $(n+1)$ is a factor of $w_j$. As the flow carried by $P^*$ is $\istar-1$, the total flow in $\PP$ is $K \cdot (n+1) + (\istar-1)$ for some integer $K \geq 1$. 
By max-flow min-cut duality, $\wstar \mod (n+1) = \istar-1$. \Qed{Lemma~\ref{lem:cut}}

\end{proof}

We can now prove Theorem~\ref{thm:cut} using this reduction, the standard connection between space complexity of streaming algorithms and communication complexity, and our communication 
lower bound for hidden-pointer chasing in Theorem~\ref{thm:hpck}. 

\begin{proof}[Proof of Theorem~\ref{thm:cut}]
	Let $\alg$ be a $p$-pass streaming algorithm for computing the value of a minimum $s$-$t$ cut in weighted directed graphs. To avoid confusion, in the following, 
	we use $N$ to denote the number of vertices in the graph $G$ and $n$ for the size of universes in $\HPC$. Hence, our goal is to prove a lower bound of $\Omega(N^2/p^5)$ on the space complexity of $\alg$. 
	
	We give a reduction from $\HPCk$ for $k=2p+1$. Given an instance of $\HPCk$, the players first construct the graph $G(V,E,w)$ in the reduction of this section based on their inputs with no communication. 
	Next, they create a stream $\sigma$ of edges of $E$ such that edges depending on input to $P_D$ appear first, then $P_C$, $P_B$ and $P_A$ 
	in this order and input-independent edges appear last. The players run $\alg$ on $\sigma$ and communicate the state of $\alg$ between each other whenever
	necessary to compute the value of a minimum weighted $s$-$t$ cut in $G$. 
	
	By Lemma~\ref{lem:cut}, the value of the minimum $s$-$t$ cut in $G$ immediately determines the pointer $z_k$, hence proving the correctness of the protocol. 
	The number of phases and communication cost of this protocol can be determined as follows. 
	Each pass of the streaming algorithm translates into at most two phases in the protocol and hence the resulting protocol has strictly smaller than
	$k$ phases. The total communication by players in this protocol is at most $O(k \cdot S)$ where $S$ denotes the space complexity of $\alg$. As such, by Theorem~\ref{thm:hpck}, we have, 
	$k \cdot S = \Omega(n^2/k^2)$ which implies $S = \Omega(n^2/k^3)$. Since the total number of vertices in the graph is $N = O(k \cdot n)$ and $k = \Theta(p)$, we obtain a 
	lower bound of $\Omega(N^2/p^5)$ on the space complexity of $\alg$, finalizing the proof for the directed graphs. 
	
	To extend the results to undirected graphs, we can simply use the standard reduction of finding a maximum flow in directed graphs to finding a maximum flow
	 in undirected graphs described in, for example~\cite{Lin09} (see also Appendix C.2 in~\cite{SidfordT18}). This reduction works by turning each directed edge $e = (u,v)$ with capacity $c_e$ in the graph to three undirected edges 
	 $\set{s,v}$, $\set{u,v}$ and $\set{t,u}$ each with capacity $c_e$. It is then easy to see that after pushing an initial flow of $(s,v,u,t)$ with $c_e$ units of flow on every edge $(u,v)$, the residual graph obtained would be equivalent to the original directed graph. 
	 Hence, solving $s$-$t$ maximum flow on this undirected graph would also solve the problem for the original directed graph (see~\cite{Lin09,SidfordT18} for the formal proof). As thus reduction can be done on the graph $G(V,E,w)$ constructed
	 in this section with no further communication between the players, the results in this proof extend to undirected graphs as well, finalizing the proof.  
\end{proof}

\begin{remark}\label{rem:simple-graph}
	The reduction in this section creates a multi-graph $G$. However, we can easily transform this graph to a simple graph without changing the minimum cut value, while increasing the number of vertices by only a constant factor. 
	The transformation is as follows: turn any vertex $v^j_i$ in layer $V_j$ of the graph $G$ into three vertices $w^j_i$, $a^j_i$ and $b^j_i$. Connect $w^j_i$ to $a^j_i$ and $b^j_i$ with edges of weight $w_0$ (which is effectively infinity). 
	The input-independent edges going out of $v^j_i$ to $t$ now goes out of $w^j_i$ to $t$ instead. For any odd $j$, any edge $(v^j_i,v^{j+1}_{i'})$ is now turned into an edge $(a^j_i,w^{j+1}_{i'})$ if the edge
	was added because of $A_{x_i}$ and $(b^j_i,w^{j+1}_{i'})$ if it was added because of $B_{x_i}$. We do the same for even values of $j$ by using $C_{y_i}$ and $D_{y_i}$ instead. 
	It is easy to see that the weight of minimum $s$-$t$ is the same in this new graph and that this graph does not have any parallel edges anymore. 
\end{remark}

\newcommand{\mis}{\mathcal{M}}

\subsection{The Lexicographically-First MIS Problem}\label{sec:mis}

Proof of Theorem~\ref{thm:mis} is also by a reduction from the hidden-pointer chasing (\HPC) problem. 
We turn an instance $(\bA,\bB,\bC,\bD)$ of \HPCk over universes $\XX$ and $\YY$, into an undirected graph $G(V,E)$. 
The reduction is as follows (see Figure~\ref{fig:mis} for an example): 

\begin{itemize}
	\item The vertex-set $V$ of $G$ is partitioned into $k+1$ layers $V_0,\ldots,V_k$ each of size $n$ plus a single vertex $s$ (hence $G$ has $(k+1)n + 1$ vertices). 
	We denote the $i$-th vertex in layer $V_j$ by $v^j_i$. 
	In the lexicographic order, the vertices in layer $V_0$ appear first, followed by vertices in $V_1,\ldots,V_k$ in this order. Inside each layer $V_j$, the ordering is by the index, i.e., in the order $v^{j}_1,\ldots,v^{j}_n$. 
	
	\item The edge-set $E$ contains the following edges:
	\begin{itemize}
		\item vertex $v^0_1$ is connected to all other vertices in $V^0$. 
		\item for all $i \in [n]$, if $A_{x_i}\in \bA$ (resp. $B_{x_i} \in \bB$) does \emph{not} contain $y_{i'} \in \YY$, we connect $v^{j}_i$ in layer $V_j$ to $v^{j+1}_{i'}$ in layer $V_{j+1}$ for every \emph{even} $0 \leq j < k$.
		\item for all $i \in [n]$, if $C_{y_i}\in \bC$ (resp. $D_{y_i} \in \bD$) does \emph{not} contain $x_{i'} \in \XX$, we connect $v^{j}_i$ in layer $V_j$ to $v^{j+1}_{i'}$ in layer $V_{j+1}$ for every \emph{odd} $0 < j < k$. 
	\end{itemize}
\end{itemize}
\noindent
This concludes the description of the graph $G(V,E)$ in the reduction.
 It is straightforward to verify that this graph can be constructed from an instance $(\bA,\bB,\bC,\bD)$ with no communication between the players. 
We now establish the correctness of the reduction. 

\begin{figure}[h!]
    \centering
\begin{tikzpicture}[ auto ,node distance =1cm and 2cm , on grid , semithick , state/.style ={ circle ,top color =white , bottom color = white , draw, black , text=black}, every node/.style={inner sep=0,outer sep=0}]

\node[circle, black, line width=0.15mm, minimum height=7pt, minimum width=7pt, fill=Gray, draw] (a1){};
\node[state, circle, black, line width=0.15mm, minimum height=7pt, minimum width=7pt] (a2) [below= 0.75cm of a1]{};
\node[state, circle, black, line width=0.15mm, minimum height=7pt, minimum width=7pt] (a3) [below=0.75cm of a2]{};
\node[state, circle, black, line width=0.15mm, minimum height=7pt, minimum width=7pt] (a4) [below=0.75cm of a3]{};
\node[state, circle, black, line width=0.15mm, minimum height=7pt, minimum width=7pt] (a5) [below=0.75cm of a4]{};
\node (v0) [below=0.75cm of a5]{$V_0$};


\node[rectangle, dashed, rounded corners = 3mm, inner sep=5pt, draw,  black, fit=(a1) (a5), line width=0.15mm] {};

\node[state, circle, black, draw, line width=0.15mm, minimum height=7pt, minimum width=7pt] (b1) [right=3cm of a1]{};
\node[ circle, black, line width=0.15mm, minimum height=7pt, minimum width=7pt, fill=Gray, draw] (b2) [below= 0.75cm of b1]{};
\node[state, circle, black, line width=0.15mm, minimum height=7pt, minimum width=7pt] (b3) [below=0.75cm of b2]{};
\node[state, circle, black, line width=0.15mm, minimum height=7pt, minimum width=7pt] (b4) [below=0.75cm of b3]{};
\node[state, circle, black, line width=0.15mm, minimum height=7pt, minimum width=7pt] (b5) [below=0.75cm of b4]{};
\node (v1) [below=0.75cm of b5]{$V_1$};

\node[rectangle, dashed,rounded corners = 3mm, inner sep=5pt, draw,  black, fit=(b1) (b5), line width=0.15mm] {};

\node[state, circle, black, draw, line width=0.15mm, minimum height=7pt, minimum width=7pt] (c1) [right=3cm of b1]{};
\node[state, circle, black, line width=0.15mm, minimum height=7pt, minimum width=7pt] (c2) [below= 0.75cm of c1]{};
\node[state, circle, black, line width=0.15mm, minimum height=7pt, minimum width=7pt] (c3) [below=0.75cm of c2]{};
\node[ circle, black, line width=0.15mm, minimum height=7pt, minimum width=7pt, fill=Gray, draw] (c4) [below=0.75cm of c3]{};
\node[state, circle, black, line width=0.15mm, minimum height=7pt, minimum width=7pt] (c5) [below=0.75cm of c4]{};
\node (v2) [below=0.75cm of c5]{$V_2$};

\node[rectangle, dashed,rounded corners = 3mm, inner sep=5pt, draw,  black, fit=(c1) (c5), line width=0.15mm] {};

\node[state, circle, black, draw, line width=0.15mm, minimum height=7pt, minimum width=7pt] (d1) [right=3cm of c1]{};
\node[state, circle, black, line width=0.15mm, minimum height=7pt, minimum width=7pt] (d2) [below= 0.75cm of d1]{};
\node[ circle, black, line width=0.15mm, minimum height=7pt, minimum width=7pt, fill=Gray, draw] (d3) [below=0.75cm of d2]{};
\node[state, circle, black, line width=0.15mm, minimum height=7pt, minimum width=7pt] (d4) [below=0.75cm of d3]{};
\node[state, circle, black, line width=0.15mm, minimum height=7pt, minimum width=7pt] (d5) [below=0.75cm of d4]{};
\node (v3) [below=0.75cm of d5]{$V_3$};

\node[rectangle, dashed, rounded corners = 3mm, inner sep=5pt, draw,  black, fit=(d1) (d5), line width=0.15mm] {};

\draw[black, bend right, line width=0.5pt] (a1) to (a2) node[near start]{}; 
\draw[black, bend right, line width=0.5pt] (a1) to (a3) node[near start]{}; 
\draw[black, bend right, line width=0.5pt] (a1) to (a4) node[near start]{}; 
\draw[black, bend right, line width=0.5pt] (a1) to (a5) node[near start]{}; 

\draw[blue, bend right, line width=1.5pt] (a1) to (b1) node[near start]{}; 
\draw[blue, bend right, line width=1.5pt] (a1) to (b3) node[near start]{}; 
\draw[blue, bend right, line width=1.5pt] (a1) to (b5) node[near start]{}; 

\draw[red, bend left, line width=1.5pt] (a1) to (b1) node[near start]{}; 
\draw[red, bend left, line width=1.5pt] (a1) to (b4) node[near start]{}; 
\draw[red, bend left, line width=1.5pt] (a1) to (b5) node[near start]{}; 

\draw[ForestGreen, bend right, line width=1.5pt] (b2) to (c2) node[near start]{}; 
\draw[ForestGreen, bend right, line width=1.5pt] (b2) to (c3) node[near start]{}; 
\draw[ForestGreen, bend right, line width=1.5pt] (b2) to (c5) node[near start]{}; 

\draw[YellowOrange, bend left, line width=1.5pt] (b2) to (c1) node[near start]{}; 
\draw[YellowOrange, bend left, line width=1.5pt] (b2) to (c2) node[near start]{}; 
\draw[YellowOrange, bend left, line width=1.5pt] (b2) to (c5) node[near start]{};

\draw[blue, bend right, line width=1.5pt] (c4) to (d1) node[near start]{}; 
\draw[blue, bend right, line width=1.5pt] (c4) to (d4) node[near start]{}; 
\draw[blue, bend right, line width=1.5pt] (c4) to (d5) node[near start]{}; 

\draw[red, bend left, line width=1.5pt] (c4) to (d1) node[near start]{}; 
\draw[red, bend left, line width=1.5pt] (c4) to (d2) node[near start]{}; 
\draw[red, bend left, line width=1.5pt] (c4) to (d5) node[near start]{};

\end{tikzpicture}
    \caption{Illustration of the graph in reducing lexicographically-first MIS from $\HPC_3$ with $n=5$. The black (thin) edges incident on $s$ are input-independent while blue, red , brown, and green (thick) edges depend on the inputs of $P_A$, $P_B$, $P_C$, and $P_D$, respectively. 
    The marked nodes denote the vertices corresponding to pointers $z_0,\ldots,z_3$.
    The edges incident on ``non-pointer'' vertices are omitted. This construction has parallel edges but similar to Remark~\ref{rem:simple-graph}, we can remove them.
    }
    \label{fig:mis}
\end{figure}
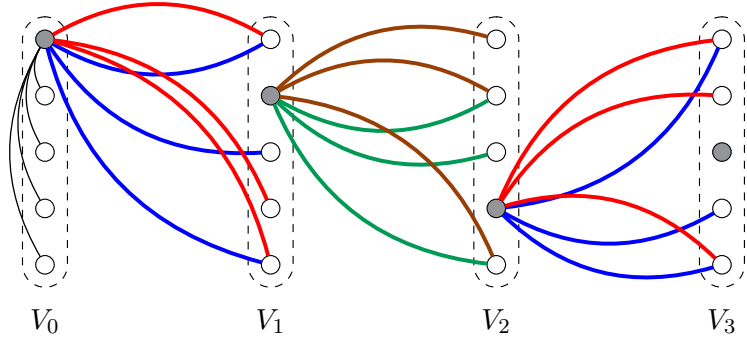

\begin{lemma} \label{lem:redmis}
In the reduction above, the pointer $z_k = x_i$ (resp. $z_k = y_i$) when $k$ is even (resp. odd) iff $v^{k}_i$ belongs to the lexicographically-first \textnormal{MIS} of $G$.
\end{lemma}

\begin{proof}
Let $\mis$ be the lexicographically-first MIS of $G$. 
We prove by induction that for any even (resp. odd) $j \in \set{0,1,\ldots,k}$, there is a unique vertex $v^{j}_i$ from layer $V_j$ that belongs to $\mis$ and that vertex corresponds to the pointer $z_j$, namely, $x_i = z_j$ (resp. $y_i = z_j$). 

The base case is trivial since $z_0=x_1$, $v^0_1$ appears first in the lexicographical ordering of vertices, and $v^0_1$ is connected to all vertices in layer $V_0$. We now prove the induction step. Suppose $j$ is even; the other case is symmetric. By induction hypothesis, $v^j_i$ is the unique
vertex in layer $V_j$ that belongs to $\mis$ where $x_i = z_j$. By construction of $G$, $v^j_i$ is connected to all vertices in layer $j+1$ except for the vertex $v^{j+1}_{i'}$, where $\set{y_{i'}} = A_{x_i} \cap B_{x_i}$. Hence, $v^{j+1}_{i'}$ 
is the unique index in $V_{j+1}$ that belongs to $\mis$. The proof is concluded by noting that $z_{j+1} = y_{i'}$ by definition. 
\end{proof}

Proof of Theorem~\ref{thm:mis} now follows from Lemma~\ref{lem:redmis} and Theorem~\ref{thm:hpck} the same exact way as in proof of Theorem~\ref{thm:cut} in the last section. For completeness, we present this proof here. 

\begin{proof}[Proof of Theorem~\ref{thm:mis}]
	Let $\alg$ be a $p$-pass streaming algorithm for finding the lexicographically-first MIS of an undirected graph. To avoid confusion, in the following, 
	we use $N$ to denote the number of vertices in the graph $G$ and $n$ for the size of universes in $\HPC$. Hence, our goal is to prove a lower bound of $\Omega(N^2/p^5)$ on the space complexity of $\alg$. 
	
	We give a reduction from $\HPCk$ for $k=2p+1$. Given an instance of $\HPCk$, the players first construct the graph $G(V,E)$ in the reduction of this section based on their inputs with no communication. 
	Next, they create a stream $\sigma$ of edges of $E$ such that edges depending on input to $P_D$ appear first, then $P_C$, $P_B$ and $P_A$ 
	in this order and input-independent edges appear last. The players then run $\alg$ on $\sigma$ and communicate the state of $\alg$ between each other whenever
	necessary to find the lexicographically-first MIS $\mis$ of $G$. 
	
	By Lemma~\ref{lem:redmis}, the vertex in layer $V_k$ of $G$ that belongs to $\mis$ determines the pointer $z_k$, hence proving the correctness of the protocol. 
	The number of phases and communication cost of this protocol can be determined as follows. 
	Each pass of the streaming algorithm translates into at most two phases in the protocol and hence the resulting protocol has strictly smaller than
	$k$ phases. The total communication by players in this protocol is at most $O(k \cdot S)$ where $S$ denotes the space complexity of $\alg$. As such, by Theorem~\ref{thm:hpck}, we have, 
	$k \cdot S = \Omega(n^2/k^2)$ which implies $S = \Omega(n^2/k^3)$. Since the total number of vertices in the graph is $N = O(k \cdot n)$ and $k = \Theta(p)$, we obtain a 
	lower bound of $\Omega(N^2/p^5)$ on the space complexity of $\alg$, finalizing the proof.	
\end{proof}

We also note that similar to the previous section, we can also turn the graph $G$ in the reduction of this section to a simple graph with no parallel edges using essentially the same gadget. We omit the details.

\section{A Lower Bound for Submodular Function Minimization}\label{sec:SFM}

A non-monotone set-function $f:U \rightarrow [M]$ is called \emph{submodular} iff for every $A \subseteq B \subseteq U$ and for every element $i \notin B$, $f(A \cup \set{i}) - f(A) \geq f(B \cup \set{i}) - f(B)$. 
In the submodular function minimization (SFM) problem, we assume access to an \emph{evaluation oracle} for $f$ that given any set $S \subseteq [n]$ returns $f(S)$; the goal is to return a set $S^*$ that minimizes $f(S^*)$. 
We say that an algorithm for SFM is \emph{$k$-adaptive} iff it makes its queries to the evaluation oracle in at most $k$ rounds of adaptive queries where the queries in each round are performed in parallel. 
We prove the following theorem on the query complexity of $k$-adaptive algorithms for SFM. 

\begin{theorem}\label{thm:sfm}
	For any $k \geq 1$, any $k$-round adaptive algorithm for submodular function minimization that with constant probability outputs the minimum value of a non-monotone submodular function $f:U \rightarrow [M]$
	for $\card{U} = N$ and $M = O(N^{k+1})$ requires $\Omega(\frac{N^2}{k^5 \cdot \log{N}})$ queries to the evaluation oracle. 
\end{theorem}
\begin{proof}
	The proof is by a reduction from $\HPC_{3k}$ similar to the proof of Theorem~\ref{thm:cut} using the fact that cut functions are submodular. 
	
	Given an instance of $\HPC_{3k}$ problem, we construct the weighted graph $G(V,E,w)$ in the reduction of Theorem~\ref{thm:cut}. Let $U := V \setminus \set{s,t}$. We define a set-function $f: U \rightarrow [M]$ 
	where for any $S \subseteq U$, $f(S)$ is defined to be the value of the cut $(\set{s} \cup S, \set{t} \cup U \setminus S)$ in $G$, i.e., the total weight of the edges going from $\set{s} \cup S$ to $V \setminus (S \cup \set{s})$. 
	We set $M = \sum_{e \in E} w_e$ and hence clearly $f(S) \leq M$. Note that by construction of $G$, $M= O(n^{k+1})$ and $N = \card{U} = O(n \cdot k)$. 
	The function $f$ is a well-known submodular function. Also, it is easy to see that minimizing $f$ corresponds to computing the minimum weighted $s$-$t$ cut in $G$. 
	
	Now let $\alg$ be a $k$-adaptive algorithm for minimizing $f$. We turn $\alg$ into a protocol for $\HPC_{3k}$ with strictly smaller than $3k$ phases. We first argue that any query asked by $\alg$ can be answered by the players in $\HPC_{3k}$ using 
	$O(\log{M})$ communication. Indeed, if $\alg$ asks for a query $S$, then each player needs to look at her input and determine the weights of the edges crossing the cut $\set{s} \cup S$, and communicate it to other players with 
	$O(\log{M})$ bits of communication. The players can on their own also add the weights of the input-independent edges and hence each player knows the answer to $f(S)$. Using this, the players can simulate running $\alg$ on $f$ 
	and by Lemma~\ref{lem:cut} solve $\HPC_{3k}$ using $O(Q \cdot \log{M})$ communication where $Q$ denotes the query complexity of $\alg$ (the players use public randomness to simulate randomness of $\alg$). Moreover, 
	each round of adaptive queries translates into at most two phases in the protocol. As such, the protocol has $<3k$ phases and hence by Theorem~\ref{thm:hpck}, we have that 
	\begin{align*}
		Q \cdot \log{M} = \Omega(\frac{n^2}{k^2}) \implies Q = \Omega(\frac{N^2}{k^5 \cdot \log{n}}),
	\end{align*}
	finalizing the proof. 
\end{proof}

We conclude with the following immediate corollary of Theorem~\ref{thm:sfm}.

\begin{corollary}\label{cor:sfm}
	For any constant $\delta \in (0,1)$, there exists an $\eps:= \eps(\delta)$ in $(0,1)$ such that any algorithm for submodular function minimization on a universe of size $N$ with query complexity $N^{2-\delta}$ requires at least $N^{\eps}$ rounds of adaptive queries to succeed with constant probability. 
\end{corollary}

The proof of this corollary is by simply setting $\eps := \delta/6$, and then applying Theorem~\ref{thm:sfm} with $k=N^{\eps}$ to obtain the desired bounds.

\bibliographystyle{abbrv}
\bibliography{general}

\begin{thebibliography}{100}

\bibitem{AbboudCKP19}
A.~Abboud, K.~Censor{-}Hillel, S.~Khoury, and A.~Paz.
\newblock Smaller cuts, higher lower bounds.
\newblock {\em CoRR}, abs/1901.01630, 2019.

\bibitem{Ablayev93}
F.~M. Ablayev.
\newblock Lower bounds for one-way probabilistic communication complexity.
\newblock In {\em Automata, Languages and Programming, 20nd International
  Colloquium, ICALP93, Lund, Sweden, July 5-9, 1993, Proceedings}, pages
  241--252, 1993.

\bibitem{AhnGM12a}
K.~J. Ahn, S.~Guha, and A.~McGregor.
\newblock Analyzing graph structure via linear measurements.
\newblock In {\em Proceedings of the Twenty-third Annual ACM-SIAM Symposium on
  Discrete Algorithms}, SODA '12, pages 459--467. SIAM, 2012.

\bibitem{AhnGM12b}
K.~J. Ahn, S.~Guha, and A.~McGregor.
\newblock Graph sketches: sparsification, spanners, and subgraphs.
\newblock In {\em Proceedings of the 31st {ACM} {SIGMOD-SIGACT-SIGART}
  Symposium on Principles of Database Systems, {PODS} 2012, Scottsdale, AZ,
  USA, May 20-24, 2012}, pages 5--14, 2012.

\bibitem{AlonBI86}
N.~Alon, L.~Babai, and A.~Itai.
\newblock A fast and simple randomized parallel algorithm for the maximal
  independent set problem.
\newblock {\em J. Algorithms}, 7(4):567--583, 1986.

\bibitem{AlonMS96}
N.~Alon, Y.~Matias, and M.~Szegedy.
\newblock The space complexity of approximating the frequency moments.
\newblock In {\em STOC}, pages 20--29. ACM, 1996.

\bibitem{AlonNRW15}
N.~Alon, N.~Nisan, R.~Raz, and O.~Weinstein.
\newblock Welfare maximization with limited interaction.
\newblock In {\em {IEEE} 56th Annual Symposium on Foundations of Computer
  Science, {FOCS} 2015, Berkeley, CA, USA, 17-20 October, 2015}, pages
  1499--1512, 2015.

\bibitem{Assadi17ca}
S.~Assadi.
\newblock Combinatorial auctions do need modest interaction.
\newblock In {\em Proceedings of the 2017 {ACM} Conference on Economics and
  Computation, {EC} '17, Cambridge, MA, USA, June 26-30, 2017}, pages 145--162,
  2017.

\bibitem{Assadi17sc}
S.~Assadi.
\newblock Tight space-approximation tradeoff for the multi-pass streaming set
  cover problem.
\newblock In {\em Proceedings of the 36th {ACM} {SIGMOD-SIGACT-SIGAI} Symposium
  on Principles of Database Systems, {PODS} 2017, Chicago, IL, USA, May 14-19,
  2017}, pages 321--335, 2017.

\bibitem{AssadiCK19}
S.~Assadi, Y.~Chen, and S.~Khanna.
\newblock Sublinear algorithms for ({\(\Delta\)} + 1) vertex coloring.
\newblock In {\em Proceedings of the Thirtieth Annual {ACM-SIAM} Symposium on
  Discrete Algorithms, {SODA} 2019, San Diego, California, USA, January 6-9,
  2019}, pages 767--786, 2019.

\bibitem{AssadiK18}
S.~Assadi and S.~Khanna.
\newblock Tight bounds on the round complexity of the distributed maximum
  coverage problem.
\newblock In {\em Proceedings of the Twenty-Nine Annual {ACM-SIAM} Symposium on
  Discrete Algorithms, {SODA} 2018}, 2018.

\bibitem{AssadiKL16}
S.~Assadi, S.~Khanna, and Y.~Li.
\newblock Tight bounds for single-pass streaming complexity of the set cover
  problem.
\newblock In {\em Proceedings of the 48th Annual {ACM} {SIGACT} Symposium on
  Theory of Computing, {STOC} 2016, Cambridge, MA, USA, June 18-21, 2016},
  pages 698--711, 2016.

\bibitem{AssadiKL17}
S.~Assadi, S.~Khanna, and Y.~Li.
\newblock On estimating maximum matching size in graph streams.
\newblock In {\em Proceedings of the Twenty-Eighth Annual {ACM-SIAM} Symposium
  on Discrete Algorithms, {SODA} 2017, Barcelona, Spain, Hotel Porta Fira,
  January 16-19}, pages 1723--1742, 2017.

\bibitem{AssadiKLY16}
S.~Assadi, S.~Khanna, Y.~Li, and G.~Yaroslavtsev.
\newblock Maximum matchings in dynamic graph streams and the simultaneous
  communication model.
\newblock In {\em Proceedings of the Twenty-Seventh Annual {ACM-SIAM} Symposium
  on Discrete Algorithms, {SODA} 2016, Arlington, VA, USA, January 10-12,
  2016}, pages 1345--1364, 2016.

\bibitem{BabaiFS86}
L.~Babai, P.~Frankl, and J.~Simon.
\newblock Complexity classes in communication complexity theory (preliminary
  version).
\newblock In {\em 27th Annual Symposium on Foundations of Computer Science,
  27-29 October 1986}, pages 337--347, 1986.

\bibitem{BalkanskiBS18}
E.~Balkanski, A.~Breuer, and Y.~Singer.
\newblock Non-monotone submodular maximization in exponentially fewer
  iterations.
\newblock {\em CoRR}, abs/1807.11462. To appear in NIPS 2018., 2018.

\bibitem{BalkanskiRS16}
E.~Balkanski, A.~Rubinstein, and Y.~Singer.
\newblock The power of optimization from samples.
\newblock In {\em Advances in Neural Information Processing Systems 29: Annual
  Conference on Neural Information Processing Systems 2016, December 5-10,
  2016, Barcelona, Spain}, pages 4017--4025, 2016.

\bibitem{BalkanskiRS17}
E.~Balkanski, A.~Rubinstein, and Y.~Singer.
\newblock The limitations of optimization from samples.
\newblock In {\em Proceedings of the 49th Annual {ACM} {SIGACT} Symposium on
  Theory of Computing, {STOC} 2017, Montreal, QC, Canada, June 19-23, 2017},
  pages 1016--1027, 2017.

\bibitem{BalkansiRS18}
E.~Balkanski, A.~Rubinstein, and Y.~Singer.
\newblock An exponential speedup in parallel running time for submodular
  maximization without loss in approximation.
\newblock {\em CoRR}, abs/1804.06355. To appear in SODA 2019., 2018.

\bibitem{BalkanskiS17}
E.~Balkanski and Y.~Singer.
\newblock Minimizing a submodular function from samples.
\newblock In {\em Advances in Neural Information Processing Systems 30: Annual
  Conference on Neural Information Processing Systems 2017, 4-9 December 2017,
  Long Beach, CA, {USA}}, pages 814--822, 2017.

\bibitem{BalkanskiS18}
E.~Balkanski and Y.~Singer.
\newblock The adaptive complexity of maximizing a submodular function.
\newblock In {\em Proceedings of the 50th Annual {ACM} {SIGACT} Symposium on
  Theory of Computing, {STOC} 2018, Los Angeles, CA, USA, June 25-29, 2018},
  pages 1138--1151, 2018.

\bibitem{BalkanskiS18m}
E.~Balkanski and Y.~Singer.
\newblock Parallelization does not accelerate convex optimization: Adaptivity
  lower bounds for non-smooth convex minimization.
\newblock {\em CoRR}, abs/1808.03880, 2018.

\bibitem{Bar-YossefJKS02}
Z.~Bar{-}Yossef, T.~S. Jayram, R.~Kumar, and D.~Sivakumar.
\newblock An information statistics approach to data stream and communication
  complexity.
\newblock In {\em 43rd Symposium on Foundations of Computer Science {(FOCS}
  2002), 16-19 November 2002, Proceedings}, pages 209--218, 2002.

\bibitem{Bar-YossefKS02}
Z.~Bar{-}Yossef, R.~Kumar, and D.~Sivakumar.
\newblock Reductions in streaming algorithms, with an application to counting
  triangles in graphs.
\newblock In {\em Proceedings of the Thirteenth Annual {ACM-SIAM} Symposium on
  Discrete Algorithms, January 6-8, 2002, San Francisco, CA, {USA.}}, pages
  623--632, 2002.

\bibitem{BarakBCR10}
B.~Barak, M.~Braverman, X.~Chen, and A.~Rao.
\newblock How to compress interactive communication.
\newblock In {\em Proceedings of the 42nd {ACM} Symposium on Theory of
  Computing, {STOC} 2010, 5-8 June 2010}, pages 67--76, 2010.

\bibitem{BateniEM17}
M.~Bateni, H.~Esfandiari, and V.~S. Mirrokni.
\newblock Almost optimal streaming algorithms for coverage problems.
\newblock In {\em Proceedings of the 29th {ACM} Symposium on Parallelism in
  Algorithms and Architectures, {SPAA} 2017, Washington DC, USA, July 24-26,
  2017}, pages 13--23, 2017.

\bibitem{BeckerKKL17}
R.~Becker, A.~Karrenbauer, S.~Krinninger, and C.~Lenzen.
\newblock Near-optimal approximate shortest paths and transshipment in
  distributed and streaming models.
\newblock In {\em 31st International Symposium on Distributed Computing, {DISC}
  2017, October 16-20, 2017, Vienna, Austria}, pages 7:1--7:16, 2017.

\bibitem{BeraC17}
S.~K. Bera and A.~Chakrabarti.
\newblock Towards tighter space bounds for counting triangles and other
  substructures in graph streams.
\newblock In {\em 34th Symposium on Theoretical Aspects of Computer Science,
  {STACS} 2017, March 8-11, 2017, Hannover, Germany}, pages 11:1--11:14, 2017.

\bibitem{BlellochFS12}
G.~E. Blelloch, J.~T. Fineman, and J.~Shun.
\newblock Greedy sequential maximal independent set and matching are parallel
  on average.
\newblock In {\em 24th {ACM} Symposium on Parallelism in Algorithms and
  Architectures, {SPAA} '12, Pittsburgh, PA, USA, June 25-27, 2012}, pages
  308--317, 2012.

\bibitem{BravermanEOPV13}
M.~Braverman, F.~Ellen, R.~Oshman, T.~Pitassi, and V.~Vaikuntanathan.
\newblock A tight bound for set disjointness in the message-passing model.
\newblock In {\em 54th Annual {IEEE} Symposium on Foundations of Computer
  Science, {FOCS} 2013, 26-29 October, 2013, Berkeley, CA, {USA}}, pages
  668--677, 2013.

\bibitem{BravermanGPW13}
M.~Braverman, A.~Garg, D.~Pankratov, and O.~Weinstein.
\newblock From information to exact communication.
\newblock In {\em Symposium on Theory of Computing Conference, STOC'13, June
  1-4, 2013}, pages 151--160, 2013.

\bibitem{BravermanMW18}
M.~Braverman, J.~Mao, and S.~M. Weinberg.
\newblock On simultaneous two-player combinatorial auctions.
\newblock In {\em Proceedings of the Twenty-Ninth Annual {ACM-SIAM} Symposium
  on Discrete Algorithms, {SODA} 2018, January 7-10, 2018}, pages 2256--2273,
  2018.

\bibitem{BravermanM13}
M.~Braverman and A.~Moitra.
\newblock An information complexity approach to extended formulations.
\newblock In {\em Symposium on Theory of Computing Conference, STOC'13, June
  1-4, 2013}, pages 161--170, 2013.

\bibitem{BravermanO17}
M.~Braverman and R.~Oshman.
\newblock A rounds vs. communication tradeoff for multi-party set disjointness.
\newblock In {\em 58th {IEEE} Annual Symposium on Foundations of Computer
  Science, {FOCS} 2017, Berkeley, CA, USA, October 15-17, 2017}, pages
  144--155, 2017.

\bibitem{BravermanR11}
M.~Braverman and A.~Rao.
\newblock Information equals amortized communication.
\newblock In {\em {IEEE} 52nd Annual Symposium on Foundations of Computer
  Science, {FOCS} 2011, October 22-25, 2011}, pages 748--757, 2011.

\bibitem{BrodyCKWY14}
J.~Brody, A.~Chakrabarti, R.~Kondapally, D.~P. Woodruff, and G.~Yaroslavtsev.
\newblock Beyond set disjointness: the communication complexity of finding the
  intersection.
\newblock In {\em {ACM} Symposium on Principles of Distributed Computing,
  {PODC} '14, Paris, France, July 15-18, 2014}, pages 106--113, 2014.

\bibitem{ChakrabartiCKM10}
A.~Chakrabarti, G.~Cormode, R.~Kondapally, and A.~McGregor.
\newblock Information cost tradeoffs for augmented index and streaming language
  recognition.
\newblock In {\em 51th Annual {IEEE} Symposium on Foundations of Computer
  Science, {FOCS} 2010, October 23-26, 2010, Las Vegas, Nevada, {USA}}, pages
  387--396, 2010.

\bibitem{ChakrabartiCM08}
A.~Chakrabarti, G.~Cormode, and A.~McGregor.
\newblock Robust lower bounds for communication and stream computation.
\newblock In {\em Proceedings of the 40th Annual {ACM} Symposium on Theory of
  Computing, May 17-20, 2008}, pages 641--650, 2008.

\bibitem{ChakrabartiK14}
A.~Chakrabarti and S.~Kale.
\newblock Submodular maximization meets streaming: Matchings, matroids, and
  more.
\newblock In {\em Integer Programming and Combinatorial Optimization - 17th
  International Conference, {IPCO} 2014, Bonn, Germany, June 23-25, 2014.
  Proceedings}, pages 210--221, 2014.

\bibitem{ChakrabartiSWY01}
A.~Chakrabarti, Y.~Shi, A.~Wirth, and A.~C. Yao.
\newblock Informational complexity and the direct sum problem for simultaneous
  message complexity.
\newblock In {\em 42nd Annual Symposium on Foundations of Computer Science,
  {FOCS} 2001, 14-17 October 2001}, pages 270--278, 2001.

\bibitem{ChakrabartiW16}
A.~Chakrabarti and A.~Wirth.
\newblock Incidence geometries and the pass complexity of semi-streaming set
  cover.
\newblock In {\em Proceedings of the Twenty-Seventh Annual {ACM-SIAM} Symposium
  on Discrete Algorithms, {SODA} 2016, Arlington, VA, USA, January 10-12,
  2016}, pages 1365--1373, 2016.

\bibitem{ChakrabartyLSW17}
D.~Chakrabarty, Y.~T. Lee, A.~Sidford, and S.~C. Wong.
\newblock Subquadratic submodular function minimization.
\newblock In {\em Proceedings of the 49th Annual {ACM} {SIGACT} Symposium on
  Theory of Computing, {STOC} 2017, Montreal, QC, Canada, June 19-23, 2017},
  pages 1220--1231, 2017.

\bibitem{ChattopadhyayM15}
A.~Chattopadhyay and S.~Mukhopadhyay.
\newblock Tribes is hard in the message passing model.
\newblock In {\em 32nd International Symposium on Theoretical Aspects of
  Computer Science, {STACS} 2015, March 4-7, 2015, Garching, Germany}, pages
  224--237, 2015.

\bibitem{Cook85}
S.~A. Cook.
\newblock A taxonomy of problems with fast parallel algorithms.
\newblock {\em Information and Control}, 64(1-3):2--21, 1985.

\bibitem{CormodeDK18a}
G.~Cormode, J.~Dark, and C.~Konrad.
\newblock Approximating the caro-wei bound for independent sets in graph
  streams.
\newblock In {\em Combinatorial Optimization - 5th International Symposium,
  {ISCO} 2018, Marrakesh, Morocco, April 11-13, 2018, Revised Selected Papers},
  pages 101--114, 2018.

\bibitem{CormodeDK18}
G.~Cormode, J.~Dark, and C.~Konrad.
\newblock Independent sets in vertex-arrival streams.
\newblock {\em CoRR}, abs/1807.08331, 2018.

\bibitem{CormodeJ17}
G.~Cormode and H.~Jowhari.
\newblock A second look at counting triangles in graph streams (corrected).
\newblock {\em Theor. Comput. Sci.}, 683:22--30, 2017.

\bibitem{CoverT06}
T.~M. Cover and J.~A. Thomas.
\newblock {\em Elements of information theory {(2.} ed.)}.
\newblock Wiley, 2006.

\bibitem{Cunningham85}
W.~H. Cunningham.
\newblock On submodular function minimization.
\newblock {\em Combinatorica}, 5(3):185--192, 1985.

\bibitem{DemaineIMV14}
E.~D. Demaine, P.~Indyk, S.~Mahabadi, and A.~Vakilian.
\newblock On streaming and communication complexity of the set cover problem.
\newblock In {\em Distributed Computing - 28th International Symposium, {DISC}
  2014, Austin, TX, USA, October 12-15, 2014. Proceedings}, pages 484--498,
  2014.

\bibitem{DobzinskiNO14}
S.~Dobzinski, N.~Nisan, and S.~Oren.
\newblock Economic efficiency requires interaction.
\newblock In {\em Symposium on Theory of Computing, {STOC} 2014, New York, NY,
  USA, May 31 - June 03, 2014}, pages 233--242, 2014.

\bibitem{DurisGS84}
P.~Duris, Z.~Galil, and G.~Schnitger.
\newblock Lower bounds on communication complexity.
\newblock In {\em Proceedings of the 16th Annual {ACM} Symposium on Theory of
  Computing, April 30 - May 2, 1984, Washington, DC, {USA}}, pages 81--91,
  1984.

\bibitem{EggertKS09}
S.~Eggert, L.~Kliemann, and A.~Srivastav.
\newblock Bipartite graph matchings in the semi-streaming model.
\newblock In {\em Algorithms - {ESA} 2009, 17th Annual European Symposium,
  September 7-9, 2009. Proceedings}, pages 492--503, 2009.

\bibitem{EmekR14}
Y.~Emek and A.~Ros{\'{e}}n.
\newblock Semi-streaming set cover - (extended abstract).
\newblock In {\em Automata, Languages, and Programming - 41st International
  Colloquium, {ICALP} 2014, Copenhagen, Denmark, July 8-11, 2014, Proceedings,
  Part {I}}, pages 453--464, 2014.

\bibitem{EneN18}
A.~Ene and H.~L. Nguyen.
\newblock Submodular maximization with nearly-optimal approximation and
  adaptivity in nearly-linear time.
\newblock {\em CoRR}, abs/1804.05379. To appear in SODA 2019., 2018.

\bibitem{EneNV18}
A.~Ene, H.~L. Nguyen, and A.~Vladu.
\newblock Submodular maximization with packing constraints in parallel.
\newblock {\em CoRR}, abs/1808.09987, 2018.

\bibitem{FahrbachMZ18a}
M.~Fahrbach, V.~S. Mirrokni, and M.~Zadimoghaddam.
\newblock Non-monotone submodular maximization with nearly optimal adaptivity
  complexity.
\newblock {\em CoRR}, abs/1808.06932, 2018.

\bibitem{FahrbachMZ18}
M.~Fahrbach, V.~S. Mirrokni, and M.~Zadimoghaddam.
\newblock Submodular maximization with optimal approximation, adaptivity and
  query complexity.
\newblock {\em CoRR}, abs/1807.07889. To appear in SODA 2019., 2018.

\bibitem{FeigenbaumKMSZ05}
J.~Feigenbaum, S.~Kannan, A.~McGregor, S.~Suri, and J.~Zhang.
\newblock On graph problems in a semi-streaming model.
\newblock {\em Theor. Comput. Sci.}, 348(2-3):207--216, 2005.

\bibitem{FeigenbaumKMSZ08}
J.~Feigenbaum, S.~Kannan, A.~McGregor, S.~Suri, and J.~Zhang.
\newblock Graph distances in the data-stream model.
\newblock {\em {SIAM} J. Comput.}, 38(5):1709--1727, 2008.

\bibitem{GavinskyKKRW07}
D.~Gavinsky, J.~Kempe, I.~Kerenidis, R.~Raz, and R.~de~Wolf.
\newblock Exponential separations for one-way quantum communication complexity,
  with applications to cryptography.
\newblock {\em STOC}, pages 516--525, 2007.

\bibitem{GhaffariGKMR18}
M.~Ghaffari, T.~Gouleakis, C.~Konrad, S.~Mitrovic, and R.~Rubinfeld.
\newblock Improved massively parallel computation algorithms for mis, matching,
  and vertex cover.
\newblock In {\em Proceedings of the 2018 {ACM} Symposium on Principles of
  Distributed Computing, {PODC} 2018, Egham, United Kingdom, July 23-27, 2018},
  pages 129--138, 2018.

\bibitem{GoelKK12}
A.~Goel, M.~Kapralov, and S.~Khanna.
\newblock On the communication and streaming complexity of maximum bipartite
  matching.
\newblock In {\em Proceedings of the Twenty-third Annual ACM-SIAM Symposium on
  Discrete Algorithms}, SODA '12, pages 468--485. SIAM, 2012.

\bibitem{GrotschelLS81}
M.~Gr{\"{o}}tschel, L.~Lov{\'{a}}sz, and A.~Schrijver.
\newblock The ellipsoid method and its consequences in combinatorial
  optimization.
\newblock {\em Combinatorica}, 1(2):169--197, 1981.

\bibitem{GuhaM07}
S.~Guha and A.~McGregor.
\newblock Lower bounds for quantile estimation in random-order and multi-pass
  streaming.
\newblock In {\em Automata, Languages and Programming, 34th International
  Colloquium, {ICALP} 2007, Wroclaw, Poland, July 9-13, 2007, Proceedings},
  pages 704--715, 2007.

\bibitem{GuhaM08}
S.~Guha and A.~McGregor.
\newblock Tight lower bounds for multi-pass stream computation via pass
  elimination.
\newblock In {\em Automata, Languages and Programming, 35th International
  Colloquium, {ICALP} 2008, July 7-11, 2008, Proceedings, Part {I:} Tack {A:}
  Algorithms, Automata, Complexity, and Games}, pages 760--772, 2008.

\bibitem{GuruswamiO13}
V.~Guruswami and K.~Onak.
\newblock Superlinear lower bounds for multipass graph processing.
\newblock In {\em Proceedings of the 28th Conference on Computational
  Complexity, {CCC} 2013, K.lo Alto, California, USA, 5-7 June, 2013}, pages
  287--298, 2013.

\bibitem{HalldorssonHLS10}
B.~V. Halld{\'{o}}rsson, M.~M. Halld{\'{o}}rsson, E.~Losievskaja, and
  M.~Szegedy.
\newblock Streaming algorithms for independent sets.
\newblock In {\em Automata, Languages and Programming, 37th International
  Colloquium, {ICALP} 2010, Bordeaux, France, July 6-10, 2010, Proceedings,
  Part {I}}, pages 641--652, 2010.

\bibitem{HalldorssonHLS16}
B.~V. Halld{\'{o}}rsson, M.~M. Halld{\'{o}}rsson, E.~Losievskaja, and
  M.~Szegedy.
\newblock Streaming algorithms for independent sets in sparse hypergraphs.
\newblock {\em Algorithmica}, 76(2):490--501, 2016.

\bibitem{HalldorssonSSW12}
M.~M. Halld{\'{o}}rsson, X.~Sun, M.~Szegedy, and C.~Wang.
\newblock Streaming and communication complexity of clique approximation.
\newblock In {\em Automata, Languages, and Programming - 39th International
  Colloquium, {ICALP} 2012, Warwick, UK, July 9-13, 2012, Proceedings, Part
  {I}}, pages 449--460, 2012.

\bibitem{Har-PeledIMV16}
S.~Har{-}Peled, P.~Indyk, S.~Mahabadi, and A.~Vakilian.
\newblock Towards tight bounds for the streaming set cover problem.
\newblock In {\em Proceedings of the 35th {ACM} {SIGMOD-SIGACT-SIGAI} Symposium
  on Principles of Database Systems, {PODS} 2016, San Francisco, CA, USA, June
  26 - July 01, 2016}, pages 371--383, 2016.

\bibitem{hardy1988inequalities}
G.~H. Hardy, J.~E. Littlewood, and G.~P{\'o}lya.
\newblock {\em Inequalities (Cambridge Mathematical Library)}.
\newblock Cambridge University Press, 1988.

\bibitem{Harvey08t}
N.~J.~A. Harvey.
\newblock {\em Matchings, matroids and submodular functions}.
\newblock PhD thesis, Massachusetts Institute of Technology, 2008.

\bibitem{Harvey08}
N.~J.~A. Harvey.
\newblock Matroid intersection, pointer chasing, and young's seminormal
  representation of \emph{S\({}_{\mbox{n}}\)}.
\newblock In {\em Proceedings of the Nineteenth Annual {ACM-SIAM} Symposium on
  Discrete Algorithms, {SODA} 2008, San Francisco, California, USA, January
  20-22, 2008}, pages 542--549, 2008.

\bibitem{HenzingerKN16}
M.~Henzinger, S.~Krinninger, and D.~Nanongkai.
\newblock A deterministic almost-tight distributed algorithm for approximating
  single-source shortest paths.
\newblock In {\em Proceedings of the 48th Annual {ACM} {SIGACT} Symposium on
  Theory of Computing, {STOC} 2016, Cambridge, MA, USA, June 18-21, 2016},
  pages 489--498, 2016.

\bibitem{IvanyosKLSW12}
G.~Ivanyos, H.~Klauck, T.~Lee, M.~Santha, and R.~de~Wolf.
\newblock New bounds on the classical and quantum communication complexity of
  some graph properties.
\newblock In {\em {IARCS} Annual Conference on Foundations of Software
  Technology and Theoretical Computer Science, {FSTTCS} 2012, December 15-17,
  2012, Hyderabad, India}, pages 148--159, 2012.

\bibitem{IwataFF00}
S.~Iwata, L.~Fleischer, and S.~Fujishige.
\newblock A combinatorial, strongly polynomial-time algorithm for minimizing
  submodular functions.
\newblock In {\em Proceedings of the Thirty-Second Annual {ACM} Symposium on
  Theory of Computing, May 21-23, 2000, Portland, OR, {USA}}, pages 97--106,
  2000.

\bibitem{IwataO09}
S.~Iwata and J.~B. Orlin.
\newblock A simple combinatorial algorithm for submodular function
  minimization.
\newblock In {\em Proceedings of the Twentieth Annual {ACM-SIAM} Symposium on
  Discrete Algorithms, {SODA} 2009, New York, NY, USA, January 4-6, 2009},
  pages 1230--1237, 2009.

\bibitem{JainRS03}
R.~Jain, J.~Radhakrishnan, and P.~Sen.
\newblock A direct sum theorem in communication complexity via message
  compression.
\newblock In {\em Automata, Languages and Programming, 30th International
  Colloquium, {ICALP} 2003, June 30 - July 4, 2003. Proceedings}, pages
  300--315, 2003.

\bibitem{JayramKS03}
T.~S. Jayram, R.~Kumar, and D.~Sivakumar.
\newblock Two applications of information complexity.
\newblock In {\em Proceedings of the 35th Annual {ACM} Symposium on Theory of
  Computing, June 9-11, 2003, San Diego, CA, {USA}}, pages 673--682, 2003.

\bibitem{JowhariG05}
H.~Jowhari and M.~Ghodsi.
\newblock New streaming algorithms for counting triangles in graphs.
\newblock In {\em Computing and Combinatorics, 11th Annual International
  Conference, {COCOON} 2005, Kunming, China, August 16-29, 2005, Proceedings},
  pages 710--716, 2005.

\bibitem{KaleT17}
S.~Kale and S.~Tirodkar.
\newblock Maximum matching in two, three, and a few more passes over graph
  streams.
\newblock In {\em Approximation, Randomization, and Combinatorial Optimization.
  Algorithms and Techniques, {APPROX/RANDOM} 2017, August 16-18, 2017,
  Berkeley, CA, {USA}}, pages 15:1--15:21, 2017.

\bibitem{KalyanasundaramS92}
B.~Kalyanasundaram and G.~Schnitger.
\newblock The probabilistic communication complexity of set intersection.
\newblock {\em {SIAM} J. Discrete Math.}, 5(4):545--557, 1992.

\bibitem{Kapralov13}
M.~Kapralov.
\newblock Better bounds for matchings in the streaming model.
\newblock In {\em Proceedings of the Twenty-Fourth Annual {ACM-SIAM} Symposium
  on Discrete Algorithms, {SODA} 2013, New Orleans, Louisiana, USA, January
  6-8, 2013}, pages 1679--1697, 2013.

\bibitem{KapralovW14}
M.~Kapralov and D.~P. Woodruff.
\newblock Spanners and sparsifiers in dynamic streams.
\newblock In {\em {ACM} Symposium on Principles of Distributed Computing,
  {PODC} '14, Paris, France, July 15-18, 2014}, pages 272--281, 2014.

\bibitem{KonradMM12}
C.~Konrad, F.~Magniez, and C.~Mathieu.
\newblock Maximum matching in semi-streaming with few passes.
\newblock In {\em Approximation, Randomization, and Combinatorial Optimization.
  Algorithms and Techniques - 15th International Workshop, {APPROX} 2012, and
  16th International Workshop, {RANDOM} 2012, Cambridge, MA, USA, August 15-17,
  2012. Proceedings}, pages 231--242, 2012.

\bibitem{KremerNR95}
I.~Kremer, N.~Nisan, and D.~Ron.
\newblock On randomized one-round communication complexity.
\newblock In {\em Proceedings of the Twenty-Seventh Annual {ACM} Symposium on
  Theory of Computing, 29 May-1 June 1995, Las Vegas, Nevada, {USA}}, pages
  596--605, 1995.

\bibitem{KumarMVV13}
R.~Kumar, B.~Moseley, S.~Vassilvitskii, and A.~Vattani.
\newblock Fast greedy algorithms in mapreduce and streaming.
\newblock In {\em 25th {ACM} Symposium on Parallelism in Algorithms and
  Architectures, {SPAA} '13, Montreal, QC, Canada - July 23 - 25, 2013}, pages
  1--10, 2013.

\bibitem{KushilevitzN97}
E.~Kushilevitz and N.~Nisan.
\newblock {\em Communication complexity}.
\newblock Cambridge University Press, 1997.

\bibitem{KutzkovP14a}
K.~Kutzkov and R.~Pagh.
\newblock Triangle counting in dynamic graph streams.
\newblock In {\em Algorithm Theory - {SWAT} 2014 - 14th Scandinavian Symposium
  and Workshops, Copenhagen, Denmark, July 2-4, 2014. Proceedings}, pages
  306--318, 2014.

\bibitem{LeeSW15}
Y.~T. Lee, A.~Sidford, and S.~C. Wong.
\newblock A faster cutting plane method and its implications for combinatorial
  and convex optimization.
\newblock In {\em {IEEE} 56th Annual Symposium on Foundations of Computer
  Science, {FOCS} 2015, Berkeley, CA, USA, 17-20 October, 2015}, pages
  1049--1065, 2015.

\bibitem{Lin09}
H.~Lin.
\newblock Reducing directed max flow to undirected max flow.
\newblock {\em Unpublished manuscript}, 2009.

\bibitem{Lin91}
J.~Lin.
\newblock Divergence measures based on the shannon entropy.
\newblock {\em {IEEE} Trans. Information Theory}, 37(1):145--151, 1991.

\bibitem{DistanceOP}
List of open problems in sublinear algorithms: Problem 14.
\newblock \url{https://sublinear.info/14}.

\bibitem{RandomWalkOP}
List of open problems in sublinear algorithms: Problem 22.
\newblock \url{https://sublinear.info/22}.

\bibitem{StrongOP}
List of open problems in sublinear algorithms: Problem 29.
\newblock \url{https://sublinear.info/29}.

\bibitem{Luby86}
M.~Luby.
\newblock A simple parallel algorithm for the maximal independent set problem.
\newblock {\em {SIAM} J. Comput.}, 15(4):1036--1053, 1986.

\bibitem{McGregor05}
A.~McGregor.
\newblock Finding graph matchings in data streams.
\newblock In {\em Approximation, Randomization and Combinatorial Optimization,
  Algorithms and Techniques, 8th International Workshop on Approximation
  Algorithms for Combinatorial Optimization Problems, {APPROX} 2005 and 9th
  InternationalWorkshop on Randomization and Computation, {RANDOM} 2005,
  Berkeley, CA, USA, August 22-24, 2005, Proceedings}, pages 170--181, 2005.

\bibitem{McGregor14}
A.~McGregor.
\newblock Graph stream algorithms: a survey.
\newblock {\em {SIGMOD} Record}, 43(1):9--20, 2014.

\bibitem{McGregorVV16}
A.~McGregor, S.~Vorotnikova, and H.~T. Vu.
\newblock Better algorithms for counting triangles in data streams.
\newblock In {\em Proceedings of the 35th {ACM} {SIGMOD-SIGACT-SIGAI} Symposium
  on Principles of Database Systems, {PODS} 2016, San Francisco, CA, USA, June
  26 - July 01, 2016}, pages 401--411, 2016.

\bibitem{McGregorV17}
A.~McGregor and H.~T. Vu.
\newblock Better streaming algorithms for the maximum coverage problem.
\newblock In {\em 20th International Conference on Database Theory, {ICDT}
  2017, March 21-24, 2017, Venice, Italy}, pages 22:1--22:18, 2017.

\bibitem{MiltersenNSW95}
P.~B. Miltersen, N.~Nisan, S.~Safra, and A.~Wigderson.
\newblock On data structures and asymmetric communication complexity.
\newblock In {\em Proceedings of the Twenty-Seventh Annual {ACM} Symposium on
  Theory of Computing, 29 May-1 June 1995, Las Vegas, Nevada, {USA}}, pages
  103--111, 1995.

\bibitem{MunroP78}
J.~I. Munro and M.~Paterson.
\newblock Selection and sorting with limited storage.
\newblock In {\em 19th Annual Symposium on Foundations of Computer Science, Ann
  Arbor, Michigan, USA, 16-18 October 1978}, pages 253--258, 1978.

\bibitem{Nemirovski94}
A.~Nemirovski.
\newblock On parallel complexity of nonsmooth convex optimization.
\newblock {\em J. Complexity}, 10(4):451--463, 1994.

\bibitem{NisanW91}
N.~Nisan and A.~Wigderson.
\newblock Rounds in communication complexity revisited.
\newblock In {\em Proceedings of the 23rd Annual {ACM} Symposium on Theory of
  Computing, May 5-8, 1991, New Orleans, Louisiana, {USA}}, pages 419--429,
  1991.

\bibitem{PapadimitriouS84}
C.~H. Papadimitriou and M.~Sipser.
\newblock Communication complexity.
\newblock {\em J. Comput. Syst. Sci.}, 28(2):260--269, 1984.

\bibitem{PonzioRV99}
S.~Ponzio, J.~Radhakrishnan, and S.~Venkatesh.
\newblock The communication complexity of pointer chasing: Applications of
  entropy and sampling.
\newblock In {\em Proceedings of the Thirty-First Annual {ACM} Symposium on
  Theory of Computing, May 1-4, 1999, Atlanta, Georgia, {USA}}, pages 602--611,
  1999.

\bibitem{Razborov92}
A.~A. Razborov.
\newblock On the distributional complexity of disjointness.
\newblock {\em Theor. Comput. Sci.}, 106(2):385--390, 1992.

\bibitem{RubinsteinSW18}
A.~Rubinstein, T.~Schramm, and S.~M. Weinberg.
\newblock Computing exact minimum cuts without knowing the graph.
\newblock In {\em 9th Innovations in Theoretical Computer Science Conference,
  {ITCS} 2018, January 11-14, 2018, Cambridge, MA, {USA}}, pages 39:1--39:16,
  2018.

\bibitem{SahaG09}
B.~Saha and L.~Getoor.
\newblock On maximum coverage in the streaming model {\&} application to
  multi-topic blog-watch.
\newblock In {\em Proceedings of the {SIAM} International Conference on Data
  Mining, {SDM} 2009, Sparks, Nevada, {USA}}, pages 697--708, 2009.

\bibitem{SarmaGP11}
A.~D. Sarma, S.~Gollapudi, and R.~Panigrahy.
\newblock Estimating pagerank on graph streams.
\newblock {\em J. {ACM}}, 58(3):13:1--13:19, 2011.

\bibitem{Schrijver00}
A.~Schrijver.
\newblock A combinatorial algorithm minimizing submodular functions in strongly
  polynomial time.
\newblock {\em J. Comb. Theory, Ser. {B}}, 80(2):346--355, 2000.

\bibitem{SidfordT18}
A.~Sidford and K.~Tian.
\newblock Coordinate methods for accelerating $\ell_{\infty}$ regression and
  faster approximate maximum flow.
\newblock {\em CoRR}, abs/1808.01278, 2018.

\bibitem{VerbinY11}
E.~Verbin and W.~Yu.
\newblock The streaming complexity of cycle counting, sorting by reversals, and
  other problems.
\newblock In {\em Proceedings of the Twenty-Second Annual {ACM-SIAM} Symposium
  on Discrete Algorithms, {SODA} 2011, January 23-25, 2011}, pages 11--25,
  2011.

\bibitem{WeinsteinW15}
O.~Weinstein and D.~P. Woodruff.
\newblock The simultaneous communication of disjointness with applications to
  data streams.
\newblock In {\em Automata, Languages, and Programming - 42nd International
  Colloquium, {ICALP} 2015, July 6-10, 2015, Proceedings, Part {I}}, pages
  1082--1093, 2015.

\bibitem{Yao79}
A.~C. Yao.
\newblock Some complexity questions related to distributive computing
  (preliminary report).
\newblock In {\em Proceedings of the 11h Annual {ACM} Symposium on Theory of
  Computing, April 30 - May 2, 1979, Atlanta, Georgia, {USA}}, pages 209--213,
  1979.

\bibitem{Yao83}
A.~C. Yao.
\newblock Lower bounds by probabilistic arguments (extended abstract).
\newblock In {\em 24th Annual Symposium on Foundations of Computer Science,
  Tucson, Arizona, USA, 7-9 November 1983}, pages 420--428, 1983.

\bibitem{Yehudayoff16}
A.~Yehudayoff.
\newblock Pointer chasing via triangular discrimination.
\newblock {\em Electronic Colloquium on Computational Complexity {(ECCC)}},
  23:151, 2016.

\bibitem{Zelke11}
M.~Zelke.
\newblock Intractability of min- and max-cut in streaming graphs.
\newblock {\em Inf. Process. Lett.}, 111(3):145--150, 2011.

\end{thebibliography}


\clearpage

\appendix


\section{Further Related Work}\label{app:related}

Understanding space/pass tradeoffs for streaming algorithms dates all the way back to the early results on median-finding~\cite{MunroP78} more than four decades ago and has remained a focus
of attention since; we refer the interested reader to~\cite{GuhaM07,GuhaM08,ChakrabartiCM08,ChakrabartiCKM10} and references therein.

A closely related line of work to graph streaming algorithms that have received a significant attention in recent years is on 
streaming algorithms for submodular optimization and in particular set cover and
maximum coverage~\cite{SahaG09,DemaineIMV14,EmekR14,AssadiKL16,Assadi17sc,Har-PeledIMV16,ChakrabartiW16,BateniEM17,KumarMVV13,ChakrabartiK14,McGregorV17,AssadiK18}. Particularly relevant to our work,~\cite{ChakrabartiW16} uses a reduction from the multi-party tree pointer chasing problem~\cite{ChakrabartiCM08} to prove
an $\Omega(\frac{\log{n}}{\log\log{n}})$ pass lower bound for approximating set cover with $m$ sets and $n$ elements using $O(n \cdot \poly\set{\log{n},\log{m}})$ space (this can also be interpreted as a lower bound
for the edge-cover problem on hyper-graphs with $n$ vertices and $m$ hyper-edges in the graph streaming model). For the set cover problem, a lower bound of $\Omega(\frac{m \cdot n^{1/\alpha}}{p})$ space for $p$-pass streaming $\alpha$-approximation algorithms is established in~\cite{Assadi17sc} using a reduction from the set disjointness problem (this can also be interpreted as a lower bound for the dominating set problem on graphs with $n=m$ vertices in the graph streaming model). 

Similar-in-spirit round/communication tradeoffs for distributed computation of many graph and related problems have also been studied in the literature~\cite{DobzinskiNO14,AlonNRW15,Assadi17ca,BravermanO17,BravermanMW18,AssadiK18}. For 
example,~\cite{BravermanO17} proves an $\Omega(\frac{\log{n}}{\log\log{n}})$ round lower bound for protocols with low communication that can approximate matchings in a communication model in which players correspond to vertices of an $n$-vertex graph. 
Similarly,~\cite{AssadiK18} proves an $\Omega(\frac{\log{n}}{\log\log{n}})$ round lower bound for constrained submodular maximization in a communication model where $n$ elements of a universe are partitioned between the players. 

Adaptivity lower bounds for submodular 
optimization~\cite{BalkanskiRS16,BalkanskiRS17,BalkanskiS17,BalkanskiS18,BalkansiRS18,BalkanskiBS18,FahrbachMZ18,FahrbachMZ18a,EneN18,EneNV18} is another topic related to our work. 
For example,~\cite{BalkanskiS18} proves that $\Omega(\frac{\log{n}}{\log\log{n}})$ rounds of adaptivity are necessary for constrained submodular maximization with polynomial query complexity. 
Additionally,~\cite{BalkanskiS17} proved that no non-adaptive algorithm can obtain a better than $1/2$ approximation to submodular minimization with polynomially many queries. 
Finally, if one goes (way) beyond submodular optimization and considers minimizing a non-smooth convex function, then an $\Omgt(n^{1/3})$ lower bound on rounds of adaptivity is known
for any algorithm that makes polynomially many queries~\cite{Nemirovski94,BalkanskiS18m}. 

The appearance of the same logarithmic term in these lower bounds is not merely a coincidence. The core idea behind all these results (with the exception of~\cite{BravermanO17}) is a round-elimination type 
argument (see, e.g.~\cite{MiltersenNSW95}) that is a reminiscent of the lower bounds for the tree pointer chasing problem~\cite{ChakrabartiCM08} (see~\cite{AlonNRW15,AssadiK18} and~\cite{BalkanskiS17} for the details on, respectively, the communication lower bounds and 
the adaptivity lower bounds). As such, these results also inherit the shortcoming of the tree pointer chasing problem in having an exponential dependence on number of rounds, leading to at most logarithmic bound in the round/adaptivity lower bound. 

However, we shall also emphasize that most lower bounds mentioned above hold even for ``simpler'' variants of the problem, say by allowing approximation and/or considering simpler constraints such as cardinality constraint for submodular maximization. 
For these simpler variants, these bounds are essentially tight as there do exist approximation algorithms with round/adaptivity complexity that almost match these bounds. 
Nevertheless, once we consider ``harder'' variants of these problems, say, by switching to the exact solution in case of maximum matching or more general constraints such as $p$-systems in submodular maximization, 
no such efficient algorithms are known. At the same time, no better lower bounds are also known for these harder variants (see, e.g.~\cite{DobzinskiNO14} that posed the question of round/communication tradeoffs for finding perfect matchings in the communication model). We hope that our approach in this paper can also pave the path for obtaining stronger lower bounds in these settings.

\section{Background and Preliminaries}\label{app:prelim}

We use the following basic inequality in our proofs. 

 \begin{proposition}\label{prop:rearr}
        For any two lists of numbers $a_1 \le a_2 \le \dots \le a_n$ and $b_1 \ge b_2 \ge \dots \ge b_n$, $\sum_{i=1}^{n} a_ib_i \le \frac{1}{n}\sum_{i=1}^n a_i \cdot \sum_{i=1}^n b_i$.
    \end{proposition}
    \begin{proof}
    	The rearrangement inequality~\cite{hardy1988inequalities} states that for any list of numbers $x_1 \leq \dots \leq x_n$ and $y_1\leq \dots\leq y_n$ and any permutation $\sigma$ of $[n]$, 
	\begin{align*}
		x_1 \cdot y_n + \dots + x_n \cdot y_1 \leq x_1 \cdot y_{\sigma(1)} + \dots + x_n \cdot y_{\sigma(n)} \leq x_1 \cdot y_1 + \dots + x_n \cdot y_n.
	\end{align*}
        By rearrangement inequality, for any $0 \le j < n$,
        \[\sum_{i=1}^{n} a_i b_i \le \sum_{i=1}^n a_i b_{i+j}, \]
        where, with a slight abuse of notation, we use $b_{i+j}$ for $i+j > n$ to denote $b_{i+j - n}$. As such, 
        \[
        \sum_{i=1}^n a_i b_i \le \frac{1}{n} \sum_{j=0}^{n-1} \sum_{i=1}^n a_i b_{i+j} = \frac{1}{n} \sum_{i=1}^n (a_i \sum_{j=0}^{n-1} b_{i+j}) = \frac{1}{n} \sum_{i=1}^n a_i \cdot \sum_{i=1}^n b_i
        \qed 
        \]
         
    \end{proof}

\subsection{Background on Information Theory}\label{sec:info}

We briefly introduce some definitions and facts from information theory that are needed. We refer the interested reader to~\cite{CoverT06} for an excellent introduction to this field. 

For a random variable $\rA$, we use $\supp{\rA}$ to denote the support of $\rA$ and $\distribution{\rA}$ to denote its distribution. 
When it is clear from the context, we may abuse the notation and use $\rA$ directly instead of $\distribution{\rA}$, for example, write 
$A \sim \rA$ to mean $A \sim \distribution{\rA}$, i.e., $A$ is sampled from the distribution of random variable $\rA$. 
We denote the \emph{Shannon Entropy} of a random variable $\rA$ by
$\en{\rA}$, which is defined as: 
\begin{align}
	\en{\rA} := \sum_{A \in \supp{\rA}} \Pr\paren{\rA = A} \cdot \log{\paren{1/\Pr\paren{\rA = A}}} \label{eq:entropy}
\end{align} 
\noindent
The \emph{conditional entropy} of $\rA$ conditioned on $\rB$ is denoted by $\en{\rA \mid \rB}$ and defined as:
\begin{align}
\en{\rA \mid \rB} := \Ex_{B \sim \rB} \bracket{\en{\rA \mid \rB = B}}, \label{eq:cond-entropy}
\end{align}
where 
$\en{\rA \mid \rB = B}$ is defined in a standard way by using the distribution of $\rA$ conditioned on the event $\rB = B$ in Eq~(\ref{eq:entropy}). 
The \emph{mutual information} of two random variables $\rA$ and $\rB$ is denoted by
$\mi{\rA}{\rB}$ and is defined as:
\begin{align}
\mi{\rA}{\rB} := \en{\rA} - \en{\rA \mid  \rB} = \en{\rB} - \en{\rB \mid  \rA} = \mi{\rB}{\rA}. \label{eq:mi}
\end{align}
\noindent
The \emph{conditional mutual information} $\mi{\rA}{\rB \mid \rC}$ is $\en{\rA \mid \rC} - \en{\rA \mid \rB,\rC}$ and hence by linearity of expectation:
\begin{align}
	\mi{\rA}{\rB \mid \rC} = \Ex_{C \sim \rC} \bracket{\mi{\rA}{\rB \mid \rC = C}}. \label{eq:cond-mi}
\end{align} 

When it may lead to confusion, we use the subscript $\dist$ in $\HH_{\dist}$ and $\II_{\dist}$ to mean that the random variables in these
terms are distributed according to the distribution $\dist$. 


\subsubsection{Useful Properties of Entropy and Mutual Information}\label{sec:prop-en-mi}

We shall use the following basic properties of entropy and mutual information throughout. 

\begin{fact}[cf.~\cite{CoverT06}; Chapter~2]\label{fact:it-facts}
  Let $\rA$, $\rB$, $\rC$, and $\rD$ be four (possibly correlated) random variables.
   \begin{enumerate}
  \item \label{part:uniform} $0 \leq \en{\rA} \leq \log{\card{\supp{\rA}}}$. The right equality holds
    iff $\distribution{\rA}$ is uniform.
  \item \label{part:info-zero} $\mi{\rA}{\rB} \geq 0$. The equality holds iff $\rA$ and
    $\rB$ are \emph{independent}.
  \item \label{part:cond-reduce} \emph{Conditioning on a random variable can only reduce the entropy}:
    $\en{\rA \mid \rB,\rC} \leq \en{\rA \mid  \rB}$.  The equality holds iff $\rA \perp \rC \mid \rB$.
    \item \label{part:sub-additivity} \emph{Subadditivity of entropy}: $\en{\rA,\rB \mid \rC}
    \leq \en{\rA \mid C} + \en{\rB \mid  \rC}$.
   \item \label{part:ent-chain-rule} \emph{Chain rule for entropy}: $\en{\rA,\rB \mid \rC} = \en{\rA \mid \rC} + \en{\rB \mid \rC,\rA}$.
  \item \label{part:chain-rule} \emph{Chain rule for mutual information}: $\mi{\rA,\rB}{\rC \mid \rD} = \mi{\rA}{\rC \mid \rD} + \mi{\rB}{\rC \mid  \rA,\rD}$.
   \end{enumerate}
\end{fact}


\noindent
We also use the following two standard propositions regarding the effect of conditioning on mutual information.

\begin{proposition}\label{prop:info-increase}
  For random variables $\rA, \rB, \rC, \rD$, if $\rA \perp \rD \mid \rC$, then, 
  \[\mi{\rA}{\rB \mid \rC} \leq \mi{\rA}{\rB \mid  \rC,  \rD}.\]
\end{proposition}
 \begin{proof}
  Since $\rA$ and $\rD$ are independent conditioned on $\rC$, by
  \itfacts{cond-reduce}, $\HH(\rA \mid  \rC) = \HH(\rA \mid \rC, \rD)$ and $\HH(\rA \mid  \rC, \rB) \ge \HH(\rA \mid  \rC, \rB, \rD)$.  We have,
	 \begin{align*}
	  \mi{\rA}{\rB \mid  \rC} &= \HH(\rA \mid \rC) - \HH(\rA \mid \rC, \rB) = \HH(\rA \mid  \rC, \rD) - \HH(\rA \mid \rC, \rB) \\
	  &\leq \HH(\rA \mid \rC, \rD) - \HH(\rA \mid \rC, \rB, \rD) = \mi{\rA}{\rB \mid \rC, \rD}. \qed
	\end{align*}
	
\end{proof}

\begin{proposition}\label{prop:info-decrease}
  For random variables $\rA, \rB, \rC,\rD$, if $ \rA \perp \rD \mid \rB,\rC$, then, 
  \[\mi{\rA}{\rB \mid \rC} \geq \mi{\rA}{\rB \mid \rC, \rD}.\]
\end{proposition}
 \begin{proof}
 Since $\rA \perp \rD \mid \rB,\rC$, by \itfacts{cond-reduce}, $\HH(\rA \mid \rB,\rC) = \HH(\rA \mid \rB,\rC,\rD)$. Moreover, since conditioning can only reduce the entropy (again by \itfacts{cond-reduce}), 
  \begin{align*}
 	\mi{\rA}{\rB \mid  \rC} &= \HH(\rA \mid \rC) - \HH(\rA \mid \rB,\rC) \geq \HH(\rA \mid \rD,\rC) - \HH(\rA \mid \rB,\rC) \\
	&= \HH(\rA \mid \rD,\rC) - \HH(\rA \mid \rB,\rC,\rD) = \mi{\rA}{\rB \mid \rC,\rD}. \qed
 \end{align*}
 
\end{proof}
\noindent
Finally, we also use the following simple inequality that states that conditioning on a random variable can only increase the mutual information
by the entropy of the conditioned variable. 

\begin{proposition}\label{prop:bar-hopping}
  For random variables $\rA, \rB$ and $\rC$, 
$
  \mi{\rA}{\rB \mid \rC} \leq \mi{\rA}{\rB} + \en{\rC}.
  $ 
\end{proposition}
\begin{proof}
	By chain rule for mutual information (\itfacts{chain-rule}), we can write:
	\begin{align*}
		\mi{\rA}{\rB \mid \rC} &= \mi{\rA}{\rB,\rC} - \mi{\rA}{\rC} = \mi{\rA}{\rB} + \mi{\rA}{\rC \mid \rB} - \mi{\rA}{\rC} \\
		&\leq \mi{\rA}{\rB} + \en{\rC \mid \rB} \leq \mi{\rA}{\rB} + \en{\rC},
	\end{align*}
	where the first two equalities are by chain rule (\itfacts{chain-rule}), the second inequality is by definition of mutual information and its positivity (\itfacts{info-zero}), and the last one is because conditioning 
	can only reduce the entropy (\itfacts{cond-reduce}).
\end{proof}

\subsubsection{Measures of Distance Between Distributions}\label{sec:prob-distance}

We shall make use of several measures of distance (or divergence) between distributions in our proofs. We define these measures here and present their main properties that we use in this paper. 

\paragraph{KL-divergence.} For two distributions $\mu$ and $\nu$, the \emph{Kullback-Leibler divergence} between $\mu$ and $\nu$ is denoted by $\kl{\mu}{\nu}$ and defined as: 
\begin{align}
\kl{\mu}{\nu}:= \Ex_{a \sim \mu}\Bracket{\log\frac{\Pr_\mu(a)}{\Pr_{\nu}(a)}}. \label{eq:kl}
\end{align}
We have the following relation between mutual information and KL-divergence. 
\begin{fact}\label{fact:kl-info}
	For random variables $\rA,\rB,\rC$, 
	\[\mi{\rA}{\rB \mid \rC} = \Ex_{(b,c) \sim {(\rB,\rC)}}\Bracket{ \kl{\distribution{\rA \mid \rC=c}}{\distribution{\rA \mid \rB=b,\rC=c}}}.\] 
\end{fact}

\paragraph{Total variation distance.} We denote the total variation distance between two distributions $\mu$ and $\nu$ on the same 
support $\Omega$ by $\tvd{\mu}{\nu}$, defined as: 
\begin{align}
\tvd{\mu}{\nu}:= \max_{\Omega' \subseteq \Omega} \paren{\mu(\Omega')-\nu(\Omega')} = \frac{1}{2} \cdot \sum_{x \in \Omega} \card{\mu(x) - \nu(x)}.  \label{eq:tvd}
\end{align}
\noindent
We use the following basic properties of total variation distance. 
\begin{fact}\label{fact:tvd-small}
	Suppose $\mu$ and $\nu$ are two distributions for $\event$, then, 
	$
	\Pr_{\mu}(\event) \leq \Pr_{\nu}(\event) + \tvd{\mu}{\nu}.
$
\end{fact}

The following Pinskers' inequality bounds the total variation distance between two distributions based on their KL-divergence, 

\begin{fact}[Pinsker's inequality]\label{fact:pinskers}
	For any distributions $\mu$ and $\nu$, 
	$
	\tvd{\mu}{\nu} \leq \sqrt{\frac{1}{2} \cdot \kl{\mu}{\nu}}.
	$ 
\end{fact}

\paragraph{Hellinger distance.} For two distributions $\mu$ and $\nu$, the \emph{Hellinger distance} between $\mu$ and $\nu$ is denoted by $\hd{\mu}{\nu}$ and is defined as:
\begin{align}
	\hd{\mu}{\nu} := \sqrt{\frac{1}{2}\sum_{x \in \Omega}(\sqrt{\mu(x)}-\sqrt{\nu(x)})^2} = \sqrt{1-\sum_{x \in \Omega}\sqrt{\mu(x)\nu(x)}}. \label{eq:hd}
\end{align}
The following inequalities relate Hellinger distance and total variation distance (the proof follows from Cauchy-Schwartz). 
\begin{fact}\label{fact:hellinger-tvd}
	For any distributions $\mu$ and $\nu$, $\hdt{\mu}{\nu} \leq \tvd{\mu}{\nu} \leq \sqrt{2} \cdot \hd{\mu}{\nu}.$
\end{fact}
One can also relate Hellinger distance to the KL-divergence as follows. 
\begin{fact}[cf.~\cite{Lin91}]\label{fact:hellinger-kl}
	For any distributions $\mu$ and $\nu$, $\hdt{\mu}{\nu} \leq \frac{1}{2} \cdot \paren{\kl{\mu}{\frac{\mu+\nu}{2}} + \kl{\nu}{\frac{\mu+\nu}{2}}}$.
\end{fact}


\subsection{Background on Communication and Information Complexity}\label{sec:cc-ic}

\paragraph{Communication complexity.} We briefly review the standard definitions of the two-party communication model of Yao~\cite{Yao79}.  See the text by Kushilevitz and Nisan~\cite{KushilevitzN97} for 
an extensive overview of communication complexity. In Section~\ref{sec:hpc}, we also use a standard generalization of this model to allow for more than two players, but we defer
the necessary definitions to that section. 

Let $P: \mathcal{X} \times \mathcal{Y} \rightarrow \mathcal{Z}$ be a relation.  Alice receives an input $X
\in \mathcal{X}$ and Bob receives $Y \in \mathcal{Y}$, where $(X,Y)$ are chosen from a
joint distribution $\dist$ over $\mathcal{X} \times \mathcal{Y}$. We allow players to have access to both public and private randomness. 
They communicate with each other by exchanging messages such that each message
depends only on the private input and random bits of the player sending the message, and the already communicated messages plus the public randomness. 
At the end, one of the players need to output an answer  $Z$ such that $Z \in P(X,Y)$.  

We use $\prot$ to denote a protocol used by the players. We always assume that the protocol $\prot$ can be randomized (using both public and
private randomness), \emph{even against a prior distribution $\dist$ of inputs}. For any
$0 < \delta < 1$, we say $\prot$ is a $\delta$-error protocol for $P$ over a distribution
$\dist$, if the probability that for an input $(X,Y)$, $\prot$ outputs some $Z$ where $Z \notin P(X,Y)$ is at most
$\delta$ (the probability is taken over the randomness of \emph{both} the distribution and the protocol).

\begin{definition}[Communication cost]
  The \emph{communication cost} of a protocol $\prot$ on an input
  distribution $\dist$, denoted by $\CC{\prot}{\dist}$, is the \emph{worst-case} bit-length of the transcript
  communicated between Alice and Bob in the protocol $\prot$, when the inputs are chosen from $\dist$.
\end{definition}

Communication complexity of a problem $P$ is defined as the minimum communication cost of a protocol $\prot$ that solves $P$ on every distribution $\dist$ with probability at least $2/3$.

\paragraph{Information complexity.} There are several possible definitions of information
cost of a communication prtocol that have been considered depending on the application (see, e.g.,~\cite{ChakrabartiSWY01,Bar-YossefJKS02,BarakBCR10,BravermanR11,BravermanEOPV13}).  
We use the notion of \emph{internal information cost}~\cite{BarakBCR10} that measures the average amount of information each player learns about the input of the other player
by observing the transcript of the protocol. 
\begin{definition}[Information cost]
  Consider an input distribution $\dist$ and a protocol $\prot$. Let $(\rX,\rY) \sim \dist$ denote the random variables for the input of Alice and Bob and $\rProt$ be the 
  the random variable for the transcript of the protocol \emph{concatenated} with the public randomness $\rR$ used by $\prot$. 
  The \emph{(internal) information cost} of $\prot$ with respect to
  $\dist$ is $\ICost{\prot}{\dist}:=\mii{\rProt}{\rX \mid \rY}{\dist} + \mii{\rProt}{\rY \mid \rX}{\dist}$. 
\end{definition}

One can also define information complexity of a problem $P$ similar to communication complexity with respect to the information cost. However, we avoid presenting this definition formally due to some subtle technical issues
 that need to be addressed which lead to multiple different but similar-in-spirit definitions. As such, we state our results directly in terms of information cost. 

Note that any public coin protocol is a distribution over private coins protocols, run by
first using public randomness to sample a random string $\rR=R$ and then running the
corresponding private coin protocol $\prot^R$. We also use $\rProt^R$ to denote the transcript of the protocol $\prot^R$. 
We have the following standard proposition. 
\begin{proposition}\label{prop:public-random}
	 For any distribution $\dist$ and any protocol $\prot$ with public randomness $\bR$, $$\ICost{\prot}{\dist} = \mii{\rProt}{\rX \mid \rY, \rR}{\dist} + \mii{\rProt}{\rY \mid \rX, \rR}{\dist} = \Ex_{R \sim \rR}\bracket{\ICost{\prot^R}{\dist}}.$$ 
\end{proposition}
\begin{proof} By definition of internal information cost, 
\begin{align*}
	\ICost{\prot}{\dist} &= \mii{\rProt}{\rX \mid \rY}{\dist} + \mii{\rProt }{\rY \mid \rX}{\dist} = \mi{\rProt,\rR}{\rX \mid \rY} + \mi{\rProt,\rR}{\rY \mid \rX} \tag{$\Prot$ denotes the transcript and the public randomness} \\
	&= \mi{\rR}{\rX \mid \rY}  + \mi{\rProt}{\rX \mid \rY , \rR} + \mi{\rR}{\rY \mid \rX}  + \mi{\rProt}{\rY \mid \rX , \rR}  \tag{chain rule of mutual information, \itfacts{chain-rule}} \\
	&= \mi{\rProt}{\rX \mid \rY , \rR} + \mi{\rProt}{\rY \mid \rX , \rR} \tag{$\mi{\rR}{\rX \mid \rY}  = \mi{\rR}{\rY \mid \rX}  = 0$ since $\rR \perp \rX, \rY$ and \itfacts{info-zero}} \\
	&= \Ex_{R \sim \rR}\bracket{\mi{\rProt}{\rX \mid \rY , \rR=R} + \mi{\rProt}{\rY \mid \rX , \rR=R}} = \Ex_{R \sim \rR}\bracket{\ICost{\prot^R}{\dist}}, 
\end{align*}
concluding the proof.
\end{proof}

The following well-known proposition relates communication cost and information cost. 
\begin{proposition}[cf.~\cite{BravermanR11}]\label{prop:cc-ic}
  For any distribution $\dist$ and any protocol $\prot$: $\ICost{\prot}{\dist} \leq \CC{\prot}{\dist}$.
\end{proposition}
\begin{proof}
Let us assume first that $\prot$ only uses private randomness and thus $\rProt$ only contain the transcript.  
For any $b \in [\CC{\prot}{\dist}]$, we define $\Prot_b$ to be the $b$-th bit of the transcript. We have, 
\begin{align*}
	 \ICost{\prot}{\dist} &= \mi{\rProt}{\rX \mid \rY}+ \mi{\rProt }{\rY \mid \rX} \\
	 &= \sum_{b=1}^{\CC{\prot}{\dist}} \mi{\rProt_b}{\rX \mid \rProt^{<b},\rY} + \mi{\rProt_b}{\rY \mid \rProt^{<b},\rX} \tag{by chain rule of mutual information in \itfacts{chain-rule}}\\
	 &= \sum_{b=1}^{\CC{\prot}{\dist}} \Ex_{\Prot^{<b}}\bracket{\mi{\rProt_b}{\rX \mid \rProt^{<b} = \Prot^{<b},\rY} + \mi{\rProt_b}{\rY \mid \rProt^{<b}= \Prot^{<b},\rX}}.
\end{align*}
Consider each term in the RHS above. By conditioning on $\Prot^{<b}$, the player that transmit $\rProt_b$ would become fix. If this player is Alice, then $\mi{\rProt_b}{\rY \mid \rProt^{<b}= \Prot^{<b},\rX} = 0$, because
$\rProt_b$ is only a function of $(\rProt^{<b},\rX)$ in this case; similarly, if this player is Bob, then $\mi{\rProt_b}{\rX \mid \rProt^{<b}= \Prot^{<b},\rY} = 0$. Moreover, 
$\mi{\rProt_b}{\rX \mid \rProt^{<b} = \Prot^{<b},\rY} \leq \en{\rProt_b}  \leq 1$ and similarly $\mi{\rProt_b}{\rY \mid \rProt^{<b}= \Prot^{<b},\rX} \leq 1$. As such, the above term can be upper bounded 
by $\CC{\prot}{\dist}$. To finalize the proof, note that by Proposition~\ref{prop:public-random}, for any public-coin protocol $\prot$, $\ICost{\prot}{\dist} = \Ex_{R \sim \rR}\bracket{\ICost{\prot^R}{\dist}} \leq \Ex_{R \sim \rR}\bracket{\CC{\prot^R}{\dist}} \leq \CC{\prot}{\dist}$, where
the first inequality is by the first part of the argument.  
\end{proof}

Proposition~\ref{prop:cc-ic} provides a convinent way of proving communication complexity lower bounds by lower bounding information cost of any protocol.

\subsection*{Rectangle Property of Communication Protocols}\label{sec:statstics}
We conclude this section by mentioning some basic properties of communication protocols. For any protocol $\prot$ and inputs $x \in \mathcal{X}$ and $y \in \mathcal{Y}$, we define $\Prot_{x,y}$ as the transcript
of the protocol conditioned on the input $x$ to Alice and input $y$ to Bob. Note that for randomized protocols, $\Prot_{x,y}$ is a random variable which we denote by $\rProt_{x,y}$. 

The following is referred to as the rectangle property of deterministic protocols. 
\begin{fact}[Rectangle property]\label{fact:rectangle}
	For any deterministic protocol $\prot$ and inputs $x,x' \in \mathcal{X}$ to Alice and $y,y' \in \mathcal{Y}$ to Bob, if $\Prot_{x,y} = \Prot_{x',y'}$, then $\Prot_{x,y'} = \Prot_{x',y}$. 
\end{fact}
\noindent
Fact~\ref{fact:rectangle} implies that the set of inputs consistent with any transcript $\Prot_{x,y}$ of a deterministic protocol forms a combinatorial rectangle. 
One can also extend the rectangle property of deterministic protocols to randomized protocols using the following fact. 
\begin{fact}[Cut-and-paste property; cf.~\cite{Bar-YossefJKS02}]\label{fact:r-rectangle}
	For any randomized protocol $\prot$ and inputs $x,x' \in \mathcal{X}$ to Alice and $y,y' \in \mathcal{Y}$ to Bob, $\hd{\rProt_{x,y}}{\rProt_{x',y'}} = \hd{\rProt_{x,y'}}{\rProt_{x',y}}$. 
\end{fact}

\section{Communication Phases in \HPC}\label{app:phase}

 An important notion in computing $\HPC$ is a \emph{communication phase} defined as follows:  
Let $\prot$ be any protocol for $\HPC$. We partition the communication steps of $\prot$ into multiple \emph{phases} starting from phase one. 
In an \emph{odd} phase in $\prot$, the players $P_C$ and $P_D$ can communicate back and forth with each other (without restriction on the number of rounds of interaction), but once one of them 
sends a single message (possibly more than one bit) to either $P_A$ or $P_B$ this phase is concluded. In an \emph{even} phase of $\prot$, $P_A$ and $P_B$ are allowed to communicate back and forth and then again once one of them
sends a single message to either $P_C$ or $P_D$ this phase is concluded. One can always uniquely partition the communication steps of any protocol into multiple phases. We refer to a protocol $\prot$ as a \emph{$k$-phase} protocol iff its communication steps consists of $k$ phases. See Figure~\ref{fig:hpc-phase} for an illustration. 

\begin{figure}[h!]
    \centering
    \subcaptionbox{In phase one, $P_C$ and $P_D$ communicate back and forth.}[0.24\textwidth]{

\begin{tikzpicture}[ auto ,node distance =1cm and 2cm , on grid , semithick , state/.style ={ circle ,top color =white , bottom color = white , draw, black , text=black}, every node/.style={inner sep=0,outer sep=0}]

\node[state,rectangle, top color=white, black, bottom color=white, text=black, minimum height=30pt, minimum width = 30pt, rounded corners=2mm,  align=center]  (PA){$P_A$};
\node[state,rectangle, top color=white, black, bottom color=white, text=black, minimum height=30pt, minimum width = 30pt, rounded corners=2mm,  align=center]  (PB)[right=70pt of PA]{$P_B$};
\node[state,rectangle, top color=white, black, bottom color=white, text=black, minimum height=30pt, minimum width = 30pt, rounded corners=2mm,  align=center]  (PC)[below=55pt of PA]{$P_C$};
\node[state,rectangle, top color=white, black, bottom color=white, text=black, minimum height=30pt, minimum width = 30pt, rounded corners=2mm,  align=center]  (PD)[right=70pt of PC]{$P_D$};

\node[inner sep=3pt, draw, dotted, blue, fit=(PA) (PB), line width=0.5mm] {};
\node[inner sep=3pt, draw, dotted, red, fit=(PC) (PD), line width=0.5mm] {};

\draw[<->,line width=1pt] (PC) -- (PD);

\end{tikzpicture}
}   \hspace{2mm}
    \subcaptionbox{Phase one ends when $P_C$ or $P_D$ sends a message to $P_A$ or $P_B$.}[0.22\textwidth]{

    \begin{tikzpicture}[ auto ,node distance =1cm and 2cm , on grid , semithick , state/.style ={ circle ,top color =white , bottom color = white , draw, black , text=black}, every node/.style={inner sep=0,outer sep=0}]

\node[state,rectangle, top color=white, black, bottom color=white, text=black, minimum height=30pt, minimum width = 30pt, rounded corners=2mm,  align=center]  (PA){$P_A$};
\node[state,rectangle, top color=white, black, bottom color=white, text=black, minimum height=30pt, minimum width = 30pt, rounded corners=2mm,  align=center]  (PB)[right=70pt of PA]{$P_B$};
\node[state,rectangle, top color=white, black, bottom color=white, text=black, minimum height=30pt, minimum width = 30pt, rounded corners=2mm,  align=center]  (PC)[below=55pt of PA]{$P_C$};
\node[state,rectangle, top color=white, black, bottom color=white, text=black, minimum height=30pt, minimum width = 30pt, rounded corners=2mm,  align=center]  (PD)[right=70pt of PC]{$P_D$};

\node[inner sep=3pt, draw, dotted, blue, fit=(PA) (PB), line width=0.5mm] {};
\node[inner sep=3pt, draw, dotted, red, fit=(PC) (PD), line width=0.5mm] {};

\draw[->,line width=1pt, dashed] (PC) -- (PB); 
\draw[->,line width=1pt, dashed] (PC) -- (PA); 
\draw[->,line width=1pt, dashed] (PD) -- (PB); 
\draw[->,line width=1pt, dashed] (PD) -- (PA); 

\end{tikzpicture}
}
\hspace{2mm}
    \subcaptionbox{In phase two, $P_A$ and $P_B$ communicate back and forth with each other.}[0.22\textwidth]{

    \begin{tikzpicture}[ auto ,node distance =1cm and 2cm , on grid , semithick , state/.style ={ circle ,top color =white , bottom color = white , draw, black , text=black}, every node/.style={inner sep=0,outer sep=0}]

\node[state,rectangle, top color=white, black, bottom color=white, text=black, minimum height=30pt, minimum width = 30pt, rounded corners=2mm,  align=center]  (PA){$P_A$};
\node[state,rectangle, top color=white, black, bottom color=white, text=black, minimum height=30pt, minimum width = 30pt, rounded corners=2mm,  align=center]  (PB)[right=70pt of PA]{$P_B$};
\node[state,rectangle, top color=white, black, bottom color=white, text=black, minimum height=30pt, minimum width = 30pt, rounded corners=2mm,  align=center]  (PC)[below=55pt of PA]{$P_C$};
\node[state,rectangle, top color=white, black, bottom color=white, text=black, minimum height=30pt, minimum width = 30pt, rounded corners=2mm,  align=center]  (PD)[right=70pt of PC]{$P_D$};

\node[inner sep=3pt, draw, dotted, blue, fit=(PA) (PB), line width=0.5mm] {};
\node[inner sep=3pt, draw, dotted, red, fit=(PC) (PD), line width=0.5mm] {};

\draw[<->,line width=1pt] (PA) -- (PB); 

\end{tikzpicture}
}\hspace{2mm}
    \subcaptionbox{Phase two ends when $P_A$ or $P_B$ sends a message to $P_C$ or $P_D$.}[0.22\textwidth]{

    \begin{tikzpicture}[ auto ,node distance =1cm and 2cm , on grid , semithick , state/.style ={ circle ,top color =white , bottom color = white , draw, black , text=black}, every node/.style={inner sep=0,outer sep=0}]

\node[state,rectangle, top color=white, black, bottom color=white, text=black, minimum height=30pt, minimum width = 30pt, rounded corners=2mm,  align=center]  (PA){$P_A$};
\node[state,rectangle, top color=white, black, bottom color=white, text=black, minimum height=30pt, minimum width = 30pt, rounded corners=2mm,  align=center]  (PB)[right=70pt of PA]{$P_B$};
\node[state,rectangle, top color=white, black, bottom color=white, text=black, minimum height=30pt, minimum width = 30pt, rounded corners=2mm,  align=center]  (PC)[below=55pt of PA]{$P_C$};
\node[state,rectangle, top color=white, black, bottom color=white, text=black, minimum height=30pt, minimum width = 30pt, rounded corners=2mm,  align=center]  (PD)[right=70pt of PC]{$P_D$};

\node[inner sep=3pt, draw, dotted, blue, fit=(PA) (PB), line width=0.5mm] {};
\node[inner sep=3pt, draw, dotted, red, fit=(PC) (PD), line width=0.5mm] {};

\draw[<-,line width=1pt, dashed] (PC) -- (PB); 
\draw[<-,line width=1pt, dashed] (PC) -- (PA); 
\draw[<-,line width=1pt, dashed] (PD) -- (PB); 
\draw[<-,line width=1pt, dashed] (PD) -- (PA); 

\end{tikzpicture}
}
    \caption{Illustration of a two-phase communication protocol for the \HPC problem.}
    \label{fig:hpc-phase}
\end{figure}
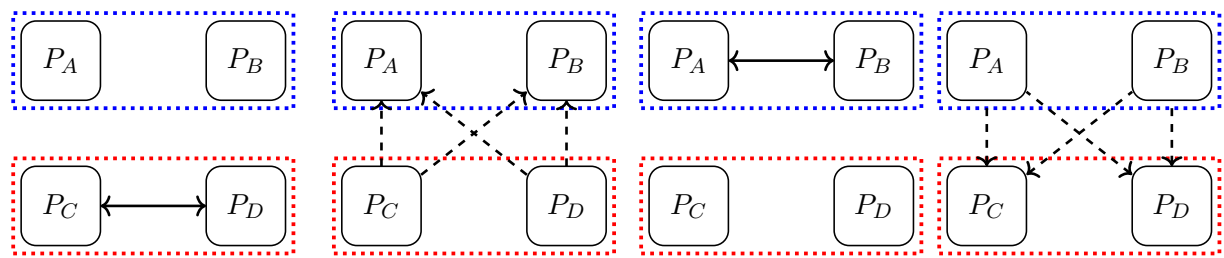

\clearpage

\section{A Schematic Organization of Proof of Lemma~\ref{lem:PI-ic}}\label{app:schematic}

\textcolor{white}{We have}

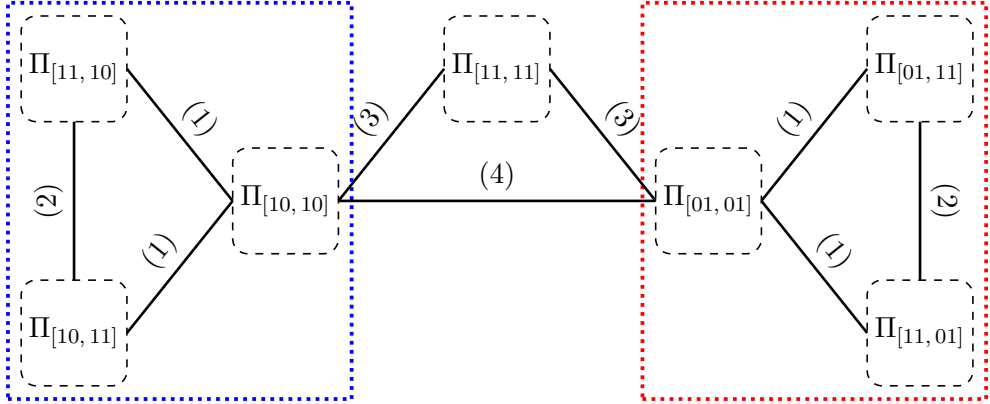
\begin{figure}[h!]
    \centering
    \begin{tikzpicture}[ auto ,node distance =1cm and 2cm , on grid , semithick , state/.style ={ circle ,top color =white , bottom color = white , draw, black , text=black}, every node/.style={inner sep=0,outer sep=0}]

\node[state,rectangle, top color=white, black, bottom color=white, text=black,  dashed, minimum height=40pt, minimum width = 40pt, rounded corners=2mm]  (P1010){$\Prot_{\disjin{1}{0}{1}{0}}$};
\node[state,rectangle, top color=white, black, bottom color=white, text=black,  dashed, minimum height=40pt, minimum width = 40pt, rounded corners=2mm]  (P1110)[above left=50pt and 80pt of P1010]{$\Prot_{\disjin{1}{1}{1}{0}}$};
\node[state,rectangle, top color=white, black, bottom color=white, text=black,   dashed, minimum height=40pt, minimum width = 40pt, rounded corners=2mm]  (P1011)[below left=50pt and 80pt of P1010]{$\Prot_{\disjin{1}{0}{1}{1}}$};

\node[state,rectangle, top color=white, black, bottom color=white, text=black,   dashed, minimum height=40pt, minimum width = 40pt, rounded corners=2mm]  (P1111)[above right=50pt and 80pt of P1010]{$\Prot_{\disjin{1}{1}{1}{1}}$};

\node[state,rectangle, top color=white, black, bottom color=white, text=black,  dashed, minimum height=40pt, minimum width = 40pt, rounded corners=2mm]  (P0101)[right=160pt of P1010]{$\Prot_{\disjin{0}{1}{0}{1}}$};
\node[state,rectangle, top color=white, black, bottom color=white, text=black,  dashed, minimum height=40pt, minimum width = 40pt, rounded corners=2mm]  (P0111)[above right=50pt and 80pt of P0101]{$\Prot_{\disjin{0}{1}{1}{1}}$};
\node[state,rectangle, top color=white, black, bottom color=white, text=black,   dashed, minimum height=40pt, minimum width = 40pt, rounded corners=2mm]  (P1101)[below right=50pt and 80pt of P0101]{$\Prot_{\disjin{1}{1}{0}{1}}$};

\node[inner sep=5pt, draw, dotted, blue, fit=(P1010) (P1110) (P1011), line width=0.5mm] (K1){};
\node[opacity=0, text opacity=1] (K1n)[above=90pt of K1]{target element = 1};

\node[inner sep=5pt, draw, dotted, red, fit=(P0101) (P0111) (P1101), line width=0.5mm] (K2){};
\node[opacity=0, text opacity=1] (K2n)[above=90pt of K2]{target element = 2};

\path[every node/.style={sloped,anchor=south,auto=false}, draw, line width=1pt]
        (P1011.east) edge              node {(1)} (P1010.west)
        (P1110.east) edge              node {(1)} (P1010.west)
        (P1101.west) edge             node {(1)} (P0101.east)
        (P0111.west) edge             node {(1)} (P0101.east)
        
        (P1011) edge              node {(2)}  (P1110)
        (P0111) edge              node {(2)}  (P1101)
        
        (P1010.east) edge node{(3)} (P1111.west)
        (P0101.west) edge node{(3)} (P1111.east)
        
           (P1010.east) edge node{(4)} (P0101.west);

\end{tikzpicture}

    \caption{Organization of the proof of Lemma~\ref{lem:PI-ic}. Each box denotes the transcript of the protocol for a specific input to players. The boxes in the left are for inputs with target element $k=1$, while the ones on the right are for $k=2$. 
    The middle box is the transcript obtained by running the protocol on $\disjin{1}{1}{1}{1}$ which is \emph{not} a valid input to $\PI$. The strategy in the proof is to show that distribution of all these transcript are close to each other. Each edge
    between two boxes shows the step for establishing the distance between the distribution of the transcripts on its endpoints. The steps are as follows: 
    \newline
    \newline
	\textcolor{white}{1} \hspace{0.25cm} Step (1): Follows from the contradicting assumption on the information revealed by the protocol (in Claim~\ref{clm:disj2-info-terms} and Claim~\ref{clm:disj2-h-terms-1}). 	
    \newline
    \newline
	\textcolor{white}{1} \hspace{0.5cm} Step (2):  ~~Follows from the triangle inequality between the distances (in Claim~\ref{clm:disj2-h-terms-2}). 
    \newline
    \newline
	\textcolor{white}{1} \hspace{0.35cm} Step (3):  Follows from the cut-and-paste property (Fact~\ref{fact:r-rectangle}), applied to the two left most boxes and the two right most ones, respectively.  
    \newline
    \newline
    	\textcolor{white}{1} \hspace{0.35cm} Step (4):  Follows from the cut-and-paste property (Fact~\ref{fact:r-rectangle}), applied to the two left most boxes and the two right most ones, respectively (in Claim~\ref{clm:disj2-h-terms-2}).  
\newline
\newline
      The proof then is finalized by applying the triangle inequality to all pairs of boxes with no edge in the figure~(in Claim~\ref{clm:disj2-sep-terms}). At this point, we obtain that the transcript of the protocol is essentially distributed the
      same regardless of the input, hence
      the protocol cannot possibly distinguish between the cases when target element is $1$ versus the ones when it is $2$ with a non-negligible advantage over random guessing. 
    }
    \label{fig:pi-ic}
\end{figure}


\end{document}